\documentclass{article}

\usepackage{amsmath}       
\usepackage{amsthm}        
\usepackage{amssymb}       
\usepackage{braket}       
\usepackage{amsfonts}
\usepackage[all]{xy}
\usepackage{xspace}
\usepackage{calc}
\usepackage{comment}
\usepackage{pgfplots}
\usepgflibrary{fpu}
\usepackage{mathtools}
\usepackage{tabularx}
\usepackage{ifthen}             
\usepackage{graphicx}
\usepackage{docmute}
\usepackage{array}
\usepackage{subfloat} 
\usepackage{textcomp}
\usepackage{bbm}            
\usepackage{etoolbox}  
\usepackage{verbatim}
\usepackage{stmaryrd}
\usepackage{cancel}
\usepackage{xcolor}
\colorlet{Black}{black}
\usepackage{tensor}
\usepackage{enumerate}  
\usepackage{thm-restate}
\usepackage[numbers,sort]{natbib}
\usepackage{enumitem}
\usepackage{tocloft}
\usepackage{pdfpages}
\usepackage[export]{adjustbox}

\usepackage[margin=3cm]{geometry}

\newenvironment{tz}[1][]{%
                                \begin{tikzpicture}[baseline={([yshift=-.8ex]current bounding                        box.center)},#1] %
                                }{%
                        \end{tikzpicture} %
                        }

\newcommand\ignore[1]{}

\DeclareRobustCommand{\SkipTocEntry}[5]{}

\usepackage{tikz}
\usetikzlibrary{matrix}
\usetikzlibrary{decorations.pathreplacing}
\usetikzlibrary{decorations.markings}
\usetikzlibrary{arrows}
\usetikzlibrary{calc}
\usetikzlibrary{shapes.misc}
\usetikzlibrary{fit}
\usepgflibrary{decorations.pathmorphing}
\usepgflibrary{shapes.geometric}

\pgfdeclarelayer{edgelayer}
\pgfdeclarelayer{nodelayer}
\pgfdeclarelayer{foreground}
\pgfsetlayers{edgelayer,nodelayer,main,foreground}

\tikzstyle{none}=[inner sep=0pt]

\tikzstyle{rn}=[circle,fill=Red,draw=Black,line width=0.8 pt]
\tikzstyle{gn}=[circle,fill=Lime,draw=Black,line width=0.8 pt]
\tikzstyle{bl}=[circle,fill=Blue,draw=Black,line width=0.8 pt]

\tikzstyle{simple}=[-,draw=Black,thick]
\tikzstyle{arrow}=[-,draw=Black,postaction={decorate},decoration={markings,mark=at position .5 with {\arrow{>}}},thick]
\tikzstyle{tick}=[-,draw=Black,postaction={decorate},decoration={markings,mark=at position .5 with {\draw (0,-0.1) -- (0,0.1);}},line width=2.000]
\makeatother\def\thickness{0.7pt}
\tikzstyle{dot}=[circle, draw=black, fill=black, inner sep=.5ex, line width=\thickness, node on layer=foreground]

\makeatletter
\pgfkeys{%
  /tikz/on layer/.code={
    \pgfonlayer{#1}\begingroup
    \aftergroup\endpgfonlayer
    \aftergroup\endgroup
  },
  /tikz/node on layer/.code={
     \gdef\node@@on@layer{%
      \setbox\tikz@tempbox=\hbox\bgroup\pgfonlayer{#1}\unhbox\tikz@tempbox\endpgfonlayer\egroup}
    \aftergroup\node@on@layer
  },
  /tikz/end node on layer/.code={
    \endpgfonlayer\endgroup\endgroup
  }
}

\def\node@on@layer{\aftergroup\node@@on@layer}

\makeatletter
\def\calign@preamble{%
   &\hfil\strut@
    \setboxz@h{\@lign$\m@th\displaystyle{##}$}%
    \ifmeasuring@\savefieldlength@\fi
    \set@field
    \hfil
    \tabskip\alignsep@
}
\let\cmeasure@\measure@
\patchcmd\cmeasure@{\divide\@tempcntb\tw@}{}{}{}
\patchcmd\cmeasure@{\divide\@tempcntb\tw@}{}{}{}
\patchcmd\cmeasure@{\ifodd\maxfields@
  \global\advance\maxfields@\@ne
  \fi}{}{}{}    
\newenvironment{calign}
{%
  \let\align@preamble\calign@preamble
  \let\measure@\cmeasure@
  \align
}
{%
  \endalign
}  
\makeatother
\tikzset{
    master/.style={
        execute at end picture={
            \coordinate (lower right) at (current bounding box.south east);
            \coordinate (upper left) at (current bounding box.north west);
        }
    },
    slave/.style={
        execute at end picture={
            \pgfresetboundingbox
            \path (upper left) rectangle (lower right);
        }
    }
}
\tikzset{blob/.style={draw, circle, fill=white, inner sep=1pt, minimum width=15pt, font=\scriptsize, line width=0.7pt}}
\tikzset{greenregion/.style={fill=green, fill opacity=0.3, draw=none}}
\tikzset{redregion/.style={fill=red, fill opacity=0.3, draw=none}}
\tikzset{blueregion/.style={fill=blue, fill opacity=0.3, draw=none}}
\tikzset{yellowregion/.style={fill=yellow, fill opacity=0.5, draw=none}}
\tikzset{cyanregion/.style={fill=cyan, fill opacity=0.3, draw=none}}
\tikzset{orangeregion/.style={fill=orange, fill opacity=0.6, draw=none}}
\tikzset{solidgreenregion/.style={fill=green!30, fill opacity=1, draw=none}}
\tikzset{solidredregion/.style={fill=red!30, fill opacity=1, draw=none}}
\tikzset{solidblueregion/.style={fill=blue!30, fill opacity=1, draw=none}}
\tikzset{solidyellowregion/.style={fill=yellow!30, fill opacity=1, draw=none}}
\tikzset{string/.style={line width=0.7pt}}
\tikzset{zig/.style={decoration={zigzag,segment length=3, amplitude=0.5}}}
\tikzset{bnd/.style={draw,string}}   %alternative: draw, string 
\tikzset{projector/.style={circle, draw, font=\scriptsize, inner sep=-5pt, minimum width=0.35cm, string, fill=white}}
\tikzset{dimension/.style={font=\scriptsize, inner sep=1pt}}

\tikzset{arrow data/.style 2 args={
      decoration={
         markings,
         mark=at position #1 with \arrow{#2}},
         postaction=decorate}
}

\tikzset{along path/.style={every path/.style={}, sloped, allow upside down}}

\def\zxnormal {
                \def \zxscale{0.55}
                \def\zxnodescale{0.8}
                \def\vertexscale{0.7}
                \def\zxshift{0.075cm}
                \def\hadscale{0.8}
                \def\trianglescale{1}
                \def\boxscale{1}
                }

\def\zxgreen{white}

\def\zxwhite{white}

\tikzset{front/.style ={node on layer=foreground}}
\tikzset{zx/.style = {string, scale=\zxscale}}
\tikzset{zxnode/.style n args={1}{blob,scale=\zxnodescale,fill=#1,node on layer=foreground}}
\tikzset{box/.style={draw, rectangle, fill=white, inner sep=1pt, minimum width=10pt,minimum height=10pt, font=\scriptsize, line width=0.7pt,scale=\zxnodescale,node on layer=foreground}}
\tikzset{boxvertex/.style={draw, rectangle, fill=white, line width=0.733pt,scale=0.75*\vertexscale}}
\tikzset{bigbox/.style={draw, rectangle, fill=white,  minimum width=\boxscale *18pt,minimum height=\boxscale*8pt, line width=0.7pt,scale=\zxnodescale}}
\newlength{\unitbox}
\tikzset{widebox/.style ={draw,rectangle, fill=white, line width=0.7pt,scale=0.75*\zxnodescale,minimum height=15pt,inner sep=1pt,  minimum width = \unitbox,   anchor=center }}
\tikzset{wideboxm/.style n args={1}{draw,rectangle, fill=white, line width=0.7pt,scale=0.75*\zxnodescale,minimum height=15pt,inner sep=1pt,  minimum width =2\unitbox+#1\unitbox,   anchor=center }}
\tikzset{triangleup/.style n args={1}{draw, shape=isosceles triangle, isosceles triangle stretches, fill=white, line width=0.7pt,scale=0.75*\zxnodescale,minimum height=15pt,inner sep=1pt,  minimum width = #1*\trianglescale cm +0.15*\trianglescale cm,  shape border rotate=90, anchor=south }}
\tikzset{triangledown/.style n args={1}{draw, shape=isosceles triangle, isosceles triangle stretches, fill=white, line width=0.7pt,scale=0.75*\zxnodescale,minimum height=15pt,inner sep=1pt,  minimum width = #1*\trianglescale cm +0.15*\trianglescale cm,  shape border rotate=-90, anchor=north }}
\tikzset{zxvertex/.style n args={1}{draw,fill=#1,circle,line width=0.7pt,scale=0.75*\vertexscale}}
\tikzset{zxdown/.style={yshift=-\zxshift}}
\tikzset{zxup/.style={yshift=\zxshift}}

\newcommand\mult[3]{ %mult{base}{width}{height}  
\draw[string] (#1.center) to [out=up, in=-135] +(0.5*#2,#3) to [out=-45, in=up] +(0.5*#2,-#3);
\node[zxvertex=\zxgreen,zxdown] at ($(#1)+(0.5*#2,#3)$){};
}

\newcommand\unit[2]{ %mult{base}{height}  
\draw[string] (#1.center) to + (0, -#2);
\node[zxvertex=\zxgreen] at ($(#1) +(0,-#2)$){};
}

\newcommand{\Tr}{\mathrm{Tr}}

\renewcommand{\to}[1][]{\ensuremath{\xrightarrow{#1}}}

\allowdisplaybreaks[1]

\theoremstyle{plain} %%% Plain Theorem Styles.

\newtheorem{theorem}{Theorem}[section]
\newtheorem{lemma}[theorem]{Lemma}
\newtheorem{corollary}[theorem]{Corollary}          
\newtheorem{proposition}[theorem]{Proposition}

\newtheorem*{theorem*}{Theorem}

\theoremstyle{definition} %%%% Definition-like Commands  
\newtheorem{definition}[theorem]{Definition}
\newtheorem{remark}[theorem]{Remark}

\newtheorem{notation}[theorem]{Notation}

\theoremstyle{remark}  %%%% Remark-like Commands
\newtheorem{example}[theorem]{Example}

\newtheoremstyle{special_statement} 
        {\topskip}% Space above
        {\topskip}% Space below
        {\addtolength{\leftskip}{2.5em} \itshape }% Body font
        {}% Indent amount % \parindent
        {\bfseries}% Theorem head font
        {:}% Punctuation after theorem head
        {.5em}% Space after theorem head
        {}% Theorem head spec (can be left empty, meaning `normal')
\theoremstyle{special_statement}

\DeclareMathOperator{\Hom}{Hom}
\DeclareMathOperator{\End}{End}

\newcommand{\id}{\mathrm{id}}

\newcommand{\Bimod}{\mathrm{Bimod}}

\newcommand{\Aut}{\ensuremath{\mathrm{Aut}}}

\newcommand{\Rep}{\mathrm{Rep}}
\newcommand{\Mat}{\mathrm{Mat}}

\newcommand{\Mod}{\mathrm{Mod}}
\newcommand{\diag}{\mathrm{diag}}

\newcommand{\Hilb}{\ensuremath{\mathrm{Hilb}}}

\newcommand{\op}{\text{op}}

\newcommand{\F}{\ensuremath{\mathrm{SSFA}}}

\makeatletter
\DeclareFontFamily{OMX}{MnSymbolE}{}
\DeclareSymbolFont{MnLargeSymbols}{OMX}{MnSymbolE}{m}{n}
\SetSymbolFont{MnLargeSymbols}{bold}{OMX}{MnSymbolE}{b}{n}
\DeclareFontShape{OMX}{MnSymbolE}{m}{n}{
    <-6>  MnSymbolE5
   <6-7>  MnSymbolE6
   <7-8>  MnSymbolE7
   <8-9>  MnSymbolE8
   <9-10> MnSymbolE9
  <10-12> MnSymbolE10
  <12->   MnSymbolE12
}{}
\DeclareFontShape{OMX}{MnSymbolE}{b}{n}{
    <-6>  MnSymbolE-Bold5
   <6-7>  MnSymbolE-Bold6
   <7-8>  MnSymbolE-Bold7
   <8-9>  MnSymbolE-Bold8
   <9-10> MnSymbolE-Bold9
  <10-12> MnSymbolE-Bold10
  <12->   MnSymbolE-Bold12
}{}

\let\llangle\@undefined
\let\rrangle\@undefined
\DeclareMathDelimiter{\llangle}{\mathopen}%
                     {MnLargeSymbols}{'164}{MnLargeSymbols}{'164}
\DeclareMathDelimiter{\rrangle}{\mathclose}%
                     {MnLargeSymbols}{'171}{MnLargeSymbols}{'171}
\makeatother

\RequirePackage{doi}
\usepackage{hyperref}

\usepackage[utf8]{inputenc}
\usepackage{diagrams}

\title{A covariant Stinespring theorem}
\author{Dominic Verdon\footnote{\href{mailto:dominic.verdon@bristol.ac.uk}{dominic.verdon@bristol.ac.uk}}}
\date{School of Mathematics, University of Bristol, United Kingdom \\[2ex]
\today}

\begin{document}
\normalsize
\zxnormal
\maketitle 

\begin{abstract}
We prove a finite-dimensional covariant Stinespring theorem for compact quantum groups. Let $G$ be a compact quantum group, and let $\mathcal{T}:= \Rep(G)$ be the rigid $C^*$-tensor category of finite-dimensional continuous unitary representations of $G$. Let $\Mod(\mathcal{T})$ be the rigid $C^*$-2-category of cofinite semisimple finitely decomposable $\mathcal{T}$-module categories. We show that finite-dimensional $G$-$C^*$-algebras can be identified with equivalence classes of 1-morphisms out of the object $\mathcal{T}$ in $\Mod(\mathcal{T})$. For 1-morphisms $X: \mathcal{T} \to \mathcal{M}_1$, $Y: \mathcal{T} \to \mathcal{M}_2$, we show that covariant completely positive maps between the corresponding $G$-$C^*$-algebras can be `dilated' to isometries $\tau: X \to Y \otimes E$, where $E: \mathcal{M}_2 \to \mathcal{M}_1$ is some `environment' 1-morphism. Dilations are unique up to partial isometry on the environment; in particular, the dilation minimising the quantum dimension of the environment is unique up to a unitary. When $G$ is a compact group this recovers previous covariant Stinespring-type theorems.
\end{abstract}

\setcounter{tocdepth}{2}
\tableofcontents

\section{Introduction}

\addtocontents{toc}{\protect\setcounter{tocdepth}{0}}

\paragraph{Stinespring theorems.}

In finite-dimensional (f.d.) quantum and classical physics we identify systems with f.d.~$C^*$-algebras and dynamics with completely positive trace-preserving (CPTP) linear maps (called \emph{channels}). To formulate an expressive physical theory it is useful to place symmetry restrictions on dynamics; one therefore introduces a compact group $G$ and identifies systems with f.d.~$C^*$-algebras with a $G$-action (called \emph{$G$-$C^*$-algebras} or \emph{$C^*$-dynamical systems}), and dynamics with \emph{covariant} channels intertwining the $G$-actions. Recently these notions have been generalised to compact quantum groups~\cite{Wang1998,Soltan2010,DeCommer2017}.

An essential tool in the study of finite-dimensional quantum physics, particularly quantum information theory, is the finite-dimensional Stinespring theorem~\cite[Thm. 2]{Holevo2007}~\cite[Thm. 1]{Stinespring1955}, which characterises channels between f.d.~$C^*$-algebras $A,B$. In general (a special case of~\cite[Thm. 15]{Szafraniec2010}), the theorem implies that for any completely positive linear map $f: A \to B$ there exists an f.d. right Hilbert $B$-module $\mathcal{E}$, a multiplicative $*$-homomorphism $\Phi: A \to B^*(\mathcal{E})$ (where $B^*(\mathcal{E})$ are the adjointable operators on $\mathcal{E}$) and an adjointable  $B$-module map $V: B \to \mathcal{E}$ such that $f(x) = V^{\dagger} \Phi(x) V$. The completely positive map $f^{\dagger}$ is trace-preserving if and only if $V$ is an isometry. This reduces the study of dynamics between indecomposable f.d.~$C^*$-algebras to the study of isometric maps between Hilbert modules.

A covariant version of the Stinespring theorem, applicable to covariant channels between $G$-$C^*$-algebras for a compact group $G$, has also appeared~\cite[Thm. 1]{Scutaru1979}\cite[Thm. 2.1]{Paulsen1982}. The statement is similar to the non-covariant case, except that $\mathcal{E}$ is now a $G$-equivariant Hilbert module, and $V$ is an intertwiner of representations. This result reduces dynamics between between f.d. matrix $G$-$C^*$-algebras to isometric intertwiners between equivariant Hilbert modules. 

In this work we prove a covariant Stinespring theorem which extends to the case where $G$ is any compact quantum group and holds for maps between any pair of $G$-$C^*$-algebras. We also show uniqueness of the dilation up to a partial isometry on the environment. (In fact, our results are in fact somewhat more general than this --- the theory works in any rigid $C^*$-tensor category, not just in the category of f.d. continuous unitary representations of a compact quantum group.)

Our results can be interpreted as showing that the theory of finite-dimensional $G$-$C^*$-algebras and completely positive maps, which is formulated in the rigid $C^*$-tensor category $\mathcal{T} := \Rep(G)$, is the Morita-theoretical `shadow' of a theory whose dynamics are given by isometric 2-morphisms in the \emph{semisimple $C^*$-2-category} $\Mod(\mathcal{T})$ (which is equivalent to the 2-category of f.d. $G$-$C^*$-algebras, finitely generated $G$-equivariant Hilbert bimodules, and equivariant bimodule morphisms). From a utilitarian standpoint, one can use the 2-category $\Mod(\mathcal{T})$ to construct, study and manipulate covariant channels. More foundationally, however, the result suggests that the more fundamental theory may be one which identifies systems with 1-morphisms and dynamics with isometric 2-morphisms in $\Mod(\mathcal{T})$. In the non-covariant case, such a theory has already been proposed~\cite{Vicary2012,Vicary2012a}\cite[Chap. 8]{Heunen2019}; there, objects were identified with classical information, 1-morphisms with classically controlled quantum systems, and isometric 2-morphisms with classically controlled dynamics. In the covariant case, there are new phenomena (such as inequivalent simple objects) which require interpretation. We here limit ourselves to the presentation of the theorem, leaving questions of interpretation for future work. 

\paragraph{A categorical formulation of covariant physics.}
We now explain what we mean by Morita theory. For any \emph{simple object} $r$ of a semisimple $C^*$-2-category $\mathcal{C}$, the endomorphism category $\End(r)$ is a \emph{rigid $C^*$-tensor category}. Morita theory describes how objects, 1-morphisms and 2-morphisms in $\mathcal{C}$ appear as certain algebraic structures in $\End(r)$.

The category $\Rep(G)$ of f.d. continuous unitary representations of a compact quantum group $G$ is a rigid $C^*$-tensor category. (In fact, every rigid $C^*$-tensor category $\mathcal{T}$ with a faithful unitary $\mathbb{C}$-linear tensor functor $\mathcal{T} \to \Hilb$ (called a \emph{fibre functor}), where $\Hilb$ is the category of finite-dimensional Hilbert spaces and linear maps, is equivalent to $\Rep(G)$ for some $G$.)
The theory of $G$-$C^*$-algebras and covariant channels admits a natural formulation in terms of algebraic structures in the category $\Rep(G)$. In brief, every finite-dimensional $G$-$C^*$-algebra $A$ possesses a canonical $G$-invariant functional $\phi: A \to \mathbb{C}$, such that, with the inner product $\braket{x|y} = \phi(x^*y)$, $A$ becomes a f.d. continuous unitary $G$-representation carrying an algebra structure. This gives a correspondence between $G$-$C^*$-algebras and \emph{separable standard Frobenius algebras} ($\F$s) $A,B,\dots$ in $\Rep(G)$. Likewise, CP maps between $G$-$C^*$-algebras correspond to \emph{CP morphisms} $A \to B$; channels are CP morphisms preserving the \emph{counit}. (See Section~\ref{sec:gc*alg} for a more detailed summary and references.) More generally, one may also consider $\F$s and CP morphisms in rigid $C^*$-tensor categories $\mathcal{T}$ not possessing a fibre functor, although in this case there is no obvious way to identify the $\F$s and CP morphisms with concrete $C^*$-algebras and linear maps. The theory of covariant finite-dimensional physics can therefore be identified with the theory of $\F$s and CP morphisms in a rigid $C^*$-tensor category $\mathcal{T}$.

For these $\F$s and CP morphisms to arise by Morita theory, we need to define a rigid $C^*$-2-category in which $\mathcal{T}$ embeds as an endomorphism category. Just as we `unpacked' a compact quantum group $G$ to obtain a rigid $C^*$-tensor category $\Rep(G)$ with fibre functor in which our physical theory is formulated, we want to `unpack' $\Rep(G)$ further to obtain a rigid $C^*$-2-category. We here consider two ways to `unpack' a rigid $C^*$-tensor category $\mathcal{T}$, which are $C^*$-adaptations of well-known constructions. One is the (strict) 2-category $\Mod(\mathcal{T})$, whose objects are semisimple cofinite finitely decomposable left $\mathcal{T}$-module categories, whose 1-morphisms are unitary $\mathcal{T}$-module functors, and whose 2-morphisms are morphisms of $\mathcal{T}$-module functors --- in this 2-category, $\mathcal{T}$ embeds as the endomorphism category $\End_{T}(\mathcal{T})$ of $\mathcal{T}$ considered as a $\mathcal{T}$-module category. The other is the 2-category $\Bimod(\mathcal{T})$, whose objects are $\F$s in $\mathcal{T}$, whose 1-morphisms are dagger bimodules, and whose 2-morphisms are bimodule homomorphisms --- in this 2-category, $\mathcal{T}$ embeds as the endomorphism category $\End(\mathbbm{1})$ of the trivial $\F$.

\paragraph{Semisimple $C^*$-2-categories.} In fact, $\Bimod(\mathcal{T})$ and $\Mod(\mathcal{T})$ are equivalent semisimple $C^*$-2-categories.

Following~\cite{Douglas} we say that a rigid $C^*$-2-category is furthermore \emph{presemisimple} if it is locally semisimple and additive, and every object decomposes as a finite direct sum of simple objects. The missing ingredient for semisimplicity is splitting of dagger idempotents at the 1-morphism level, which has no parallel in the theory of rigid $C^*$-tensor categories. We therefore need to propose a definition of a dagger idempotent 1-morphism. In~\cite{Douglas}, which treated the non-unitary case, idempotent 1-morphisms were defined as separable algebras in endomorphism categories. In the unitary $C^*$-setting, we do not want to work with all separable algebras, and so need to tighten this definition. We propose that the relevant idempotents in the $C^*$ setting are $\F$s in endomorphism categories. This is motivated physically by the fact that these idempotents can be identified with $G$-$C^*$-algebras.

Our definition of splitting of such an idempotent is Morita-theoretical in nature. It is well-known (e.g.~\cite{Lauda2005}) that, in a 2-category with duals, every 1-morphism $X: r \to s$ out of $r$ induces a Frobenius algebra (the `pair of pants' algebra) on the object $X \otimes X^*$ of the endomorphism category $\End(r)$. In a presemisimple $C^*$-2-category, this construction can be normalised (Proposition~\ref{prop:pairofpants}) to produce $\F$s in the rigid $C^*$-multitensor category $\End(r)$ from \emph{separable} 1-morphisms out of $r$. (Separability of a 1-morphism is a sort of nondegeneracy condition (Definition~\ref{def:separable1morph}).) We say that an $\F$ $A$ in the endomorphism category $\End(r)$ \emph{splits} if it is isomorphic to the pair of pants $X \otimes X^*$ for some separable 1-morphism $X: r \to s$. We say that a rigid $C^*$-tensor category is \emph{semisimple} if it is presemisimple and all $\F$s in all endomorphism categories split.

As we already mentioned, we show that $\Bimod(\mathcal{T})$ and $\Mod(\mathcal{T})$ are semisimple $C^*$-2-categories (Proposition~\ref{prop:bimodsemisimple}, Corollary~\ref{cor:modtsemisimpleetc.}), and that there is an equivalence $\Bimod(\mathcal{T}) \simeq \Mod(\mathcal{T})$ (Theorem~\ref{thm:eilenbergwatts}). These 2-categories can be seen as higher idempotent completions of the rigid $C^*$-tensor category $\mathcal{T}$. We say that a semisimple $C^*$-2-category is \emph{connected} if the Hom-category between any pair of nonzero objects is nonzero; we observe that every connected semisimple $C^*$-2-category $\mathcal{C}$ is equivalent to $\Mod(\mathcal{T})$ for some rigid $C^*$-tensor category $\mathcal{T}$~(Proposition~\ref{prop:reconstruction}). (Here $\mathcal{T}$ can be chosen as the endomorphism category of any simple object in $\mathcal{C}$.)

\paragraph{A classification of $G$-$C^*$-algebras.}

These results already yield a Morita-theoretical characterisation of f.d. $G$-$C^*$-algebras; or, more generally, of $\F$s in a rigid $C^*$-tensor category $\mathcal{T}$. Indeed, by semisimplicity of $\Mod(\mathcal{T})$, every $F$ in $\mathcal{T}$ arises as a pair of pants algebra $X \otimes X^*$ for some separable 1-morphism $X: \mathcal{T} \to \mathcal{M}$ in $\Mod(\mathcal{T})$. 

This results in a classification that has appeared elsewhere~\cite{DeCommer2012,Neshveyev2013a,Neshveyev2018}, although we have not found the notion of equivalence up to a phase in $\End_{\mathcal{T}}(\mathcal{M})$ in other works. We say that an $\F$ in $\mathcal{T}$ is \emph{simple} if it cannot be decomposed as a nontrivial direct sum. We say that two $\F$s in $\mathcal{T}$ are \emph{Morita equivalent} if they are equivalent as objects of $\Bimod(\mathcal{T})$. Let $\mathcal{M}$ be a cofinite semisimple indecomposable right $\tilde{\mathcal{T}}$-module category, for some rigid $C^*$-tensor category $\tilde{\mathcal{T}}$; we say that two objects $X_1,X_2$ of $\mathcal{M}$ are \emph{equivalent up to a phase in $\tilde{\mathcal{T}}$} if there is an object $\theta$ of $\tilde{\mathcal{T}}$ with unit dimension such that $X_1$ is unitarily isomorphic to $X_2 \tilde{\otimes} \theta$. We then have the following theorem:
\begin{theorem*}[Theorem \ref{thm:classificationoffs}]
Let $\mathcal{T}$ be a rigid $C^*$-tensor category. There is a bijective correspondence between:
\begin{itemize}
\item Morita equivalence classes of simple $\F$s in $\mathcal{T}$. 
\item Equivalence classes of cofinite semisimple indecomposable left $\mathcal{T}$-module categories. 
\end{itemize}
Let $\mathcal{M}$ be a cofinite semisimple indecomposable left $\mathcal{T}$-module category. Since $\mathcal{M}$ is indecomposable, the category $\End_{\mathcal{T}}(\mathcal{M})$ of $\mathcal{T}$-module endofunctors on $\mathcal{M}$ is a rigid $C^*$-tensor category with a right action on $\mathcal{M}$. There is a bijective correspondence between:
\begin{itemize}
\item Isomorphism classes of simple $\F$s in the corresponding Morita class.
\item Isomorphism classes of objects in $\mathcal{M}$, up to a phase in $\End_{\mathcal{T}}(\mathcal{M})$.
\end{itemize}
\end{theorem*}
\noindent
In particular, the \emph{connected} (a.k.a \emph{ergodic}) $G$-$C^*$-algebras of e.g.~\cite{Bichon2005,DeCommer2012,Arano2015} are those arising from simple objects of $\mathcal{M}$ (Proposition~\ref{prop:ergodic}).

\paragraph{The covariant Stinespring theorem.}

We then consider dynamics; that is, CP morphisms and channels between $\F$s in $\mathcal{T}$. We embed $\mathcal{T}$ as the endomorphism category $\End_{\mathcal{T}}(\mathcal{T})$ in $\Mod(\mathcal{T})$. By semisimplicity, every $\F$  in $\mathcal{T}$ is isomorphic to a pair of pants algebra $X \otimes X^*$ for a separable 1-morphism $X: \mathcal{T} \to \mathcal{M}$. In this context, we now state our main theorem. The equations below use the diagrammatic calculus for pivotal dagger 2-categories (Section~\ref{sec:pivdagcats}); the main point is that $f$ can be expressed entirely in terms of the dilation $\tau$ and the rigid structure of the 2-category $\Mod(\mathcal{T})$.
\begin{theorem*}
[Theorem~\ref{thm:covstinespring}]
Let $X: \mathcal{T} \to \mathcal{M}_1$, $Y: \mathcal{T} \to \mathcal{M}_2$ be separable 1-morphisms in $\Mod(\mathcal{T})$, and let $f: X \otimes X^* \to Y \otimes Y^*$ be a CP morphism between the corresponding $\F$s in $\End_{\mathcal{T}}(\mathcal{T}) \simeq \mathcal{T}$. 

Then there exists a 1-morphism $E: \mathcal{M}_2 \to \mathcal{M}_1$ (the `environment') and a 2-morphism $\tau: X \to Y \otimes E$ such that the following equation holds:
\begin{calign}\label{eq:stinespringintro}
\includegraphics[scale=.8]{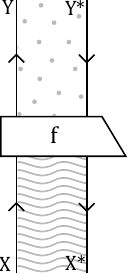}
~~=~~
\includegraphics[scale=.8]{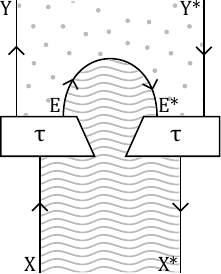}
\end{calign}
We say that $\tau$ is a \emph{dilation} of $f$. The morphism 
\begin{calign}\label{eq:ssisometrycondintro}
\includegraphics[scale=.8]{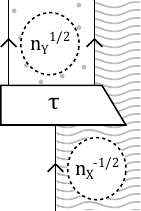}
\end{calign}
is an isometry if and only if $f$ is a channel. (Here $n_X$ and $n_Y$ are normalising factors associated with the canonical trace (Definition~\ref{def:separable1morph}).)

In the other direction, for any 1-morphism $E: \mathcal{M}_2 \to \mathcal{M}_1$ and 2-morphism $\tau: X \to Y \otimes E$, the morphism $f: X \otimes X^* \to Y \otimes Y^*$ defined by~\eqref{eq:stinespringintro} is CP, and a channel if and only if~\eqref{eq:ssisometrycondintro} is an isometry.

Different dilations for a CP morphism $f: X \otimes X^* \to Y \otimes Y^*$ are related by a partial isometry on the environment. Specifically, let $\tau_1: X \to Y \otimes E_1$, $\tau_2: X \to Y \otimes E_2$ be two dilations of $f$. Then there exists a partial isometry $\alpha: E_1 \to E_2$ such that 
\begin{align*}
(\id_Y \otimes \alpha) \circ \tau_1 = \tau_2 
&&
(\id_Y \otimes \alpha^{\dagger}) \circ \tau_2 = \tau_1 
\end{align*}
In particular, the dilation minimising the quantum dimension of the environment $d(E)$ is unique up to unitary $\alpha$. (A concrete construction of the minimal dilation from any other dilation is specified in the last paragraph of the proof.)
\end{theorem*}
\noindent 
This theorem recovers the aforementioned previous results in the literature: the  f.d. noncovariant Stinespring theorem follows from setting $\mathcal{T} := \Hilb$, and the f.d. covariant Stinespring theorem for a compact group $G$ follows from setting $\mathcal{T} := \Rep(G)$ (Example~\ref{ex:vanillastinespring}).

We finish by reiterating the proposal that, rather than identifying systems with $\F$s in $\mathcal{T}$ and dynamics with CP morphisms, we might identify systems with 1-morphisms in $\Mod(\mathcal{T})$ and dynamics with isometries. We can then see how the 2-categorical theory extends the algebraic theory: $\F$s and CP morphisms in $\mathcal{T}$ correspond to 1-morphisms $X: \mathcal{T} \to \mathcal{M}$ in $\Mod(\mathcal{T})$ and isometries of type $X \to Y \otimes E$, whereas the 2-categorical theory encompasses all 1-morphisms, and all isometries. In fact, the 2-categorical theory unites the theories of $\F$s and CP morphisms in all rigid $C^*$-tensor categories categorically Morita equivalent to $\mathcal{T}$, which appear as the endomorphism categories of simple objects in $\Mod(\mathcal{T})$. However, the physical interpretation of this extended theory is still unclear.

\subsection{Related work}

\paragraph{Categorical Morita equivalence of compact quantum groups.} In~\cite{Neshveyev2018} a notion of a Morita-Galois object was introduced. This is a $G_1$-$G_2$-$C^*$-algebra whose category of equivariant right Hilbert $A$-modules is an invertible $\Rep(G_1)$-$\Rep(G_2)$-bimodule category; such an algebra can be reconstructed from any $\Rep(G_1)$-$\Rep(G_2)$-bimodule category, which includes the Hom-category $\Hom(\mathcal{M}_1, \mathcal{M}_2)$ between simple objects in any semisimple $C^*$-category with fibre functor. We also note that, in the language of this work, two compact quantum groups (or more generally two rigid $C^*$-tensor categories $\mathcal{T}_1,\mathcal{T}_2$) are categorically Morita equivalent precisely when there is an equivalence $\Mod(\mathcal{T}_1) \simeq \Mod(\mathcal{T}_2)$.

\paragraph{Previous covariant Stinespring theorems.} As we have mentioned, our result recovers~\cite[Thm. 2]{Holevo2007}, the finite-dimensional cases of~\cite[Thm. 1]{Stinespring1955}\cite[Thm. 1]{Scutaru1979}\cite[Thm. 2.1]{Paulsen1982}, and the special case of~\cite[Thm. 15]{Szafraniec2010} applying to f.d. unital $C^*$-algebras. Since the category $\Mod(\Rep(G))$ is equivalent to the 2-category of equivariant finitely generated Hilbert bimodules over f.d. $G$-$C^*$-algebras~\cite[P.13]{Neshveyev2018}, it may also be interesting to consider recent work on completely positive maps between Hilbert $C^*$-modules~\cite{Asadi2009,Bhat2010,Joita2010}.

\paragraph{Module categories.} The 2-category $\Mod(\mathcal{T})$ defined here is as a semisimple $C^*$-version of the 2-category of exact module categories over a tensor category~\cite[Rem. 7.12.15]{Etingof2016}\cite{Ostrik2003}. 
Our proofs demonstrate that when working in the $C^*$-setting there is no need to use abelian category theory; all one needs is linearity, the $C^*$-axioms and idempotent splitting. We hope this will make module category theory more accessible to researchers in the quantum information community.

\paragraph{Torsion-freeness for rigid $C^*$-tensor categories.} Several works considered torsion-freeness for compact quantum groups~\cite{Meyer2008,Voigt2011} and more generally for rigid $C^*$-tensor categories~\cite{Arano2015}. It follows from our results that a rigid $C^*$-tensor category $\mathcal{T}$ is torsion-free precisely when the 2-category $\Mod(\mathcal{T})$ has a single simple object up to equivalence (which is necessarily $\mathcal{T}$ itself). In this way we can already characterise $\Mod(\mathcal{T})$ for connected compact groups with torsion-free fundamental group~\cite[\S{}7.2]{Meyer2008}, quantum groups $SU_q(2)$ for $q \in (-1,1) \backslash 0$~\cite[Prop. 3.2] {Voigt2011},  free orthogonal quantum groups~\cite[Cor. 7.7]{Voigt2011}, and free unitary quantum groups~\cite[Cor. 2.9]{Arano2015}.

\paragraph{Irreducibly covariant channels and Temperley-Lieb channels.} A number of recent works have studied \emph{irreducibly covariant} channels for compact groups from a quantum information-theoretical perspective, e.g.~\cite{Mozrzymas2017,Siudzinska2018,Nuwair2014}. In our language, these are channels between connected matrix $G$-$C^*$-algebras, i.e. $G$-$C^*$-algebras corresponding to simple $1$-morphisms $X: \mathcal{T} \to \mathcal{T}$ in $\Mod(\mathcal{T})$. 

Recent work has also considered \emph{Temperley-Lieb} channels covariant for actions of the free orthogonal quantum groups $O_F^{+}$~\cite{Brannan2016,Brannan2020}. Since these compact quantum groups are torsion-free~\cite[Cor. 7.7]{Voigt2011}, all simple $O_F^{+}$-$C^*$-algebras are matrix $O_F^{+}$-$C^*$-algebras $X \otimes X^*$ for some $X \in \Rep(G)$. In this setting, it is observed that covariant channels $X \otimes X^* \to Y \otimes Y^*$ may be constructed from isometries $X \to Y \otimes E$ in $\Rep(O_F^{+})$. In fact, the covariant Stinespring theorem proven here implies that every covariant channel may be constructed in this way, and also characterises when two such isometries produce equivalent channels.

\paragraph{Operator algebras and CP maps in rigid $C^*$-tensor categories.}
In~\cite{Jones2017} the authors give a fully general definition of a $C^*$-algebra in a rigid $C^*$-tensor category $\mathcal{T}$. (From the results in~\cite{Vicary2011}, in the finite-dimensional case these probably correspond to $\F$s in $\mathcal{T}$; indeed, for connected $\F$s this was shown in~\cite{Jones2017a}.) Whereas we focus on the finite-dimensional case, where Morita equivalence classes of f.d. $G$-$C^*$-algebras correspond to \emph{cofinite} $\Rep(G)$-module categories, in~\cite{Jones2017} infinite-dimensional $G$-$C^*$-algebras are treated, which corresponds to dropping the  cofiniteness condition on the module categories. A Stinespring theorem is also proposed in~\cite{Jones2017}, but from quite a different perspective. We hope to generalise the results of this paper to the infinite-dimensional setting in future work.

After the completion of this paper we were made aware of the recent work~\cite[\S{}5.3]{Henriques2020}. There a CP map between pair of pants algebras in a rigid $C^*$-2-category is defined by~\eqref{eq:stinespring}; it is shown that this definition matches the definition of a CP map given in~\cite{Jones2017}. The present work relates this definition to the usual Stinespring's theorem for $G$-$C^*$-algebras, as well as proving uniqueness of a dilation up to partial isometry.

\paragraph{Q-system completion for $C^*$ 2-categories.} In two recent works~\cite{Chen2021,Giorgetti2020} a notion of idempotent completion of a $C^*$-2-category has been defined. We discuss the relationship between these definitions and the $\Bimod$ construction defined here in Remark~\ref{rem:idempcompletions}.

\paragraph{Standard duals for rigid $C^*$-2-categories}
In~\cite{Giorgetti2019} a notion of standard duality for rigid $C^*$-2-categories with finite-dimensional centres was introduced; we make use of this here, in the special case of presemisimple $C^*$-2-categories. In this presemisimple case we find that the characterisation of standard duals using the equivalence $\Mat(\mathcal{C}) \simeq \mathcal{C}$ (Remark~\ref{rem:matrixstandardduals}) is useful for calculations. 

\subsection{Summary} In Section~\ref{sec:background} we review necessary background material on 2-category theory, covering 2-categories and their diagrammatic calculus (Section~\ref{sec:2cats}), 2-functors and icons (Section~\ref{sec:psfct}), pivotal dagger 2-categories (Section~\ref{sec:pivdagcats}), rigid $C^*$-2-categories (Section~\ref{sec:linearstructure}) and semisimplicity (Section~\ref{sec:semisimp}). 

In Section~\ref{sec:bimod} we define the 2-category $\Bimod(\mathcal{T})$ for a rigid $C^*$-tensor category $\mathcal{T}$, and prove that it is a semisimple $C^*$-2-category. In Section~\ref{sec:modt} we define the 2-category $\Mod(\mathcal{T})$. In Section~\ref{sec:eilwatts} we show the equivalence $\Bimod(\mathcal{T}) \simeq \Mod(\mathcal{T})$, and observe that every connected semisimple $C^*$-2-category is equivalent to $\Mod(\mathcal{T})$ for some $\mathcal{T}$.

In Section~\ref{sec:stinespring} we classify $G$-$C^*$-algebras (Section~\ref{sec:gc*alg}) and prove the covariant Stinespring theorem (Section~\ref{sec:covstinespring}) and covariant Choi theorem (Section~\ref{sec:choi}). 

\addtocontents{toc}{\protect\setcounter{tocdepth}{2}}

\section{Background on 2-category theory}\label{sec:background}

In this section we will review some definitions and results about 2-categories. 

\subsection{Diagrammatic calculus for 2-categories}\label{sec:2cats}

Sometimes the noun `2-category' is taken to indicate a strict 2-category. In this work, by contrast, when we say `2-category' we mean the general, fully weak notion, which is sometimes called a bicategory. We will explicitly use the adjective `strict' to distinguish strict 2-categories. We assume that the reader is familiar with the definition of a 2-category; see e.g.~\cite[Def. 2.1.3]{Johnson2021}.

The abstract definition of a 2-category is important for checking whether some collection of data does or does not constitute a 2-category. When working concretely inside a given 2-category it is more convenient to use the diagrammatic calculus for 2-categories, which takes account of the coherence implied by the pentagon and triangle equalities.

We assume that the reader is familiar with this calculus, which is summarised in, for example,~\cite{Marsden2014}\cite[\S{}8.1.2]{Heunen2019}. It is a straightforward extension of the graphical calculus for monoidal categories. 
\begin{itemize}
\item The objects $r,s,\dots$ of the 2-category are represented by labelled regions. (To avoid having too many letters in the diagrams, we will often use black-and-white pattern shading rather than letter labels in order to show which regions correspond to which objects.)
\item  The 1-morphisms $X,Y,\dots: r \to s$ are represented by edges, separating the region $r$ on the left from the region $s$ on the right. Identity 1-morphisms are invisible in the diagrammatic calculus.
\item The 2-morphisms $f,g,\dots:X \to Y$ are represented by labelled boxes with an $X$-edge entering from below and a $Y$-edge leaving above. Identity 2-morphisms are invisible in the diagrammatic calculus, as are the components of the associators and L/R unitors and their inverses.
\item Vertical and horizontal composition of 2-morphisms are represented by vertical and horizontal juxtaposition respectively. We represent vertical and horizontal composition by $\circ$ and $\otimes$ respectively.
\end{itemize}
We do not keep track of identity 1-morphisms and bracketing of 1-morphism composites in the diagrammatic calculus. However, 2-categorical coherence implies that we do not need to keep track of this information while performing our calculations. Indeed, once a choice of bracketing and identity 1-morphisms is specified for the source and target of a diagram (we call such a choice a \emph{parenthesis scheme}), the diagram represents a unique 2-morphism, however it is interpreted.
\begin{proposition}[{\cite[Prop. 4.1]{Bartlett2008}}]\label{prop:uniquedef2morph}
After specifying a parenthesis scheme for the source and target 1-morphisms of a 2-morphism diagram, the 2-morphism represented by the diagram is independent of the choice of parentheses, associators and unitors used to interpret the interior of the diagram.
\end{proposition}

\subsection{Diagrammatic calculus for 2-functors}\label{sec:psfct}

Here we use the term 2-functor for a functor between 2-categories whose coherence constraints are isomorphisms; this is sometimes called a \emph{pseudofunctor}. We assume the reader is familiar with this notion; see e.g.~\cite[Def. 4.1.2]{Johnson2021}. We remark that if $\mathcal{C},\mathcal{C}'$ are monoidal categories (i.e. one-object 2-categories), then a 2-functor $\mathcal{C} \to \mathcal{C}'$ is simply a monoidal functor.

In order to represent 2-functors within the diagrammatic calculus, we use the calculus of \emph{functorial boxes}~\cite{Mellies2006}. In this  calculus we represent the functor $F_{r,s}$ by drawing a coloured box around 1- and 2-morphisms in $\mathcal{C}(r,s)$.  The calculus is summarised in~\cite[\S{}2.2]{Verdon2020}.

\begin{notation}
In this work we will use shading in diagrams for two reasons: firstly to distinguish which regions in a diagram correspond to which objects, and secondly to indicate functorial boxes. To reduce confusion we use colour only for functorial boxes. That is, regions corresponding to objects are pattern-shaded in black and white, whereas functorial boxes are in colour. 
\end{notation}
\noindent
Finally, we assume that the reader is familiar with the definition of an \emph{icon} between 2-functors~\cite[Def. 4.6.2]{Johnson2021}\cite{Lack2010}. We say that an icon is \emph{invertible} if all its 2-morphism components are invertible. 

\subsection{Pivotal dagger 2-categories}\label{sec:pivdagcats}
A pivotal dagger 2-category is a straightforward horizontal categorification of a pivotal dagger category~\cite[Sec. 7.3]{Selinger2010}\cite[Def. 3.51]{Heunen2019}.

We assume the reader is familiar with the notion of duality for 1-morphisms in a 2-category (see e.g.~\cite[Def. 6.1.1]{Johnson2021}; what they call the right adjoint we call the \emph{right dual}, and what they call the triangle equations we call the \emph{snake equations}).

Let $X: r \to s$ be a 1-morphism in a 2-category, and suppose that $[X^*: s \to r,\eta: \id_s \to X^* \otimes X,\epsilon: X \otimes X^* \to \id_r]$ is a right dual for $X$. In order to represent duality in the graphical calculus, we draw an upward-facing arrow on the $X$-wire and a downward-facing arrow on the $X^*$-wire, and draw $\eta$ and $\epsilon$ as a cup and a cap, respectively. Then the snake equations become purely topological:
\begin{calign}\label{eq:snake}
\includegraphics{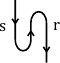}
~~=~~
\includegraphics{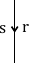}
&&
\includegraphics{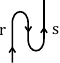}
~~=~~
\includegraphics{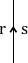}
\ignore{&&
\includegraphics{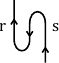}
~~=~~
\includegraphics{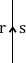}
~~~~~~~~
\includegraphics{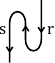}
~~=~~
\includegraphics{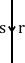}\\\nonumber
\text{right dual} &&\text{left dual}}
\end{calign}
\begin{proposition}[{\cite[Lemmas 3.6, 3.7]{Heunen2019}}]\label{prop:nestedduals}
If $[X^*,\eta_X,\epsilon_X]$ and $[Y^*,\eta_Y,\epsilon_Y]$ are right duals for $X:r \to s$ and $Y: s \to t$ respectively, then $[Y^* \otimes X^*, \eta_{X \otimes Y},\epsilon_{X \otimes Y}]$ is a right dual for $X \otimes Y$, where $\eta_{X \otimes Y}$ and $\epsilon_{X \otimes Y}$ are defined as follows:
\begin{calign}\nonumber\includegraphics{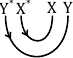}
&&
\includegraphics{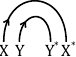}
\\\label{eq:nestedduals}
\eta_{X \otimes Y} && \epsilon_{X \otimes Y}
\end{calign}
Moreover, for any object $r$, $[\id_r,\id_{\id_r},\id_{\id_r}]$ is right dual to $\id_r$.
\end{proposition}
\begin{proposition}[{\cite[Lem. 3.4]{Heunen2019}}]\label{prop:relateduals}
Let $X: r \to s$ be a 1-morphism, and let $[X^*,\eta,\epsilon],[X^*{}',\eta',\epsilon']$ be right duals. Then there is a unique invertible 2-morphism $\alpha: X^* \to X^*{}'$ such that 
\begin{calign}\label{eq:relateduals}
\includegraphics{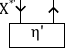}
~~=~~
\includegraphics{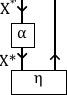}
&
\includegraphics{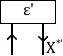}
~~=~~
\includegraphics{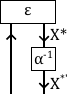}
\end{calign}
Explicitly, $\alpha$ is defined as follows:
\begin{calign}\label{eq:relatedualsisodef}
\includegraphics{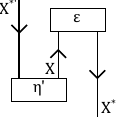}
\end{calign}
\end{proposition}
\noindent
Using duality, we can define a notion of transposition for 2-morphisms.
\begin{definition}
Let $X,Y: r \to s$ be 1-morphisms with chosen right duals $[X^*,\eta_X,\epsilon_X]$ and $[Y^*,\eta_Y,\epsilon_Y]$. For any 2-morphism $f:X \to Y$, we define its \emph{right transpose} (a.k.a. \emph{mate}) $f^*: Y^* \to X^*$ as follows:
\begin{calign}\label{eq:rtranspose}
\includegraphics{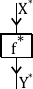}
~~=~~
\includegraphics{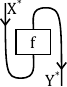}
\end{calign}
A choice of a right dual for every 1-morphism in $\mathcal{C}$ thus defines a contravariant 2-functor $*: \mathcal{C} \to \mathcal{C}$, the \emph{dual functor}, whose multiplicators and unitors are defined using Proposition~\ref{prop:nestedduals} and Proposition~\ref{prop:relateduals}.
\end{definition}
\begin{definition}\label{def:pivcat}
We say that a 2-category $\mathcal{C}$ with chosen right duals is \emph{pivotal} if there is an invertible icon $\iota:** \to \id$ from the double duals 2-functor $* \circ *: \mathcal{C} \to \mathcal{C}$ to the identity 2-functor. The invertible icon is called a \emph{pivotal structure}.
\end{definition}
\noindent
We assume the reader is familiar with the notion of a dagger 2-category and a unitary (a.k.a. dagger) 2-functor (see e.g.~\cite{Heunen2016}). Here are some basic definitions we will use throughout.
\begin{definition}\label{def:isomunitary}
Let $\mathcal{C}$ be a dagger 2-category.  We say that a 2-morphism $f: X \to Y$ is:
\begin{itemize}
\item An \emph{isometry} if $f^{\dagger} \circ f  = \id_X$. 
\item A \emph{coisometry} if $f \circ f^{\dagger} = \id_Y$.
\item \emph{Unitary} if it is an isometry and a coisometry.
\item A \emph{partial isometry} if $(f^{\dagger} \circ f)^2 = f^{\dagger} \circ f$ (or equivalently $(f \circ f^{\dagger})^2 = f \circ f^{\dagger}$).
\item \emph{Positive} if $X=Y$ and there exists some 2-morphism $g: X \to X'$ such that $f = g^{\dagger} \circ g$.
\end{itemize}
\end{definition}
\begin{definition}\label{def:daggerequiv}
Let $\mathcal{C}$ be a dagger 2-category. We say that a 1-morphism $X: r \to s$ is an \emph{equivalence} if there exists a 1-morphism $X^{-1}: s \to r$ and unitary 2-morphisms $\alpha: \id_s \to X^{-1} \otimes X$ and $\beta: X \otimes X^{-1} \to \id_r$. We sometimes write that $[X,X^{-1}, \alpha, \beta]: r \to s$ is an equivalence. If an equivalence $X: r \to s$ exists we say that the objects $r$ and $s$ are \emph{equivalent}.
\end{definition}
\noindent 
The following lemma is common knowledge.
\begin{lemma}
Let $[X,X^{-1},\alpha, \beta]$ be an equivalence in a dagger 2-category. Then there exists an equivalence $[X,X^{-1},\alpha,\beta']$ such that $[X^{-1},\alpha,\beta']$ is a right dual for $X$ (we call such an equivalence an \emph{adjoint equivalence}).
\end{lemma}
\noindent
Following~\cite{Penneys2018}, we say that a choice of right duals on a dagger 2-category is a \emph{unitary duals functor} if the associated duals functor is a dagger 2-functor. Given a unitary duals functor, there is a canonical associated pivotal structure~\cite[\S{}7.3]{Selinger2010} (for which, in particular, all of the 2-morphism components of the pivotal structure are unitary~\cite[Cor. 3.10]{Penneys2018}). In this case one may define left cups and caps as the daggers of the right cups and caps, which satisfy snake equations analogous to~\eqref{eq:snake}.
\begin{definition}
We call a dagger 2-category equipped with a unitary duals functor a \emph{pivotal dagger 2-category}.
\end{definition}
We use the following useful notation to represent morphisms in a pivotal dagger 2-category. Let $f: X \to Y$ be a 2-morphism. We first make the box for the 2-morphism $f$ asymmetric by tilting the right vertical edge. We now represent the transpose $f^*: Y^* \to X^*$ by rotating the box, as though we had `yanked' both ends of the wire in the RHS of~\eqref{eq:rtranspose}:
\begin{calign}\label{eq:graphcalctranspose}
\includegraphics{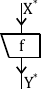}
:=
\includegraphics{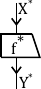}
\end{calign}
We represent the dagger $f^{\dagger}: Y \to X$ by reflection in a horizontal axis, preserving the direction of any arrows:
\begin{calign}\label{eq:graphcalcdagger}
\includegraphics[scale=1]{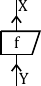}
~~:=~~
\includegraphics[scale=1]{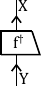}
\end{calign}
Finally, we represent the \emph{conjugate} $f_*:= (f^*)^{\dagger} = (f^{\dagger})^*$ by reflection in a vertical axis:
\begin{calign}\label{eq:conjugate}
\includegraphics[scale=1]{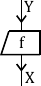}
~~:=~~
\includegraphics[scale=1]{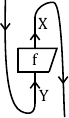}
~~=~~
\includegraphics[scale=1]{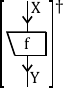}
\end{calign}
Using this notation, 2-morphisms now freely slide around cups and caps.
\begin{proposition}[{\cite[Lemma 3.12, Lemma 3.26]{Heunen2019}}]\label{prop:sliding}
Let $\mathcal{C}$ be a pivotal dagger 2-category and $f:X \to Y$ a 2-morphism. Then:
\begin{calign}\nonumber
\includegraphics{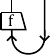}
~~=~~
\includegraphics{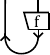}
&
\includegraphics{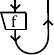}
~~=~~
\includegraphics{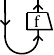}
&
\includegraphics{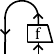}
~~=~~
\includegraphics{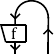}
&
\includegraphics{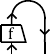}
~~=~~
\includegraphics{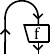}
\\\label{eq:sliding1}
\includegraphics{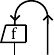}
~~=~~
\includegraphics{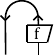}
&
\includegraphics{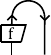}
~~=~~
\includegraphics{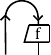}
&
\includegraphics{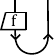}
~~=~~
\includegraphics{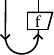}
&
\includegraphics{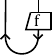}
~~=~~
\includegraphics{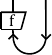}
\end{calign}
\end{proposition}
\begin{definition}\label{def:trace}
Let $X: r \to s$ be an 1-morphism and let $f: X \to X$ be a 2-morphism in a pivotal dagger 2-category $\mathcal{C}$. We define the \emph{right trace} of $f$ to be the following 2-morphism $\Tr_{R}(f): \id_r \to \id_r$: 
\begin{calign}\nonumber
\includegraphics{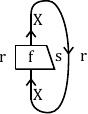}
\end{calign}
We define the \emph{right dimension} $\dim_R(X)$ of an 1-morphism $X: r \to s$ to be $\Tr_R(\id_{X})$. 
The \emph{left trace} $\Tr_L(f):\id_s \to \id_s$ and \emph{left dimension} $\dim_L(X)$ are defined analogously using the right cup and left cap.
\end{definition}

\subsection{Rigid $C^*$-2-categories}\label{sec:linearstructure}

We assume the reader is familiar with the notion of a rigid $C^*$-tensor category (see e.g.~\cite[\S{}2.1]{Neshveyev2013}). We assume that our rigid $C^*$-tensor categories are semisimple, but we do not assume that the endomorphism algebra of the tensor unit is one-dimensional.

We will now review the notion of a (presemisimple) $C^*$-2-category.
\begin{definition}
We say that a dagger 2-category is $\mathbb{C}$-linear if:
\begin{itemize}
\item For any 1-morphisms $X_1,X_2: r \to s$, the Hom-set $\Hom(X_1,X_2)$ is a complex vector space.
\item Horizontal and vertical composition induce linear maps on $\Hom$-spaces, and the dagger induces antilinear maps.
\end{itemize}
We say that an $\mathbb{C}$-linear dagger 2-category is furthermore a \emph{$C^*$-2-category} if:
\begin{itemize}
\item The vector spaces of 2-morphisms are Banach spaces, and $||f \circ g|| \leq ||f|| ||g||$.
\item $||f^{\dagger} \circ f|| = ||f||^2$ for any 2-morphism $f: X \to Y$; in particular, for any 1-morphism $X$ the $*$-algebra $\End(X)$ is a $C^*$-algebra.
\item For any 2-morphism $f: X \to Y$, the 2-morphism $f^{\dagger} \circ f$ is a positive element of the $C^*$-algebra $\End(X)$.
\end{itemize}
A \emph{$C^*$-category} can be defined in the obvious analogous way, so that the $\Hom$-categories of a $C^*$-2-category are all $C^*$-categories. We say that a 2-functor or functor is $\mathbb{C}$-linear if it induces linear maps on morphism spaces.

We say that a $C^*$-2-category is \emph{rigid} if it has duals for 1-morphisms.\footnote{This definition of rigidity for $C^*$-2-categories is only really satisfactory when $\End(\id_{r})$ is finite-dimensional for all objects $r$ of $\mathcal{C}$; the problem is that a unitary dual functor is not known to exist in general~\cite{Zito2007}. We only work with presemisimple $C^*$-2-categories, which all satisfy this condition.}
\end{definition}
\begin{remark}\label{rem:polar}
We observe that the Hom-categories of a rigid $C^*$-2-category are $W^*$-categories in the sense of~\cite[Def. 2.1]{Ghez1985}, since the Hom-spaces are finite-dimensional. This gives us a polar decomposition~\cite[Cor. 2.7]{Ghez1985}. Indeed, for any 2-morphism $f: X \to Y$, we define $|f| := (f^{\dagger} \circ f)^{1/2}$, where this is the positive square root in the f.d. $C^*$-algebra $\End(X)$. Then there exists a unique partial isometry $u: X \to Y$ such that: 
\begin{align*}
f = u \circ |f| && u^{\dagger} \circ u = s(|f|) && u \circ u^{\dagger} = s(|f^{\dagger}|)
\end{align*}
Here $s(|f|)$ is the \emph{support} of $|f|$, i.e. the least projection of all the projections $p$ in $\End(X)$ such that $p \circ |f| = |f| \circ p  = |f|$~\cite[Def. 1.10.3]{Sakai2012}.
\end{remark}
\noindent
We recall the following definitions for $C^*$-1-categories:
\begin{itemize}
\item A \emph{direct sum} of two objects $X_1,X_2$ is an object $X_1 \oplus X_2$ together with isometries $i_1: X_1 \to X_1 \oplus X_2$, $i_2: X_2 \to X_1 \oplus X_2$ such that $i_1 \circ i_1^{\dagger} + i_2 \circ i_2^{\dagger} = \id_{X_1 \oplus X_2}$. 
\item A \emph{zero object} is an object ${\bf 0}$ such that $\Hom({\bf 0},{\bf 0})$ is the zero-dimensional vector space.
\item We say that the category is \emph{additive} if it has a zero object and pairwise direct sums. 
\item For any object $X$, we say that a morphism $f \in \End(X)$ is a \emph{dagger idempotent} if $f = f^{\dagger} = f \circ f$. We say that a \emph{splitting} of the dagger idempotent is an object $V$ together with an isometry $\iota_f: V \to X$ such that $f = \iota_f \circ \iota_f^{\dagger}$. We say that the category is \emph{idempotent complete} if every dagger idempotent has a splitting. 
\item We say that the category is \emph{semisimple} if it is additive and idempotent complete, and the $C^*$-algebra $\End(X)$ is finite-dimensional for every object $X$. In a semisimple category every object is a finite direct sum of \emph{simple} objects, i.e. objects $X_i$ such that $\End(X_i) \cong \mathbb{C}$.
\end{itemize}
We say that a $C^*$-2-category $\mathcal{C}$ is \emph{locally additive}, \emph{locally semisimple}, etc. if all its $\Hom$-categories are. 

The following definitions are obvious unitary adaptations of those from~\cite[\S{}1]{Douglas}.
\begin{definition}\label{def:additive2cat}
Let $\mathcal{C}$ be a locally additive $C^*$-2-category. 
\begin{itemize}
\item We say that a \emph{zero object} in $\mathcal{C}$ is an object ${\bf 0}$ such that the category $\End({\bf 0})$ is the terminal 1-category. 
\item We say that a \emph{direct sum} of two objects $r_1, r_2$ in $\mathcal{C}$ is an object $r_1 \boxplus r_2$ with inclusion and projection 1-morphisms $\iota_i: r_i \to r_1 \boxplus r_2$, $\rho_i: r_1 \boxplus r_2 \to r_i$ such that:
\begin{itemize}
\item $\iota_i \otimes \rho_i$ is unitarily isomorphic to $\id_{r_i}$.
\item $\iota_1 \otimes \rho_2 \in \Hom(r_1,r_2)$ and $\iota_2 \otimes \rho_1 \in \Hom(r_2,r_1)$ are zero 1-morphisms.
\item $\id_{r_1 \oplus r_2}$ is a direct sum of $\rho_1 \otimes \iota_1$ and $\rho_2 \otimes \iota_2$.
\end{itemize}
\item We say that $\mathcal{C}$ is \emph{additive} if it has a zero object and direct sums. 
\end{itemize}
\end{definition}
\noindent
In order to define semisimplicity for $C^*$-2-categories we will need a notion of idempotent completeness which will be introduced in Section~\ref{sec:semisimp}. However, following~\cite[$\S{}1$]{Douglas} we can already define the following weaker notion. 
\begin{definition}
An additive $C^*$-2-category is \emph{presemisimple} if it is locally semisimple,  rigid, and every object is a finite direct sum of objects $\{r_i\}$ with simple identity, i.e. $\id_{r_i}$ is a simple object of $\End(r_i)$.

In a presemisimple $C^*$-2-category an object has simple identity if and only if it is not decomposable as a nontrivial direct sum. We call such objects \emph{simple}.
\end{definition}
\noindent
It is easy to check that zero objects and direct sums in presemisimple $C^*$-2-categories are unique up to equivalence and preserved under $\mathbb{C}$-linear unitary 2-functors. 

In order to perform computations in presemisimple $C^*$-2-categories we will make use of a convenient equivalence that categorifies matrix notation for morphisms in semisimple 1-categories~\cite[\S{}2.2.4]{Heunen2019}. These results are certainly known to experts~\cite[Chap. 8]{Heunen2019}\cite[\S{}2.1]{Reutter2019}, although we have not seen proofs elsewhere. We provide a summary in Appendix~\ref{app:mat2cats}.

Every presemisimple $C^*$-2-category $\mathcal{C}$ has a canonical unitary dual functor. This follows immediately from the more general result in~\cite{Giorgetti2019}; indeed, presemisimplicity implies finite-dimensional centres, in the language of that work. For the following proposition, we observe that for any object $r$ in $\mathcal{C}$, the $C^*$-algebra $\End(\id_r)$ is commutative. In particular, there is a unique trace mapping each of the minimal orthogonal projections to 1, which we call $\Tr_r: \End(\id_r) \to \mathbb{C}$. 

\begin{proposition}[{\cite[Prop. 7.3.3]{Giorgetti2019}}]\label{prop:standardintrinsic}
Let $X: r \to s$ be a 1-morphism in a presemisimple $C^*$-2-category $\mathcal{C}$ and let $[X^*,\eta,\epsilon]$ be a right dual. Define a map $\phi_X: \End(X) \to \mathbb{C}$ as follows:
$$\phi_X(T) := \Tr_s [ \eta^{\dagger} \circ (\id_{X^*} \otimes T) \circ \eta]$$
Define a second map $\psi_X: \End(X) \to \mathbb{C}$ as follows:
$$
\psi_X(T) := \Tr_r [ \epsilon \circ (T \otimes \id_{X^*}) \circ \epsilon^{\dagger} ]
$$
We say that $[X^*,\eta,\epsilon]$ is a \emph{standard dual} for $X$ precisely when $\phi_X = \psi_X$. In this case the map $\phi_X=\psi_X$ is tracial, positive and faithful, and does not depend on the choice of standard dual. 
\end{proposition}
\noindent 
A standard dual exists for every object~\cite[Def. 7.29]{Giorgetti2019}. It is straightforward to show (following the same approach as in the 1-categorical case~\cite[Thm. 2.2.21]{Neshveyev2013}) that a choice of standard duals for every object defines a unitary dual functor on $\mathcal{C}$. Different choices of standard duals are related by a unitary isomorphism (Proposition~\ref{prop:relateduals}). The tensor product of standard duals (Proposition~\ref{prop:nestedduals}) is standard.

\subsection{Semisimplicity}\label{sec:semisimp}

To define semisimplicity of a rigid $C^*$-2-category we need a notion of idempotent splitting at the level of 1-morphisms. In~\cite[\S{}1.3]{Douglas} it was proposed that categorified idempotents in the non-unitary setting correspond to separable monads (i.e. separable algebras in endomorphism categories). Semisimplicity corresponds to splitting of these algebras (we will explain what this means shortly). 

In the unitary $C^*$-setting, we do not want to work with all separable algebras, and so need to tighten this definition of an idempotent. We propose that the relevant idempotents in a presemisimple $C^*$-2-category are \emph{standard separable Frobenius algebras} in endomorphism categories. There is a physical motivation for this definition: as we will see in Section~\ref{sec:gc*alg}, in the category of representations of a compact quantum group $G$, Frobenius algebras correspond to pairs of a finite-dimensional $G$-$C^*$-algebra (a.k.a. $C^*$-dynamical system) and a $G$-invariant linear functional. There is a unique choice of linear functional on a $G$-$C^*$-algebra such that the corresponding Frobenius algebra is standard and separable. 

\begin{remark}\label{rem:idempcompletions}
The notion of idempotent splitting in the $C^*$-setting has already been considered in previous works; we mention now how our assumptions of standardness and separability compare. In~\cite{Chen2021}, the \emph{$Q$-system completion} of a $C^*$-2-category is defined. These $Q$-systems are separable Frobenius algebras, but they are not standard, since there is no assumption of rigidity on the $C^*$-2-category. Because there is no assumption of rigidity of the original $C^*$-2-category, the question of rigidity of the $Q$-system completion does not arise in their work. Here our additional standardness assumption is used to show rigidity of the idempotent completion.

However, in~\cite{Giorgetti2020}, an idempotent completion on a rigid $C^*$-2-category was studied, and it was stated there that, even without the standardness assumption, the completion is rigid. Therefore, it seems that it is possible to drop the standardness assumption on the Frobenius algebras, although we do not do this here. 

We remark that the idempotent completions in both these works are more general than the one we define here, since they complete a general 2-category rather than just a tensor category. It would not be hard to extend our completion to a 2-category, but we did not need this for our purposes. 
\end{remark}  

\subsubsection{Standard separable Frobenius algebras}

In this section, let $\mathcal{T}$ be a rigid $C^*$-tensor category. 
\begin{definition}\label{def:Frobenius}
An \emph{algebra} $[A,m,u]$ in $\mathcal{T}$ is an object $A$ with multiplication and unit morphisms, depicted as follows:%
\begin{calign}%\nonumber
\begin{tz}[zx,master]
\coordinate (A) at (0,0);
\draw (0.75,1) to (0.75,2);
\mult{A}{1.5}{1}
\end{tz}
&
\begin{tz}[zx,slave]
\coordinate (A) at (0.75,2);
\unit{A}{1}
\end{tz}
\\[0pt]\nonumber
m:A\otimes A \to A& u: \mathbbm{1} \to A 
\end{calign}\hspace{-0.2cm}
These morphisms satisfy the following associativity and unitality equations:
\begin{calign}\label{eq:assocandunitality}
\includegraphics[scale=1]{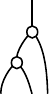}
~~=~~
\includegraphics[scale=1]{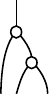}
&
\includegraphics[scale=1]{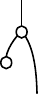}
~~=~~
\includegraphics[scale=1]{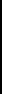}
~~=~~
\includegraphics[scale=1]{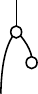}
\end{calign}
Analogously, a \textit{coalgebra} $[A,\delta,\epsilon]$ is an object $A$ with a  comultiplication $\delta: A \to A\otimes A$ and a counit $\epsilon:A\to \mathbbm{1}$ obeying the following coassociativity and counitality equations:
\begin{calign}\label{eq:coassocandcounitality}
\includegraphics[scale=1]{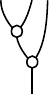}
~~=~~
\includegraphics[scale=1]{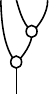}
&
\includegraphics[scale=1]{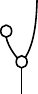}
~~=~~
\includegraphics[scale=1]{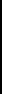}
~~=~~
\includegraphics[scale=1]{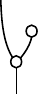}
\end{calign}
The dagger of an algebra $[A,m,u]$ is a coalgebra $[A,m^{\dagger},u^{\dagger}]$.  A algebra $[A,m,u]$ in $\mathcal{T}$ is called \textit{Frobenius} if the algebra and adjoint coalgebra structures are related by the following \emph{Frobenius equation}:
\begin{calign}\label{eq:Frobenius}
\includegraphics[scale=1]{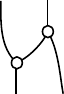}
~~=~~
\includegraphics[scale=1]{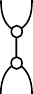}
~~=~~ 
\includegraphics[scale=1]{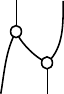}
\end{calign}
\end{definition}
\begin{definition}
Frobenius algebras are canonically self-dual. Indeed, it is easy to check that for any Frobenius algebra $A$ the following cup and cap fulfil the snake equations~\eqref{eq:snake}:
\begin{calign}\label{eq:cupcapfrob}
\includegraphics[scale=1]{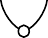}
~~:=~~
\includegraphics[scale=1]{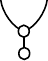}
&
\includegraphics[scale=1]{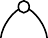}
~~:=~~
\includegraphics[scale=1]{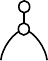}
\end{calign}
If the cup and cap~\eqref{eq:cupcapfrob} are a standard duality for $A$ (in the sense of Proposition~\ref{prop:standardintrinsic}), we say that the Frobenius algebra is \emph{standard}.
\end{definition}
\noindent
A Frobenius algebra is \emph{separable} (a.k.a. \emph{special}) if the following additional equation is satisfied:
\begin{calign}\label{eq:frobseparable}
\includegraphics[scale=1]{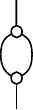}
~~=~~ 
\includegraphics[scale=1]{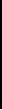}
\end{calign}
\noindent
From now on we will be concerned with separable standard Frobenius algebras ($\F$s).
\begin{definition}
Let $A,B$ be $\F$s in $\mathcal{T}$. We say that a morphism $f: A \to B$ is a \emph{$*$-homomorphism} if it obeys the following equations:
\begin{calign}\label{eq:homo}
\begin{tz}[zx, master, every to/.style={out=up, in=down},yscale=-1]
\draw (0,0) to (0,2) to [out=135] (-0.75,3);
\draw (0,2) to [out=45] (0.75, 3);
\node[zxnode=\zxwhite] at (0,1) {$f$};
\node[zxvertex=\zxwhite, zxdown] at (0,2) {};
\end{tz}
=
\begin{tz}[zx, every to/.style={out=up, in=down},yscale=-1]
\draw (0,0) to (0,0.75) to [out=135] (-0.75,1.75) to (-0.75,3);
\draw (0,0.75) to [out=45] (0.75, 1.75) to +(0,1.25);
\node[zxnode=\zxwhite] at (-0.75,2) {$f$};
\node[zxnode=\zxwhite] at (0.75,2) {$f$};
\node[zxvertex=\zxwhite, zxdown] at (0,0.75) {};
\end{tz}
&
\begin{tz}[zx,slave, every to/.style={out=up, in=down},yscale=-1]
\draw (0,0) to (0,2) ;
\node[zxnode=\zxwhite] at (0,1) {$f$};
\node[zxvertex=\zxwhite, zxup] at (0,2) {};
\end{tz}
=
\begin{tz}[zx,slave, every to/.style={out=up, in=down},yscale=-1]
\draw (0,0) to (0,0.75) ;
\node[zxvertex=\zxwhite, zxup] at (0,0.75) {};
\end{tz}
&
\begin{tz}[zx,slave, every to/.style={out=up, in=down},scale=-1]
\draw (0,0) to (0,3);
\node[zxnode=\zxwhite] at (0,1.5) {$f^\dagger$};
\end{tz}
=~~
\begin{tz}[zx,slave,every to/.style={out=up, in=down},scale=-1,xscale=1]
\draw (0,1.5) to (0,2) to [in=left] node[pos=1] (r){} (0.5,2.5) to [out=right, in=up] (1,2)  to [out=down, in=up] (1,0);
\draw (-1,3) to [out=down,in=up] (-1,1) to [out=down, in=left] node[pos=1] (l){} (-0.5,0.5) to [out=right, in=down] (0,1) to (0,1.5);
\node[zxnode=\zxwhite] at (0,1.5) {$f$};
\node[zxvertex=\zxwhite] at (l.center){};
\node[zxvertex=\zxwhite] at (r.center){};
\end{tz}
\end{calign}
We say that it is a \emph{$*$-cohomomorphism} if it obeys the following equations:
\begin{calign}\label{eq:cohomo}
\begin{tz}[zx, master, every to/.style={out=up, in=down}]
\draw (0,0) to (0,2) to [out=135] (-0.75,3);
\draw (0,2) to [out=45] (0.75, 3);
\node[zxnode=\zxwhite] at (0,1) {$f$};
\node[zxvertex=\zxwhite, zxup] at (0,2) {};
\end{tz}
=
\begin{tz}[zx, every to/.style={out=up, in=down}]
\draw (0,0) to (0,0.75) to [out=135] (-0.75,1.75) to (-0.75,3);
\draw (0,0.75) to [out=45] (0.75, 1.75) to +(0,1.25);
\node[zxnode=\zxwhite] at (-0.75,2) {$f$};
\node[zxnode=\zxwhite] at (0.75,2) {$f$};
\node[zxvertex=\zxwhite, zxup] at (0,0.75) {};
\end{tz}
&
\begin{tz}[zx,slave, every to/.style={out=up, in=down}]
\draw (0,0) to (0,2) ;
\node[zxnode=\zxwhite] at (0,1) {$f$};
\node[zxvertex=\zxwhite, zxup] at (0,2) {};
\end{tz}
=
\begin{tz}[zx,slave, every to/.style={out=up, in=down}]
\draw (0,0) to (0,0.75) ;
\node[zxvertex=\zxwhite, zxup] at (0,0.75) {};
\end{tz}
&
\begin{tz}[zx,slave, every to/.style={out=up, in=down}]
\draw (0,0) to (0,3);
\node[zxnode=\zxwhite] at (0,1.5) {$f^\dagger$};
\end{tz}
=~~
\begin{tz}[zx,slave,every to/.style={out=up, in=down},xscale=-1]
\draw (0,1.5) to (0,2) to [in=left] node[pos=1] (r){} (0.5,2.5) to [out=right, in=up] (1,2)  to [out=down, in=up] (1,0);
\draw (-1,3) to [out=down,in=up] (-1,1) to [out=down, in=left] node[pos=1] (l){} (-0.5,0.5) to [out=right, in=down] (0,1) to (0,1.5);
\node[zxnode=\zxwhite] at (0,1.5) {$f$};
\node[zxvertex=\zxwhite] at (l.center){};
\node[zxvertex=\zxwhite] at (r.center){};
\end{tz}
\end{calign}

Clearly the dagger of a $*$-homomorphism is a $*$-cohomomorphism.

If $f$ is a $*$-homomorphism and is additionally unitary, we say that it is a \emph{unitary $*$-isomorphism}. (It is easy to check that a unitary $*$-isomorphism is also a $*$-cohomomorphism.)
\end{definition}

\subsubsection{Idempotent splitting}

Let $\mathcal{C}$ be a presemisimple $C^*$-2-category, with its canonical unitary duals functor. Recall the definition of the dimension and trace in a pivotal dagger 2-category (Definition~\ref{def:trace}).

Let $X: r \to s$ be a 1-morphism, and let $[X^*,\eta,\epsilon]$ be the right dual defined by the unitary duals functor. By the $C^*$-axioms, $\dim_L(X) = \eta^{\dagger}  \circ \eta$ is a positive element of the commutative $C^*$-algebra $\End(\id_s)$. 
\begin{definition}\label{def:separable1morph}
We call a 1-morphism $X: r \to s$ in $\mathcal{C}$ \emph{separable} if $\dim_L(X)$ is invertible. We write $n_X := \sqrt{\dim_L(X)}$ for the positive square root and $n_X^{-1}$ for its (positive) inverse. 
\end{definition}
\begin{remark}
In matrix notation (Remark~\ref{rem:matrixstandardduals}) there is a $*$-isomorphism $\End(\id_s) \cong \End(\id_{\vec{\tau}})$ for some object $\vec{\tau}$ of $\Mat(\mathcal{C})$; up to permutation of the factors this isomorphism maps $\dim_L(X)$ to the matrix 
$$
\diag([\sum_k d(M_{k1}),\dots,\sum_k d(M_{kn})])
$$
where $M_{jk}$ are the entries of the 1-morphism matrix $M$ corresponding to $X$ under the equivalence $\Phi: \Mat(\mathcal{C}) \overset{\sim}{\to} \mathcal{C}$. We see that $\dim_L(X)$ is invertible precisely when the matrix $M$ has no columns of zeros.
\end{remark}
\noindent
In the following diagrams we leave regions corresponding to the object $r$ unshaded and shade regions corresponding to the object $s$ with wavy lines.
\begin{proposition}\label{prop:pairofpants}
Let $r,s$ be objects of $\mathcal{C}$  and let $X: r \to s$ be a separable 1-morphism. 

We define a \emph{pair of pants} algebra on the object $X \otimes X^*$ of the rigid $C^*$-tensor category $\End(r)$ by the following multiplication $m: (X \otimes X^*) \otimes (X \otimes X^*) \to X \otimes X^*$ and unit $u: \id_{r} \to X \otimes X^*$: 
\begin{calign}\label{eq:pairofpants}
\includegraphics[scale=1]{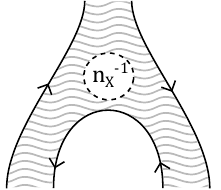}
&&
\includegraphics[scale=1]{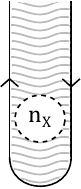}
\end{calign}
This algebra is a $\F$ in $\End(r)$.
\end{proposition}
\begin{proof}
That this is a Frobenius algebra is very easy to check (it just comes down to snake equations and isotopy) and we leave it to the reader. Separability is also clear:
\begin{calign}\label{eq:separableeq}
\includegraphics[scale=.8]{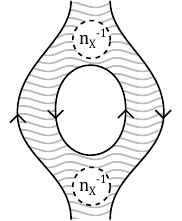}
~~
=
~~
\includegraphics[scale=.8]{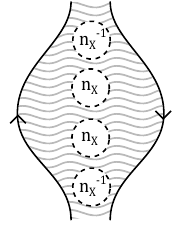}
~~=~~
\includegraphics[scale=.8]{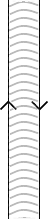}
\end{calign}
For standardness, we require (Proposition~\ref{prop:standardintrinsic}) that for any morphism $T: X \otimes X^* \to X \otimes X^*$, with respect to the Frobenius cup and cap~\eqref{eq:cupcapfrob} the left trace is equal to the right trace. This comes down to the following equation: 
\begin{calign}
\Tr_r[
\includegraphics[scale=.8]{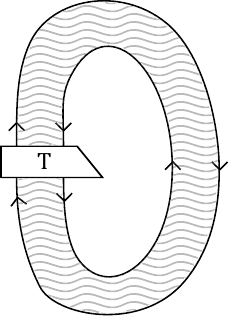}]
~~=~~
\Tr_r[\includegraphics[scale=.8]{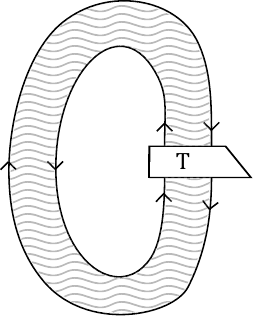}]
\end{calign}
For this we observe that the Frobenius cup and cap is simply the tensor product cup and cap on $X \otimes X^*$ (Proposition~\ref{prop:nestedduals}, which is standard.
\end{proof}
\noindent
We now define semisimplicity. This precisely corresponds to~\cite[Def. 3.34]{Chen2021}, except that our Frobenius algebras are standard.
\begin{definition}
Let $\mathcal{C}$ be a presemisimple $C^*$-2-category. Let $r$ be an object of $\mathcal{C}$. We say that a $\F$ $A$ in $\End(r)$ \emph{splits} if there exists an object $s$ of $\mathcal{C}$ and a separable 1-morphism $X: r \to s$ such that $A$ is unitarily $*$-isomorphic to the pair of pants algebra $X \otimes X^*$.

We say that $\mathcal{C}$ is \emph{semisimple} if, for every object $r$ of $\mathcal{C}$, every $\F$ in $\End(r)$ splits.
\end{definition}
\ignore{
\noindent
Before moving on we note the following lemma, which will be useful later.
\begin{lemma}\label{lem:funnyF}
Let $X: r \to s$ be a separable 1-morphism, and let $a \in \End(\id_s)$, $b \in \End(\id_r)$ be positive invertible elements such that the following multiplication and unit define an $\F$ on $X \otimes X^*$:
\begin{calign}\label{eq:funnyF}
\includegraphics[scale=.8]{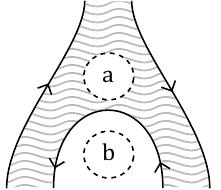}
&&
\includegraphics[scale=.8]{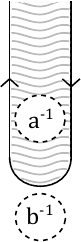}
\end{calign}
Then the following equation is obeyed:
\begin{calign}
\includegraphics[scale=.8]{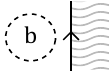}
~~=~~
\includegraphics[scale=.8]{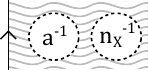}
\end{calign}
In particular, the $\F$~\eqref{eq:funnyF} is identical to the $\F$~\eqref{eq:pairofpants}.
\end{lemma}
\begin{proof}
By the same argument as~\eqref{eq:separableeq}, separability implies the following equation:
\begin{calign}\label{eq:funnyspec}
\includegraphics[scale=.8]{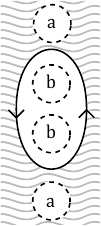}
~~=~~
\includegraphics[scale=.8]{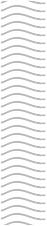}
\end{calign}
Standardness of the $\F$~\eqref{eq:funnyF}, together with standardness of the tensor product dual on $X \otimes X^*$ and Proposition~\ref{prop:standarddualsrelbyunitary}, implies that $b^{-1} \otimes \id_{X \otimes X^*} \otimes b$ is unitary. Since $b$ is positive, it follows that $$b^{-2} \otimes  \id_{X \otimes X^*} \otimes b^{2} = \id_{X \otimes X^*}.$$
Now we obtain the following series of implications:
\begin{calign}
\includegraphics[scale=.8]{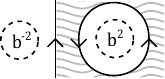}
~~=~~
\includegraphics[scale=.8]{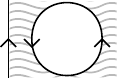}
&\Rightarrow
\includegraphics[scale=.8]{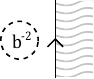}
~~=~~
\includegraphics[scale=.8]{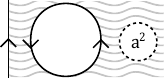}
\\
&\Rightarrow
b^{-2} \otimes \id_X = \id_X \otimes n_{X}^{2} \otimes a^2  
\\
&\Rightarrow 
b \otimes \id_X = \id_X \otimes n_X^{-1} \otimes a^{-1}
\end{calign}
Here for the first implication we postcomposed on both sides with $\id_X \otimes a^2$, and used~\eqref{eq:funnyspec}; for the second implication we used the definition $n_X^2 = \dim_L(X)$; and for the final implication we took the inverse of the positive square root on both sides in the f.d. $C^*$-algebra $\End(X)$.
\end{proof}
}

\section{Two semisimple completions of a rigid $C^*$-tensor category}

Let $\mathcal{T}$ be a rigid $C^*$-tensor category. We are about to define two semisimple $C^*$-2-categories in which $\mathcal{T}$ embeds as the endomorphism category of a fixed object. We will then show that these two 2-categories are equivalent.
 
\subsection{The 2-category $\Bimod(\mathcal{T})$}\label{sec:bimod}

The following construction is  identical to the constructions in~\cite[Def. 3.17]{Chen2021}\cite[Notation 2.16]{Giorgetti2020}, except that our Frobenius algebras are standard as well as separable. 

\subsubsection{Definition}

In what follows let $\mathcal{T}$ be a rigid $C^*$-tensor category. 
\def\d{0.5}
\def\h{2.25}
\def\inang{-45}
\begin{definition} 
Let $A$ and $B$ be $\F$s in $\mathcal{T}$. A \emph{left dagger $A$-module} is an object $M$ in $\mathcal{T}$ together with a morphism $\rho: A \otimes M \to M$ (the \emph{left action}) fulfilling the following equations:
\begin{align}\label{eq:module}
\begin{tz}[zx,every to/.style={out=up, in=down}]
\draw (0,0) to (0,3);
\draw (-\d-1.5,0) to [in=-135] (-\d-0.75,1) to[in=180-\inang] (0,\h);
\draw (-\d,0) to [in=-45] (-\d-0.75,1) ;
\node[zxvertex=\zxwhite, zxdown] at (-\d-0.75,1){};
\node[box,zxdown] at (0,\h) {$\rho$};
\end{tz}
~~=~~
\def\htop{2.25}
\def\hbot{1.25}
\begin{tz}[zx,every to/.style={out=up, in=down}]
\draw (0,0) to (0,3);
\draw (-\d-1.5,0) to [in=-135] (0,\htop);
\draw (-\d,0) to [in=-135] (0,\hbot);
\node[box,zxdown] at (0,\hbot) {$\rho$};
\node[box,zxdown] at (0,\htop) {$\rho$};
\end{tz}
&&
\begin{tz}[zx, every to/.style={out=up, in=down}]
\draw (0,0) to (0,3);
\draw (-\d,1.2) to [in=-135] (0,\h);
\node[box,zxdown] at (0,\h) {$\rho$};
\node[zxvertex=\zxwhite] at (-\d,1.2){};
\end{tz}
~~=~~~
\begin{tz}[zx, every to/.style={out=up, in=down}]
\draw (0,0) to (0,3);
\end{tz}
&&
\def\x{0.2}
\begin{tz}[zx, every to/.style={out=up, in=down},xscale=0.8]
\draw (0,0) to (0,3);
\draw (-\x,1.5) to [out=up, in=right] (-0.75-\x, 2.25) to [out=left, in=up] (-1.5-\x, 1.5) to (-1.5-\x,0);
\node[zxvertex=\zxwhite] at (-0.75-\x, 2.25){};
\node[box] at (0,1.5) {$\rho^\dagger$};
\end{tz}
~~=~~
\begin{tz}[zx, every to/.style={out=up, in=down},xscale=0.8]
\draw (0,0) to (0,3);
\draw (-1.25, 0) to [in=-135] (0,1.95) ;
\node[box] at (0,1.95) {$\rho$};
\end{tz}
\end{align}
A \emph{right dagger $B$-module} is defined similarly, with an \emph{right action} $\rho: M \otimes B \to M$ and the analogous equations.  An $A-B$-dagger bimodule is an object $M$ which is a left dagger $A$-module and a right dagger $B$-module, such that the left and right actions commute:
\begin{align}\label{eq:commute}
\begin{tz}[zx, every to/.style={out=up, in=down}]
\draw (0,0) to (0,3);
\draw (1,0) to [in=-45] (0,\h);
\draw (-1,0) to [in=-135] (0, 1.5);
\node[boxvertex,zxdown] at (0,1.5){};
\node[boxvertex,zxdown] at (0,\h) {};
\end{tz}
~=~
\begin{tz}[zx, every to/.style={out=up, in=down},xscale=-1]
\draw (0,0) to (0,3);
\draw (1,0) to [in=-45] (0,\h);
\draw (-1,0) to [in=-135] (0, 1.5);
\node[boxvertex,zxdown] at (0,1.5){};
\node[boxvertex,zxdown] at (0,\h) {};
\end{tz}
=:
\begin{tz}[zx, every to/.style={out=up, in=down}]
\draw (0,0) to (0,3);
\draw (1,0) to [in=-45] (0,\h);
\draw (-1,0) to [in=-135] (0,\h);
\node[boxvertex,zxdown] at (0,\h) {};
\end{tz}
\end{align} 
\end{definition}
\noindent
Every $\F$ $A$ has a trivial $A{-}A$-dagger bimodule ${}_AA_A$:
\begin{calign}
\begin{tz}[zx, every to/.style={out=up, in=down}]
\draw (0,0) to (0,3);
\draw (-1,0) to [in=-135] (0,2.);
\draw (1,0) to [in=-45] (0,2.);
\node[boxvertex,zxdown] at (0,2.){};
\end{tz}
~~:= ~~
\begin{tz}[zx]
\coordinate(A) at (0.25,0);
\draw (1,1) to [out=up, in=-135] (1.75,2);
\draw (1.75,2) to [out=-45, in=up] (3.25,0);
\draw (1.75,2) to (1.75,3);
\mult{A}{1.5}{1}
\node[zxvertex=\zxwhite,zxdown] at (1.75,2){};
\end{tz}
~~= ~~
\begin{tz}[zx,xscale=-1]
\coordinate(A) at (0.25,0);
\draw (1,1) to [out=up, in=-135] (1.75,2);
\draw (1.75,2) to [out=-45, in=up] (3.25,0);
\draw (1.75,2) to (1.75,3);
\mult{A}{1.5}{1}
\node[zxvertex=\zxwhite,zxdown] at (1.75,2){};
\end{tz}\end{calign}%
\begin{definition} A \textit{bimodule homomorphism} $_AM_B\to {}_AN_B$ is a morphism $f:M\to N$ that commutes with the $A$-$B$ action:
\begin{calign}\label{eq:bimodmorph}
\begin{tz}[zx]
\draw (0,0) to (0,3);
\draw (-1,0) to [out=up, in=-135] (0,2.15);
\draw (1,0) to [out=up, in=-45] (0,2.15) ;
\node[zxnode=\zxwhite] at (0,0.85) {$f$};
\node[boxvertex,zxdown] at (0,2.15){};
\end{tz}
=
\begin{tz}[zx]
\draw (0,0) to (0,3);
\draw (-1,0) to [out=up, in=-135] (0,0.85);
\draw (1,0) to [out=up, in=-45] (0,0.85) ;
\node[zxnode=\zxwhite] at (0,2.15) {$f$};
\node[boxvertex,zxdown] at (0,0.85){};
\end{tz}
\end{calign}
Two dagger bimodules are \textit{(unitarily) isomorphic} if there is a (unitary) invertible bimodule homomorphism ${}_AM_B\to{}_AN_B$.

Given two $\F$s $A,B$, the $A$-$B$ dagger bimodules and bimodule homomorphisms form a category which we write as $A$-$\Mod$-$B$. Left dagger $A$-modules and right dagger $A$-modules likewise form categories which we write as $A$-$\Mod$ and $\Mod$-$A$ respectively.
\end{definition}
\noindent
Since dagger idempotents in $\mathcal{T}$ split, we can compose dagger bimodules ${}_AM_B$ and ${}_BN_C$ to obtain an $A{-}C$-dagger bimodule ${}_AM{\otimes_B}N_C$, as follows. First we observe that the following endomorphism is a dagger idempotent (for this, we use that the Frobenius algebra $B$ is separable):
\begin{calign}\label{eq:idempotentforrelprod}
\begin{tz}[zx,every to/.style={out=up, in=down}]
\draw (0,0) to (0,3);
\draw (2,0) to (2,3);
\draw (0,2.25) to [out=-45, in=left] (1, 1.2) to[out=right, in=-135] (2,2.25);
\node[zxvertex=\zxwhite] at (1,1.2){};
\node[boxvertex,zxdown] at (0,2.25){};
\node[boxvertex,zxdown] at (2,2.25){};
\node[dimension, left] at (0,0) {$M$};
\node[dimension, right] at (2,0) {$N$};
\end{tz}
\end{calign}
The \emph{relative tensor product} ${}_AM{\otimes_B}N_C$, or \emph{tensor product of bimodules}, is defined as the object obtained by splitting this idempotent. We depict the isometry $i: M\otimes_B N\to M\otimes N$ as a downwards pointing triangle:
\begin{calign}\label{eq:moritaidempotentsplit}
\begin{tz}[zx,every to/.style={out=up, in=down}]
\draw (0,0) to (0,3);
\draw (2,0) to (2,3);
\draw (0,2.25) to [out=-45, in=left] (1, 1.2) to[out=right, in=-135] (2,2.25);
\node[zxvertex=\zxwhite] at (1,1.2){};
\node[boxvertex,zxdown] at (0,2.25){};
\node[boxvertex,zxdown] at (2,2.25){};
\end{tz}
~~=~~
\begin{tz}[zx,every to/.style={out=up, in=down}]
\draw (0,0) to (0,0.5);
\draw (0,2.5) to (0,3);
\draw (2,0) to (2,0.5);
\draw (2,2.5) to (2,3);
\draw (1,0.5) to (1,2.5);
\node[dimension, right] at (1,1.5) {$M{\otimes_B}N$};
\node[triangleup=2] at (1,0.5){};
\node[triangledown=2] at (1,2.5){};
\end{tz}
&
\begin{tz}[zx]
\clip (-1.2, -0.3) rectangle (1.9,3.3);
\draw (0,0) to (0,1);
\draw (-1,1) to (-1,2);
\draw (1,1) to (1,2);
\draw (0, 2) to (0,3);
\node[triangleup=2] at (0,2){};
\node[triangledown=2] at (0,1){};
\node[dimension, right] at (0,0) {$M{\otimes_B}N$};
\end{tz}
=~~
\begin{tz}[zx]
\clip (-0.2, -0.3) rectangle (1.9,3.3);
\draw (0,0) to (0,3);
\node[dimension, right] at (0,0) {$M{\otimes_B}N$};
\end{tz}
\end{calign}
\noindent
For dagger bimodules ${}_AM_B$ and ${}_BN_C$, the relative tensor product $M\otimes_B N$ is itself an $A{-}C$-dagger bimodule with the following action $A\otimes(M{\otimes_B}N) \otimes C\to M{\otimes_B}N$:
\begin{equation}
\begin{tz}[zx,every to/.style={out=up, in=down}]
\draw (0,-0.) to (0,1) ;
\draw (0,3) to (0,3.5);
\draw (-1,1) to (-1,2.5);
\draw (1,1) to (1,2.5);
\draw (-2,-0.) to [in=-135] (-1,1.75);
\draw (2,-0.) to [in=-45] (1,1.75);
\node[triangledown=2] at (0,1){};
\node[triangleup=2] at (0,2.5){};
\node[boxvertex] at (-1,1.75){};
\node[boxvertex] at (1,1.75){};
\end{tz}
\end{equation}
\noindent
The relative tensor product is also defined on morphisms of bimodules.
Let ${}_AM_B, {}_AM'_B$ and ${}_BN_C, {}_BN'_C$ be dagger bimodules and let $f:{}_AM_B \to {}_AM'_B$ and $g: {}_BN_C \to {}_BN'_C$ be bimodule homomorphisms. Then the relative tensor product $f \otimes_B g: {}_A M \otimes_B N_{C} \to {}_A M' \otimes_B N'_C$ is  a bimodule homomorphism defined as follows:
\begin{calign}
\includegraphics[scale=.7]{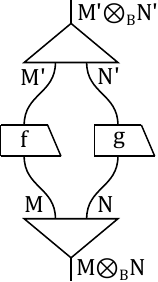}
\end{calign}
\begin{definition}
Let $\mathcal{T}$ be a rigid $C^*$-tensor category. We define a  $C^*$-2-category $\Bimod(\mathcal{T})$ as follows:
\begin{itemize}
\item \emph{Objects.} Standard separable Frobenius algebras $A, B, \dots$ in $\mathcal{T}$. 
\item \emph{Hom-categories.} $\Hom(A,B) := A$-$\Mod$-$B$. (The $C^*$-norm is that of $\mathcal{T}$.)
\item \emph{Horizontal composition.} Relative tensor product.
\item \emph{Associator.} For $M: A \to B$, $N: B \to C$, $O: C \to D$ the associator component $\alpha_{M,N,O}$ is defined as follows: 
\begin{calign}
\includegraphics[scale=.7]{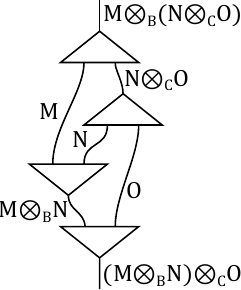}
\end{calign}
\item \emph{Identity 1-morphisms.} We define $\id_{A}: A \to A$  to be the dagger bimodule ${}_A A_{A}$.
\item \emph{Unitors.} For $M: A \to B$ the left and right unitor components $\lambda_M$ and $\rho_M$ are defined as follows:
\begin{calign}
\includegraphics[scale=.7]{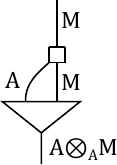}
&&
\includegraphics[scale=.7]{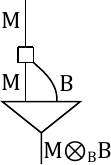}
\end{calign}
\end{itemize}
\end{definition}
\noindent
We leave to the reader the straightforward checks that $\Bimod(\mathcal{T})$ is indeed a well-defined $C^*$-2-category.
\ignore{ (e.g. the dagger of a bimodule homomorphism is a bimodule homomorphism, relative tensor product is a functor, the associator and unitors are indeed unitary natural isomorphisms obeying the pentagon and triangle equations, etc.).}

We observe that $\mathcal{T}$ embeds in $\Bimod(\mathcal{T})$ as an endomorphism category. 
\begin{proposition}\label{prop:bimodembeds}
Let $\mathbbm{1}$ be the trivial $\F$ in $\mathcal{T}$. There is a unitary isomorphism of $C^*$-tensor categories $F: \mathcal{T} \xrightarrow{\sim} \End(\mathbbm{1})$ defined as follows:
\begin{itemize}
\item Every object of $\mathcal{T}$ is taken to itself considered as a bimodule over the trivial $\F$.
\item Every morphism of $\mathcal{T}$ is taken to itself considered as a bimodule homomorphism with respect to the actions of the trivial $\F$. 
\end{itemize}
\ignore{This unitary isomorphism preserves the pivotal structure of $\mathcal{C}$ on the nose, i.e we have that $[F(X^*), F(\eta),F(\epsilon)]$ is the chosen right dual of $F(X)$ in $\Bimod(\mathcal{C})$ and $\hat{\iota}_{F(X)} = F(\iota_X)$ for any object $X$ of $\mathcal{C}$.}
\end{proposition}

\subsubsection{Semisimplicity}

We will now show that $\Bimod(\mathcal{T})$ is rigid. The definition of the right duals here is from~\cite{Yamagami2004}, which deals with the non-unitary case.
\begin{definition}
Let $A,B$ be $\F$s in $\mathcal{T}$ and let ${}_AM_B$ be a dagger bimodule. We define the \emph{dual} dagger bimodule ${}_B (M^*)_A$ as follows. The underlying object of the bimodule is the dual object $M^*$ of $M$ in the rigid $C^*$-tensor category $\mathcal{T}$. The left $B$-action is defined as follows:
\begin{calign}\label{eq:dualbactiondef}
\includegraphics[scale=.7]{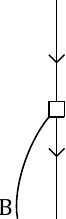}
~~:=~~
\includegraphics[scale=.7]{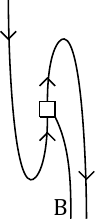}
\end{calign}
The right $A$-action is defined as follows:
\begin{calign}\label{eq:dualaactiondef}
\includegraphics[scale=.7]{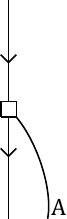}
~~:=~~
\includegraphics[scale=.7]{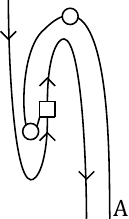}
\end{calign}
\end{definition}
\noindent
The following lemma shows that one can equally well express these actions in terms of the left cup and cap in the pivotal dagger category $\mathcal{T}$.
\begin{lemma}\label{lem:standardpivmod}
Let $A,B$ be $\F$s in $\mathcal{C}$ and let ${}_AM_B$ a dagger bimodule. Then the following equations hold:
\begin{calign}\label{eq:standardpivmod}
\includegraphics[scale=.7]{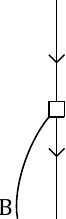}
~~=~~
\includegraphics[scale=.7]{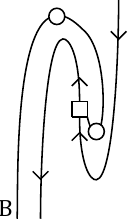}
&&
\includegraphics[scale=.7]{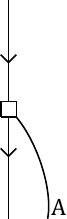}
~~=~~
\includegraphics[scale=.7]{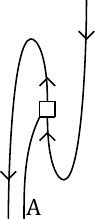}
\end{calign}
\end{lemma}
\begin{proof}
We show the second equation; the proof of the first is similar. Since $A$ is standard, by Proposition~\ref{prop:relateduals} there exists a unitary $U: A^* \to A$ such that the following equation is satisfied:
\begin{calign}
\includegraphics[scale=.7]{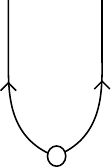}
~~=~~
\includegraphics[scale=.7]{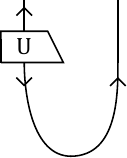}
\end{calign}
We then have the following sequence of equalities (where we offset the edge of the module action box in order to make the transpose visible):
\begin{calign}
\includegraphics[scale=.7]{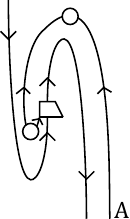}
~~=~~
\includegraphics[scale=.7]{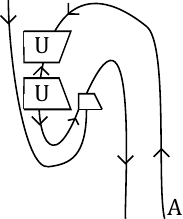}
~~=~~
\includegraphics[scale=.7]{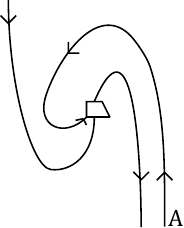}
~~=~~
\includegraphics[scale=.7]{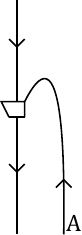}
~~=~~
\includegraphics[scale=.7]{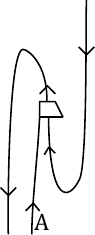}
\end{calign}
\end{proof}
\noindent
We now prove that these maps are indeed dagger module actions. 
\begin{proposition}
The maps~\eqref{eq:dualbactiondef} and \eqref{eq:dualaactiondef} give $M^*$ the structure of a left dagger $B$-module and a right dagger $A$-module respectively.
\end{proposition}
\begin{proof}
We provide the proof for the right $A$-action; the proof for the left $B$-action is similar. 
\begin{calign}
\includegraphics[scale=.7]{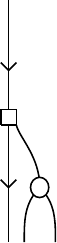}
~~=~~
\includegraphics[scale=.7]{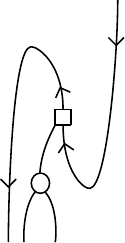}
~~=~~
\includegraphics[scale=.7]{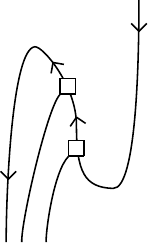}
~~=~~
\includegraphics[scale=.7]{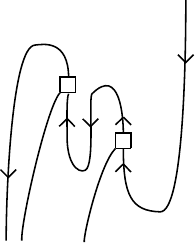}
~~=~~
\includegraphics[scale=.7]{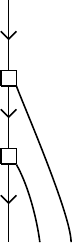}
\\
\includegraphics[scale=.7]{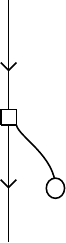}
~~=~~
\includegraphics[scale=.7]{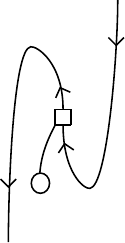}
~~=~~
\includegraphics[scale=.7]{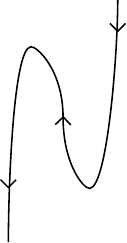}
~~=~~
\includegraphics[scale=.7]{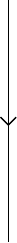}
\\
\includegraphics[scale=.7]{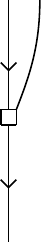}
~~=~~
\includegraphics[scale=.7]{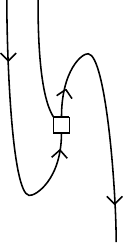}
~~=~~
\includegraphics[scale=.7]{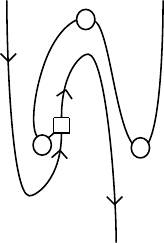}
~~=~~
\includegraphics[scale=.7]{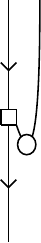}
\end{calign}
\end{proof}
\noindent
Having defined our dual 1-morphism, we now define a cup 2-morphism $\eta_{{}_A M_B}:  {}_B B_B \to  {}_B(M^*) \otimes_A M_B $ and a cap 2-morphism $\epsilon_{{}_A M_B}: {}_A M \otimes_B (M^*)_A \to {}_A A_A$ witnessing the duality:
\begin{calign}\label{eq:dualbimodcupcap}
\includegraphics[scale=.7]{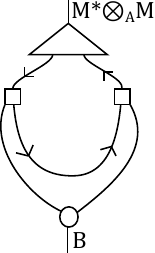} &&
\includegraphics[scale=.7]{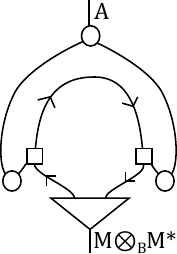}
\end{calign}
\begin{proposition}\label{prop:bimodrigid}
$\Bimod(\mathcal{T})$ is rigid; in particular, $[{}_B (M^*)_A, \eta_{{}_A M_B}, \epsilon_{{}_A M_B}]$ is a right dual for ${}_AM_B$.  
\end{proposition}
\begin{proof}
It is straightforward to check that the cup and cap~\eqref{eq:dualbimodcupcap} are bimodule homomorphisms.
We also need to check that the snake equations~\eqref{eq:snake} are satisfied. We show the second of those equations (the other is shown similarly):
\begin{calign}\label{eq:bimoddualpf}
\includegraphics[scale=.7]{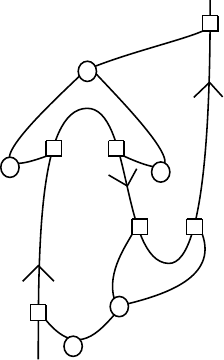}
~~=~~
\includegraphics[scale=.7]{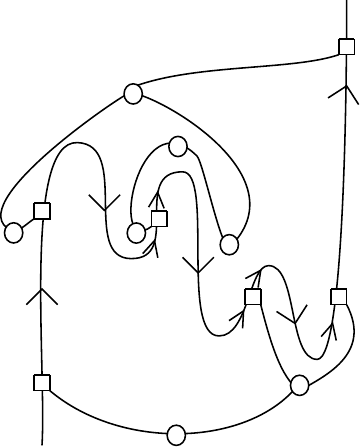}
~~=~~
\includegraphics[scale=.7]{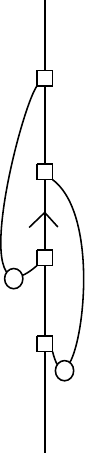}
~~=~~
\includegraphics[scale=.7]{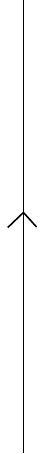}
\end{calign}
\end{proof}
\noindent
Here the first equality is by definition of the action for the dual bimodule; the second equality is by snake equations, the dagger bimodule equations and separability of the Frobenius algebras $A$ and $B$; and the third equality is by commutativity of the left and right module actions, the dagger module equations and separability of the Frobenius algebras $A$ and $B$.
\begin{remark}
In~\eqref{eq:bimoddualpf} we omitted the triangles~\eqref{eq:moritaidempotentsplit} in order to keep the size of the diagrams reasonable; the reader may insert them, and will observe that they cancel using separability of the Frobenius algebra and the following equations:
\begin{calign}\label{eq:bimoddualtriangles}
\includegraphics[scale=.7]{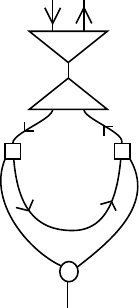}
~~=~~
\includegraphics[scale=.7]{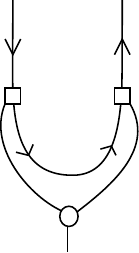}
&&
\includegraphics[scale=.7]{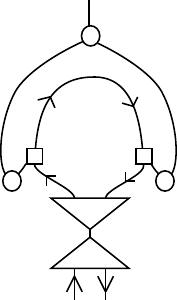}
~~=~~
\includegraphics[scale=.7]{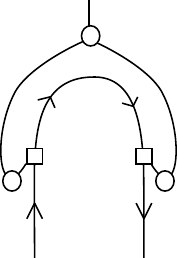}
\end{calign}
We prove the first equation of~\eqref{eq:bimoddualtriangles}, the other is shown similarly:
\begin{calign}
\includegraphics[scale=.7]{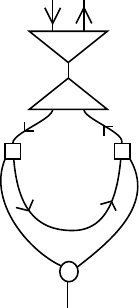}
~~=~~
\includegraphics[scale=.7]{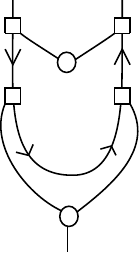}
~~=~~
\includegraphics[scale=.7]{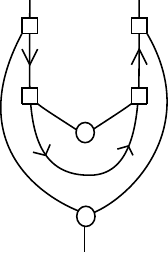}
~~=~~
\includegraphics[scale=.7]{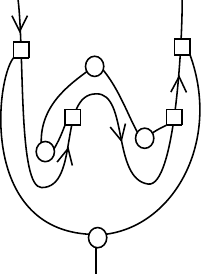}
~~=~~
\includegraphics[scale=.7]{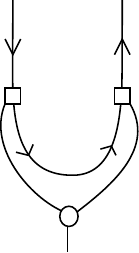}
\end{calign}
Here the last equality is by two snake equations, the dagger module equations and separability of the Frobenius algebra $A$.
\end{remark}
\noindent
We now show the other aspects of semisimplicity.
\begin{definition}
We say that an $\F$ $A$ in $\mathcal{T}$ is \emph{simple} if the $C^*$-algebra $\End({}_{A} A_A)$ of bimodule endomorphisms of the identity $A{-}A$ bimodule is one-dimensional.
\end{definition}
\begin{lemma}\label{lem:decomposingfs}
Every $\F$ in $\mathcal{T}$ may be decomposed as a direct sum of simple $\F$s. 
\end{lemma}
\begin{proof}
Let $A$ be a $\F$. Take a complete family of minimal orthogonal projections $\{p_i\}$ in the commutative finite-dimensional $C^*$-algebra $\End({}_A A_A)$, split the idempotents to obtain factors $\{A_i\}$ with isometries $\iota_i: A_i \to A$ such that $\iota_i \circ \iota_i^{\dagger} = p_i$, then define the structure of an $\F$ on $A_i$ by: 
\begin{align*}
m_i = \iota_i^{\dagger} \circ m \circ (\iota_i \otimes \iota_i)
&&
u_i = \iota_i^{\dagger} \circ u
\end{align*}
It is easy to check that the $A_i$ are Frobenius algebras and that $A \cong \oplus_i A_i$; simplicity of the $A_i$ follows from minimality of the projections $\{p_i\}$. This was already shown in~\cite[Lem. 2.8]{Neshveyev2018}. We now need only show that the $A_i$ are standard and separable. For separability:
\begin{align*}
m_i \circ m_i^{\dagger} =  \iota_i^{\dagger} \circ m \circ (p_i \otimes p_i) \circ m^{\dagger} \circ \iota_i = \iota_i^{\dagger} \circ m \circ m^{\dagger} \circ \iota_i = \id_{A_i}
\end{align*}
Here the second equality uses the fact that $p_i$ is a bimodule morphism to pull the $p_i$ through the multiplication $m$ and cancel it with $\iota_{i}^{\dagger}$; the third equality uses separability of $A$ and the fact that $\iota_i$ is an isometry.

Finally we show that the $A_i$ are standard. For this, we need to show that, for any $f \in \End(A_i)$:
$$\Tr[u_i^{\dagger} \circ m_i \circ (f \otimes \id_{A_i}) \circ m_i^{\dagger} \circ u_i] = \Tr[u_i^{\dagger} \circ m_i \circ (\id_{A_i} \otimes f) \circ m_i^{\dagger} \circ u_i] $$
We show this as follows:
\begin{align*}
\Tr[u_i^{\dagger} \circ m_i \circ (f \otimes \id_{A_i}) \circ m_i^{\dagger} \circ u_i] &= \Tr[u^{\dagger} \circ  p_i \circ  m \circ (\iota_i f \iota_i^{\dagger} \otimes p_i) \circ  m^{\dagger} \circ  p_i \circ u]
\\ &= \Tr[u^{\dagger} \circ m \circ (\iota_i f \iota_i^{\dagger} \otimes \id_{A}) \circ m^{\dagger} \circ u] 
\\ &=  \Tr[u^{\dagger} \circ m \circ (\id_{A} \otimes \iota_i f \iota_i^{\dagger}) \circ m^{\dagger} \circ u]
\\ &=  \Tr[u_i^{\dagger} \circ m_i \circ (\id_{A_i} \otimes f) \circ m_i^{\dagger} \circ u_i] 
\end{align*}
Here the first equality uses the definition of the multiplication and unit of $A_i$; the second equality uses that $p_i$ is a bimodule morphism to bring the $p_i$ next to the $\iota_i$ and $\iota_i^{\dagger}$, where they disappear; the third equality is by standardness of $A$; and for the final equality one simply repeats the process in the opposite direction.
\end{proof}

\begin{proposition}\label{prop:bimodsemisimple}
$\Bimod(\mathcal{T})$ is semisimple.
\end{proposition}
\begin{proof}
We showed in Proposition~\ref{prop:bimodrigid} that $\Bimod(\mathcal{T})$ is rigid.

There is clearly a direct sum of bimodules and a zero bimodule yielding local additivity. For an  $A$-$B$ bimodule $X$ and $p \in \End(X)$ it is straightforward to define an $A$-$B$ bimodule structure on the splitting of the idempotent $p$; local idempotent completeness follows. Then observe that every endomorphism algebra in $\Bimod(\mathcal{T})$ is finite-dimensional, since it is a subalgebra of an endomorphism algebra in $\mathcal{T}$. $\Bimod(\mathcal{T})$ is therefore locally semisimple. 

It is straightforward to check that the direct sum of $\F$s in $\mathcal{T}$ is a direct sum of objects in $\Bimod(\mathcal{T})$ in the sense of Definition~\ref{def:additive2cat}. The existence of a zero object is clear. That every object can be decomposed as a finite direct sum of simple objects is the content of Lemma~\ref{lem:decomposingfs}.

Finally, idempotent splitting is shown in~\cite[Cor. 3.37]{Chen2021}.
\end{proof}

\begin{remark}\label{rem:standarddualbimod}
Let ${}_A M_B$ be a bimodule. In this section we defined a right dual bimodule ${}_B (M^*)_A$. With the cup and cap specified in~\eqref{eq:dualbimodcupcap}, this is in general not a standard right dual in $\Bimod(\mathcal{T})$. However, it is straightforward to define a normalised cup and cap $\eta,\epsilon$ so that $[{}_B (M^*)_A,\eta,\epsilon]$ is a standard right dual for ${}_A M_B$. We leave the details to the reader. 
\end{remark}

\subsection{The 2-category $\Mod(\mathcal{T})$}\label{sec:modt}

We now define the second semisimple $C^*$-2-category in which $\mathcal{T}$ embeds. 
\begin{definition}
A \emph{semisimple left $\mathcal{T}$-module category} is a semisimple $C^*$-category $\mathcal{M}$ together with:
\begin{itemize}
\item A unitary linear bifunctor $\tilde{\otimes}: \mathcal{T} \times \mathcal{M} \to \mathcal{M}$. 
\item Unitary natural isomorphisms 
$l_X: \mathbbm{1} \tilde{\otimes} X \cong X$ and $m_{U,V,X}: (U \otimes V) \tilde{\otimes} X \cong U \tilde{\otimes} (V \tilde{\otimes} X)$ satisfying analogues of the pentagon and triangle equations~\cite[Def. 6]{Ostrik2003}. 
\end{itemize}
A semisimple right $\mathcal{T}$-module category can be defined analogously.

Following~\cite{Arano2015,Neshveyev2018}, we say that the module category $\mathcal{M}$ is \emph{cofinite} (a.k.a. \emph{proper}) if for any $X,Y \in \mathcal{M}$ we have $\Hom_{\mathcal{M}}(X,U_i \tilde{\otimes} Y) = 0$ for all but finitely many $i$, where $\{U_i\}$ are representatives of the isomorphism classes of simple objects in $\mathcal{T}$.
\end{definition}

\begin{definition}
Let $\mathcal{M}_1, \mathcal{M}_2$ be semisimple left $\mathcal{T}$-module categories. A \emph{unitary $\mathcal{T}$-module functor} $\mathcal{M}_1 \to \mathcal{M}_2$ is a unitary linear functor $F:\mathcal{M}_1 \to \mathcal{M}_2$ together with a unitary natural isomorphism $c_{U,X}: F(U \tilde{\otimes} X) \to U \tilde{\otimes} F(X)$; the $\{c_{U,X}\}$ must satisfy certain coherence equations~\cite[Def. 7]{Ostrik2003}.
\end{definition}

\begin{definition}
Let $F,G: \mathcal{M}_1 \to \mathcal{M}_2$ be unitary $\mathcal{T}$-module functors. We say that a natural transformation $\eta: F \to G$ is a \emph{morphism of $\mathcal{T}$-module functors} if the following diagram commutes for any $U \in \mathcal{T}, X \in \mathcal{M}_1$:
\begin{diagram}
F(U \tilde{\otimes} X) & \rTo^{c_{U,X}} & U \tilde{\otimes} F(X) \\
\dTo^{\eta_{U \tilde{\otimes} X}} & &  \dTo_{\id_U \otimes \eta_{X}}
\\
G(U \tilde{\otimes} X) & \rTo_{c_{U,X}} & U \tilde{\otimes} G(X)
\end{diagram}
\end{definition}
\ignore{
\noindent
Before defining what we mean by `finitely decomposable' we first consider the 2-category $\Mod'(\mathcal{T})$ whose objects are general (i.e. not necessarily finitely decomposable) cofinite semisimple left $\mathcal{T}$-module categories, whose 1-morphisms are unitary $\mathcal{T}$-module functors, and whose 2-morphisms are morphisms of $\mathcal{T}$-module functors. It is straightforward to see that $\Mod'(\mathcal{T})$ has the structure of a locally additive $\mathbb{C}$-linear dagger 2-category. It also has a direct sum for objects.
\begin{proposition}
The category $\Mod'(\mathcal{T})$ is a locally additive $\mathbb{C}$-linear dagger 2-category. 
\end{proposition}
\begin{proof}
We first observe that $\Mod_{c,s}(\mathcal{C})$ is $\mathbb{C}$-linear. It is well known that the category of $\mathbb{C}$-linear categories, $\mathbb{C}$-linear functors and natural transformations is $\mathbb{C}$-linear; in particular, the linear structure on the set $\Hom(F,G)$ of natural transformations between two $\mathbb{C}$-linear functors on $\mathbb{C}$-linear categories is defined componentwise:
\begin{align*}
(\eta_1 + \eta_2)_X := (\eta_1)_X + (\eta_2)_X 
&&
(\alpha \eta)_X := \alpha \eta_X
&&
(0)_X := 0
\end{align*}
It is easy to check that the morphisms of $\mathcal{C}$-module functors form a subspace of $\Hom(F,G)$. Therefore $\Mod_{c,s}(\mathcal{C})$ is $\mathbb{C}$-linear.

It is straightforward to see that $\Mod_{c,s}(\mathcal{C})$ inherits a dagger structure from $\mathcal{C}$, where the dagger of a natural transformation is defined componentwise:
\begin{align*}
(\eta^{\dagger})_X := (\eta_X)^{\dagger}
\end{align*}

For local additivity, we observe that for any $\mathcal{C}$-module categories $\mathcal{M}_1, \mathcal{M}_2$, $\Hom(\mathcal{M}_1,\mathcal{M}_2)$ has:
\begin{itemize} 
\item A zero object, namely the functor which maps everything in $\mathcal{D}$ to the zero object and zero morphism of $\mathcal{D}$. 
\item A direct sum. The $\mathcal{C}$-module functor $F \oplus G$ is defined as follows:
\begin{align*}
(F \oplus G)(X) := F(X) \oplus G(X) 
&& 
(F \oplus G)(f) :=F(f) \oplus G(f)
&&
c_{F \oplus G,U,X} := c_{F,U,X} \oplus c_{G,U,X}
\end{align*}
The injection morphism $\iota_F: F \to F \oplus G$ is defined componentwise as:
$$
(\iota_F)_X := \iota_X: F(X) \to F(X) \oplus G(X)
$$
where $\iota_X$ is the injection into the direct sum in $\mathcal{C}$. The injection for $\iota_G$ is defined similarly. It is straightforward to check that these injections satisfy the equations for a direct sum. 
\end{itemize} 
\end{proof}
\noindent
Since $\mathcal{M}_1 \oplus \mathcal{M}_2$ is a locally additive 2-category, we can also consider direct sums for objects.
}
\noindent
It is straightforward to define a notion of direct sum for cofinite semisimple $\mathcal{T}$-module categories (see e.g.~\cite[Prop. 7.3.4]{Etingof2016}).
\ignore{
We say that a  
\begin{definition}
Let $\mathcal{M}_1, \mathcal{M}_2$ be semisimple cofinite left $\mathcal{T}$-module categories. We define a left $\mathcal{T}$-module category $\mathcal{M}_1 \oplus \mathcal{M}_2$ as follows: 
\begin{itemize}
\item Objects: formal sums $X= X_1 \oplus X_2$, where $X_i$ is an object of $\mathcal{M}_i$.
\item Morphisms: $\Hom_{\mathcal{M}_1 \oplus \mathcal{M}_2}(X,Y) := \Hom_{M_1}(X_1, Y_1) \oplus \Hom_{M_2}(X_2,Y_2)$
\item $\mathcal{C}$-action: 
\begin{align*}
U \tilde{\otimes} X = U \tilde{\otimes} X_1 \oplus U \tilde{\otimes} X_2 
&&
f \tilde{\otimes} g = f \tilde{\otimes} g_1 \oplus f \tilde{\otimes} g_2
\end{align*}
\item Unitary natural isomorphisms: direct sums of those for the factors. 
\end{itemize}
We also define $\mathcal{T}$-module functors $\iota_i: \mathcal{M}_i \to \mathcal{M}_1 \oplus \mathcal{M}_2$ and $\rho_i: \mathcal{M}_1 \oplus \mathcal{M}_2 \to \mathcal{M}_i$ as follows (we state for $\iota_1, \rho_1$, the others are defined similarly):
\begin{itemize}
\item On objects: $\iota_1(X_1) = X_1 \oplus {\bf 0}$, $\rho_1(X_1 \oplus X_2) = X_1$.
\item On morphisms: $\iota_1(f) = f \oplus 0$, $\rho_1(f_1 \oplus g_1) = f_1$.
\item Natural isomorphism $c_{\iota_1,U,X_1}$: from uniqueness of zero object in $\mathcal{M}_2$.
\item Natural isomorphism $c_{\rho_1,U,X_1 \oplus X_2}$: trivial.
\end{itemize}
A quick check of the conditions of Definition~\ref{} shows that this is indeed a direct sum for objects.
\end{definition}
\noindent
We can now define what we mean by finite decomposability.
\begin{definition}
}
\begin{definition}
We say that a $\mathcal{T}$-module category is \emph{indecomposable} if it is not equivalent to a nontrivial direct sum of $\mathcal{T}$-module categories. We say that a $\mathcal{T}$-module category is \emph{finitely decomposable} if it is equivalent to a finite direct sum of indecomposable module categories.
\end{definition}

\begin{definition}\label{def:modt}
The (strict) 2-category $\Mod(\mathcal{T})$ is defined as follows:
\begin{itemize}
\item \emph{Objects.} Cofinite semisimple finitely decomposable left $\mathcal{T}$-module categories. 
\item \emph{1-morphisms.} Unitary $\mathcal{T}$-module functors. 
\item \emph{2-morphisms.} Morphisms of $\mathcal{T}$-module functors. 
\end{itemize}
\end{definition}
\noindent 
It is straightforward to show that $\Mod(\mathcal{T})$ is an additive $\mathbb{C}$-linear dagger 2-category in the sense of Section~\ref{sec:linearstructure}. Semisimplicity and rigidity will follow from Proposition~\ref{prop:bimodsemisimple} and Theorem~\ref{thm:eilenbergwatts}.

\subsection{Equivalence of the completions}\label{sec:eilwatts}

We now observe that the 2-categories we have defined are equivalent. 
\begin{definition}
We define a unitary $\mathbb{C}$-linear 2-functor  $\Psi: \Bimod(\mathcal{T}) \to \Mod(\mathcal{T})$ as follows:
\begin{itemize}
\item \emph{On objects}: The SSFA $A$ is mapped to its category of right dagger bimodules $\Mod$-$A$, considered as a left $\mathcal{T}$-module category under the following action:
\begin{align*}
U \tilde{\otimes} X_A &:= U \otimes X_A
\\
f \tilde{\otimes} g &:= f \otimes g
\end{align*}
\item \emph{On 1-morphisms}: A dagger bimodule ${}_A M_{B}$ is mapped to the unitary $\mathcal{T}$-module functor $\Mod$-$A \to \Mod$-$B$ given by relative tensor product, i.e. 
\begin{align*}
X_A &\mapsto X \otimes_A M_B
\\
(f: X_A \to Y_A) &\mapsto f \otimes_A \id_{{}_A M_B}
\end{align*}
The unitary natural isomorphism $\{c_{U,X}\}$ is defined using the isometries of the relative tensor product~\eqref{eq:moritaidempotentsplit}.
\item \emph{On 2-morphisms}: A bimodule homomorphism $f: {}_A M_{B} \to {}_A N_B$ is mapped to a natural isomorphism of module functors whose components are as follows: $$(\id_X \otimes_A f): X \otimes_A  M_{B} \to X \otimes_A N_{B}$$
\item \emph{Multiplicator and unitor}: Defined using the associator and right unitor of $\Bimod(\mathcal{T})$.
\end{itemize}
\end{definition}
\begin{remark}
We leave to the reader the straightforward checks that $\Psi$ is indeed a well-defined unitary $\mathbb{C}$-linear 2-functor. It is necessary to show in particular that the left $\mathcal{T}$-module category $\Mod$-$A$ is semisimple and cofinite. Semisimplicity follows from local semisimplicity of $\Bimod(\mathcal{T})$ (Proposition~\ref{prop:bimodsemisimple}). For cofiniteness observe that rigidity of $\Bimod(\mathcal{T})$ implies a linear isomorphism between the vector spaces $\Hom_{\textrm{Mod-A}}(X_A,U_i \otimes Y_A)$ and $\Hom_{\mathcal{T}}(X \otimes_A Y^*,U_i)$; cofiniteness follows by semisimplicity of $\mathcal{T}$. (This was already observed in the proof of~\cite[Thm. 3.2]{Neshveyev2018}.)
\end{remark}
\begin{theorem}\label{thm:eilenbergwatts}
The 2-functor $\Psi$ is an equivalence.
\end{theorem}
\begin{proof}
This result is well-known in the case of a fusion category $\mathcal{T}$ (where there are finitely many simple objects) and a proof was sketched in~\cite[Ex. 3.39]{Chen2021}. The proof in the general case is very similar. For the reader's convenience we provide the proof in full in Appendix~\ref{app:eilenbergwatts}.
\end{proof}
\begin{corollary}
The category $\Mod(\mathcal{T})$ is a semisimple $C^*$-2-category.
\end{corollary}
\noindent
It also follows that $\mathcal{T}$ embeds in $\Mod(\mathcal{T})$ as an endomorphism category, just as for $\Bimod(\mathcal{T})$ (Proposition~\ref{prop:bimodembeds}).
\begin{corollary}\label{cor:embedtmodt}
Let $\End_{\mathcal{T}}(\mathcal{T})$ be the category of endomorphisms of the $\mathcal{T}$-module category $\mathcal{T}$. There is an equivalence of $C^*$-tensor categories $\mathcal{T} \to \End_{\mathcal{T}}(\mathcal{T})$ defined by composing the equivalence $F$ of Proposition~\ref{prop:bimodembeds} with the equivalence $\Psi_{\mathbbm{1},\mathbbm{1}}: \mathbbm{1}$-$\Mod$-$\mathbbm{1} \to \End_{\mathcal{T}}(\mathcal{T})$.
\end{corollary}
\noindent
Finally, we observe that this result allows us to characterise 
semisimple $C^*$-2-categories in general. 
\begin{definition}
We say that a semisimple $C^*$-2-category is \emph{connected} if the Hom-category between any pair of nonzero objects is not the terminal category.
\end{definition}
\begin{proposition}[{\cite[Lem. 2.2.3]{Giorgetti2020}}]\label{prop:reconstruction}
For any object $r$ of a connected semisimple $C^*$-2-category $\mathcal{C}$, there is an equivalence $ \mathcal{C} \simeq \Bimod(\End(r))$.
\end{proposition}
\begin{proof}
The equivalence $\Delta: \mathcal{C} \overset{\sim}{\to} \Bimod(\End(r))$ is defined as follows.
\begin{itemize}
\item \emph{On objects.} For every nonzero object $s$ of $\mathcal{C}$, pick a separable 1-morphism $P_s: r \to s$ in $\mathcal{C}$. Then define $\Delta(s) := P_s \otimes (P_s)^*$, where $P_s \otimes (P_s)^*$ is the pair of pants $\F$ in $\End(r)$ corresponding to $P_s$.
\item \emph{On 1-morphisms.} For every 1-morphism $X: s \to t$ in $\mathcal{C}$, define $\Delta(X) := P_s \otimes X \otimes (P_t)^*$, which is an $\Delta(s)$-$\Delta(t)$ dagger bimodule with the following action (here and throughout the proof we leave regions corresponding to the object $r$ unshaded, we shade regions corresponding to the object $s$ with wavy lines, and we shade regions corresponding to the object $t$ with polka dots):
\begin{calign}
\includegraphics[scale=.9]{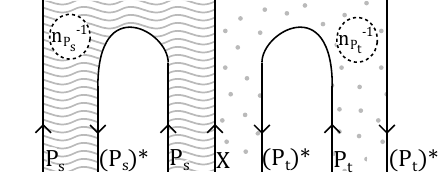}
\end{calign}
\item \emph{On 2-morphisms.} For every 2-morphism $f: X \to Y$ in $\mathcal{C}$, we define $\Delta(f):= \id_{P_s} \otimes f \otimes \id_{(P_t)^*}$.
\item \emph{Multiplicator.} Let $X: s \to t$, $Y: t \to u$ be 1-morphisms in $\mathcal{C}$. Then we define the multiplicator component $\mu_{X,Y}: \Delta(X) \otimes \Delta(Y) \to \Delta(X \otimes Y)$ as the following bimodule homomorphism (we shade regions corresponding to the object $u$ with a checkerboard effect):
\begin{calign}
\includegraphics[scale=.9]{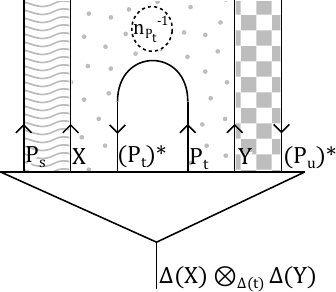}
\end{calign}
\item \emph{Unitor.} Trivial (up to unitors/associators in $\mathcal{C}$). 
\end{itemize}
It is straightforward to check that $\Delta$ is a $\mathbb{C}$-linear unitary 2-functor. It was shown in~\cite[Lem. 2.2.3]{Giorgetti2020} that it is furthermore a local equivalence. The only additional thing we must prove is essential surjectivity on objects. For any $\F$ $A$ in $\End(r)$, we must show that there exists an object $s$ of $\mathcal{C}$ such that $\Delta(s)$ is Morita equivalent to $A$.  Since $\mathcal{C}$ is semisimple, there certainly exists an object $s$ and a 1-morphism $X: r \to s$ such that $A \cong (X \otimes X^*)$. Now $\Delta(s) = P_s \otimes (P_s)^*$, where $P_s: r \to s$ is the 1-morphism chosen in the definition of the pseudofunctor $\Delta$. We claim that $X \otimes X^*$ is Morita equivalent to $P_s \otimes (P_s)^*$. This follows from~\cite[Thm. A.1]{Musto2019}, which implies that two 1-morphisms $r \to s, r \to t$ produce Morita equivalent $\F$s in $\End(r)$ iff $s$ and $t$ are equivalent objects in $\mathcal{C}$.\footnote{To be precise, the cited theorem classifies morphisms \emph{into} $r$ rather than morphisms out of $r$. However, the proof works equally well for morphisms out of $r$; just read the diagrams from left to right.}
\end{proof}
\begin{corollary}\label{cor:modtsemisimpleetc.}
Every connected semisimple $C^*$-2-category is equivalent to $\Mod(\mathcal{T})$ for some rigid $C^*$-tensor category $\mathcal{T}$.
\end{corollary}

\section{A covariant Stinespring theorem}\label{sec:stinespring}

In the last section we defined two equivalent semisimple 2-categories in which a rigid $C^*$-tensor category $\mathcal{T}$ embeds as the endomorphisms of a fixed object. We will now apply this to the study of finite-dimensional $G$-$C^*$-algebras and covariant completely positive maps.

\subsection{A classification of finite-dimensional $G$-$C^*$-algebras}\label{sec:gc*alg}

We will first briefly recall how finite-dimensional $G$-$C^*$-algebras (a.k.a. $C^*$-dynamical systems) for a compact quantum group $G$ may be identified with $\F$s in the category $\Rep(G)$ of finite-dimensional continuous unitary representations of $G$. This characterisation already appeared in~\cite{Neshveyev2018}; other relevant works include~\cite{Vicary2011,Bischoff2015,Banica1999}. See~\cite[\S{}3]{Verdon2020b} for a more thorough summary.

We consider first of all the familiar notion of a finite-dimensional $G$-$C^*$-algebra, or $C^*$-dynamical system, for an ordinary compact group $G$. Let $A$ be a finite-dimensional $C^*$-algebra. An \emph{action} of a compact group $G$ on $A$ is a continuous homomorphism $\tau: G \to \Aut(A)$, where $\Aut(A)$ is the group of $*$-automorphisms of $A$. For any such action there is a canonical invariant trace $\phi: A \to \mathbb{C}$ which is preserved under the $G$-action in the sense that $\phi(\tau(g)(x)) = \phi(x)$ for all $x \in A$. 

The canonical invariant trace $\phi$ induces an inner product $\braket{x|y} := \phi(x^*y)$ on the finite-dimensional complex vector space underlying the $C^*$-algebra $A$, which thus acquires the structure of a Hilbert space. It is not hard to show (see e.g.~\cite[\S{}3.1]{Verdon2020b}) that the $C^*$-algebra structure of $A$ further induces the structure of an $\F$ on the Hilbert space $A$; the multiplication and unit of the $\F$ are precisely the multiplication and unit of the $C^*$-algebra, and the counit of the $\F$ is the trace $\phi$. Since $\phi$ is preserved under the action $\tau$ of $G$, this action induces a continuous unitary representation of $G$ on the Hilbert space $A$ such that the structure morphisms of the $\F$ --- the multiplication $m: A \otimes A \to A$ and the unit $u: \mathbb{C} \to A$ --- are intertwiners. From a $G$-$C^*$-algebra for a compact group $G$ we have therefore constructed an $\F$ in the rigid $C^*$-tensor category $\Rep(G)$ of finite-dimensional continuous unitary representations of $G$. In the other direction, every $\F$ in $\Rep(G)$ has a natural involution such that the resulting $*$-algebra is a $G$-$C^*$-algebra. These constructions are inverse and set up a bijective correspondence between $*$-isomorphism classes of $G$-$C^*$-algebras and unitary $*$-isomorphism classes of $\F$s in $\Rep(G)$. An $\F$ in $\Rep(G)$ is therefore a $G$-$C^*$-algebra equipped with its canonical invariant trace. 

We extend the notion of symmetry by generalising from representation categories of compact groups to rigid $C^*$-tensor categories $\mathcal{T}$ with simple unit object equipped with a faithful unitary linear functor $F: \mathcal{T} \to \Hilb$, called a \emph{fibre functor}. By Tannaka-Krein-Woronowicz (T-K-W) duality~\cite[Thm. 2.3.2]{Neshveyev2013}, these are precisely the categories $\Rep(G)$ of finite-dimensional continuous unitary representations of \emph{compact quantum groups} $G$, equipped with their canonical fibre functor.\footnote{For a very elementary algebraic perspective on the definition of a compact quantum group and its finite-dimensional unitary representation theory, see~\cite[\S{}2.1.4]{Verdon2020b} (which borrows heavily from the more sophisticated presentations in~\cite{Timmerman2008,Neshveyev2013}). Here we will not even need to define a compact quantum group, since we already know what a rigid $C^*$-tensor category with fibre functor is.} We define a finite-dimensional $G$-$C^*$-algebra for a compact quantum group $G$ to be an $\F$ in the rigid $C^*$-tensor category $\Rep(G)$.\footnote{There are other, less abstract definitions of a $G$-$C^*$-algebra, e.g.~\cite{Wang1998}; this definition is equivalent, as was observed in~\cite{Banica1999,Neshveyev2018}. Indeed, given an $\F$ $A$, using the canonical fibre functor $F: \Rep(G) \to \Hilb$ one recovers (by T-K-W duality) a coaction of the Hopf $*$-algebra $A_G$ associated to $G$ on the f.d. $C^*$-algebra $F(A)$~\cite[Prop. 3.2.4]{Verdon2020b}.} To recover a concrete $C^*$-algebra from such an $\F$ $A$ one considers the object $F(A)$, which is a Hilbert space with the structure of a separable Frobenius algebra; this algebra possesses a natural involution with a $C^*$-norm. We remark that, in this more general case, the canonical invariant functional on $A$ (that is, the counit of the $\F$) may not be tracial as a concrete functional on the $G$-$C^*$-algebra $F(A)$; in fact, it is only tracial when $d(A) = \dim(F(A))$~\cite[Thm. 5.3]{Verdon2020b}. This causes no problems, provided one is content to move from completely positive trace-preserving maps to completely positive functional-preserving maps in this more general setting.

Finally, we could generalise still further and consider $\F$s in general rigid $C^*$-tensor categories. In this case there is no obvious way to identify these $\F$s with concrete $C^*$-algebras, since they do not have an associated vector space in general; however, the results we are about to obtain apply at this level of generality.

Before considering channels, we will draw one straightforward consequence of what has already been proven: namely, a classification of finite-dimensional $G$-$C^*$-algebras for a compact quantum group $G$, which amounts to a classification of $\F$s in $\Rep(G)$. In fact, we will work in the most general setting and classify $\F$s in $\mathcal{T}$ for a rigid $C^*$-tensor category $\mathcal{T}$, whether or not a fibre functor exists. This classification is certainly not new (see e.g.~\cite{DeCommer2012,Neshveyev2013a,Neshveyev2018}), although our version seems to offer some additional precision regarding the equivalence relation on objects of a module category corresponding to unitary $*$-isomorphism of the associated $\F$s.

\begin{definition}
We say that $\F$s in $\mathcal{T}$ are \emph{Morita equivalent} if they are equivalent as objects of $\Bimod(\mathcal{T})$.
\end{definition}
\begin{lemma}\label{lem:pairofpants}
Let $A, \tilde{A}$ be $\F$s in $\mathcal{T}$. Following Proposition~\ref{prop:bimodembeds}, we embed $\mathcal{T} \cong \End(\mathbbm{1})$ in $\Bimod(\mathcal{T})$. Then $A$ is Morita equivalent to $\tilde{A}$ if and only if $A$ is unitarily $*$-isomorphic to a pair of pants algebra $M \otimes M^*$ for some separable 1-morphism $M: \mathbbm{1} \to \tilde{A}$ in $\Bimod(\mathcal{T})$.
\end{lemma}
\begin{proof}
\emph{Only if.} We will show that from a Morita equivalence $A \simeq \tilde{A}$ we can construct a satisfactory 1-morphism $M: \mathbbm{1} \to \tilde{A}$. In the following diagrams we leave the regions corresponding to the object $\mathbbm{1}$ unshaded; we shade the regions corresponding to the object $A$ with wavy lines; and we shade the regions corresponding to the object $\tilde{A}$ with polka dots. 

Let $[E,E^{-1},\alpha_E,\beta_E]$ be the data associated to an adjoint equivalence $E: A \to \tilde{A}$ in $\Bimod(\mathcal{T})$. We will first consider the relationship between the right dual $[E^{-1},\alpha_E,\beta_E]$ for $E$ and the standard right dual $[E^*,\eta_E,\epsilon_E]$. Let $v: E^* \to E^{-1}$ be the isomorphism relating the right duals $E^*$ and $E^{-1}$ by Proposition~\ref{prop:relateduals}, i.e.:
\begin{calign}
\includegraphics[scale=.8]{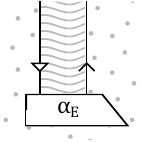}
~~=~~
\includegraphics[scale=.8]{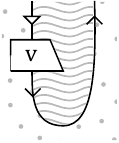}
&&
\includegraphics[scale=.8]{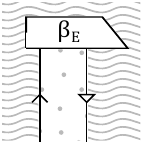}
~~=~~
\includegraphics[scale=.8]{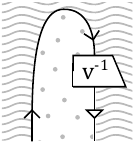}
\end{calign}
(Here we drew the $E^{-1}$ wire with a triangular downwards-pointing arrow.) 
We first observe that
\begin{equation}
v^{\dagger} = v^{-1} \otimes \dim_R(E),
\end{equation}
which can be seen by the following equation:
\begin{calign}\nonumber
\includegraphics[scale=.8]{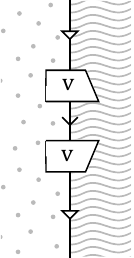}
~~=~~
\includegraphics[scale=.8]{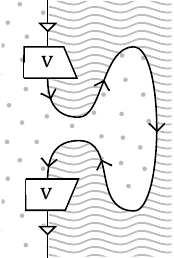}
~~=~~
\includegraphics[scale=.8]{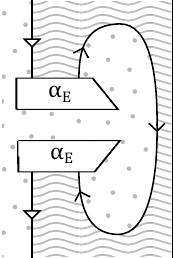}
~~=~~
\includegraphics[scale=.8]{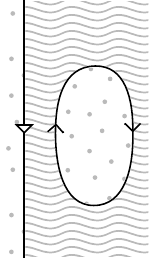}
\end{calign}
In the same way it can be shown that 
\begin{equation}
\label{eq:vinvdag}
(v^{-1})^{\dagger} = \dim_L(E) \otimes v,
\end{equation} and therefore that $\dim_L(E) \otimes \id_{E^{-1}} \otimes \dim_R(E) = \id_{E^{-1}}$. It follows that $E$ is a separable 1-morphism. We also make the following further observation for later:
\begin{calign}\label{eq:atildepipe}
\includegraphics[scale=.8]{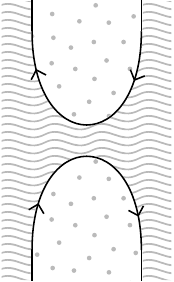}
~~=~~
\includegraphics[scale=.8]{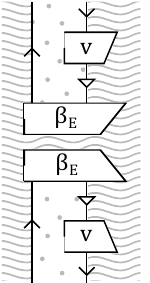}
~~=~~
\includegraphics[scale=.8]{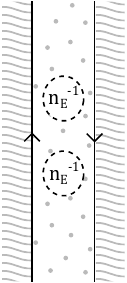}
\end{calign}
Let ${}_{\mathbbm{1}} A_A: \mathbbm{1} \to A$ be the 1-morphism in $\Bimod(\mathcal{C})$ corresponding to $A$ considered as a $\mathbbm{1}$-$A$ bimodule, and let ${}_A A_{\mathbbm{1}}: A \to \mathbbm{1}$ be the 1-morphism corresponding to $A$ considered as an $A$-$\mathbbm{1}$ bimodule. We will show that ${}_A A_{\mathbbm{1}}$ is a standard right dual for ${}_{\mathbbm{1}} A_A$ with a cup and cap we will now define. 

Recall from Lemma~\ref{lem:decomposingfs} that $A \cong \oplus_k A_k$, where $A_k$ are simple $\F$s, and let $\iota_k: A_k \to A$ be the isometric injections. In the following diagrams we draw the $\iota_k$ as downward-pointing triangles, the $\iota_k^{\dagger}$ as upward-pointing triangles, and the structure morphisms of the algebra $A_k$ as white circles with a $k$ next to them.  We define the cup and cap $\eta$, $\epsilon$ as follows (on the LHS of the following definitions is the cup/cap as it appears in $\Bimod(\mathcal{T})$, and on the RHS is the definition as a concrete bimodule homomorphism in $\mathcal{T}$):
\begin{calign}\label{eq:acupcap}
\includegraphics[scale=1]{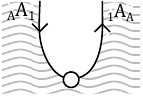}
~~:=~~
\sum_k d(A_k)^{1/4}
\includegraphics[scale=1]{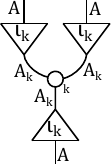}
&&
\includegraphics[scale=1]{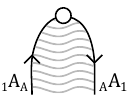}
~~:=~~
\sum_k d(A_k)^{-1/4}
\includegraphics[scale=1]{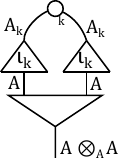}
\end{calign}
It is straightforward to check by considering the underlying morphisms in $\mathcal{T}$ that this cup and cap obey the snake equations~\eqref{eq:snake}. We will now show that the duality is standard. Let $T \in \End({}_{\mathbbm{1}} A_A)$. Clearly $\epsilon \circ (T \otimes \id_{{}_A A_{\mathbbm{1}}}) \circ \epsilon^{\dagger}$ is the following concrete morphism in $\mathcal{T}$:
\begin{calign}
\sum_k d(A_k)^{-1/2}~~
\includegraphics[scale=1]{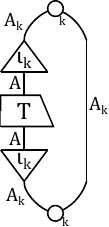}
\end{calign}
In the notation of Proposition~\ref{prop:standardintrinsic}, we therefore see that $\psi_{{}_\mathbbm{1} A_A}(T) = \sum_k d(A_k)^{-1/2} \psi_{A_k}[\iota_k^{\dagger} \circ  T \circ  \iota_k]$, where $\psi_{A_k}: \End(A_k) \to \mathbb{C}$ is the trace defined by the standard duality on $A_k$ in $\mathcal{T}$.

On the other hand, it is also clear that $\tilde{\eta}^{\dagger} \circ (\id_{{}_A A_{\mathbbm{1}}} \otimes T) \circ \tilde{\eta}$ is the following concrete morphism in $\mathcal{T}$:
\begin{calign}
\sum_k d(A_k)^{1/2}~
\includegraphics[scale=1]{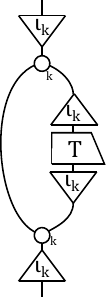}
\end{calign}
Since the $A_k$ are simple, as an $A_k$-$A_k$ bimodule endomorphism, $m_k \circ (\id_{A_k} \otimes (\iota_k^{\dagger} \circ T \circ \iota_k)) \circ m_k^{\dagger} = \alpha_{k,T} \id_{A_k}$ for some scalars $\alpha_{k,T}$, so, in the notation of Proposition~\ref{prop:standardintrinsic}, $\phi_{{}_\mathbbm{1} A_A}(T) =  \sum_k d(A_k)^{1/2} \alpha_{k,T}$. But now, taking the trace of $\alpha_{k,T} \id_{A_k}$ in $\mathcal{T}$, we obtain the following equation:
\begin{calign}
\alpha_{k,T} \cdot d(A_k) 
~~=~~
\includegraphics[scale=1]{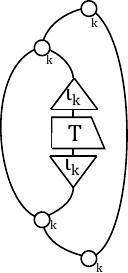}
~~=~~
\includegraphics[scale=1]{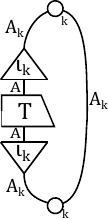}
~~=~~
\psi_{A_k}[\iota_k^{\dagger} \circ  T \circ  \iota_k].
\end{calign} 
Here for the second equality we used separability and standardness of $A_k$. It follows that 
$$\phi_{{}_\mathbbm{1} A_A}(T) = \sum_{k} d(A_k)^{1/2} \alpha_{k,T} = \sum_{k} d(A_k)^{-1/2} \psi_{A_k}[\iota_k^{\dagger} \circ T \circ \iota_k] = \psi_{{}_\mathbbm{1} A_A}(T),$$
so the right dual $[{}_A A_{\mathbbm{1}}, \eta, \epsilon]$ is indeed standard. We observe in particular that $\dim_L({}_\mathbbm{1} A_A) = \sum_k d(A_k)^{1/2} (\iota_k \circ \iota_k^{\dagger})$; ${}_\mathbbm{1} A_A$ is therefore a separable 1-morphism, with $n_{{}_\mathbbm{1} A_A} = \sum_k d(A_k)^{1/4} (\iota_k \circ \iota_k^{\dagger})$.

Now we set $M:= {}_{\mathbbm{1}}A_A \otimes E : \mathbbm{1} \to \tilde{A}$. Using the fact that the tensor product of standard duals (Proposition~\ref{prop:nestedduals}) is standard it is easy to see that
$$n_M = n_E \otimes (\alpha_E^{\dagger} \circ (\id_{E^{-1}} \otimes n_{{}_\mathbbm{1} A_A} \otimes \id_E) \circ \alpha_E).$$
It follows that $M$ is a separable 1-morphism.  

We claim that there is a unitary $*$-isomorphism between $A$ and the pair of pants $\F$ $M \otimes M^*$. It does not matter which standard right dual $M^*$ we pick in defining the pair of pants algebra $M \otimes M^*$ --- they will all produce unitarily $*$-isomorphic $\F$s --- so we pick the tensor product dual $E^* \otimes {}_A A_{\mathbbm{1}}$. This yields an $\F$ on the object $({}_{\mathbbm{1}}A_A \otimes E) \otimes (E^* \otimes {}_A A_{\mathbbm{1}})$ with the following multiplication and unit:
\begin{calign}
\includegraphics[scale=.8]{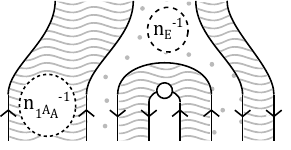}
&&
\includegraphics[scale=.8]{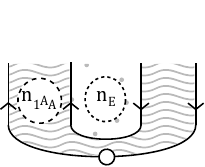}
\end{calign}
We will show that the following map $f: ({}_{\mathbbm{1}}A_A \otimes E) \otimes (E^* \otimes {}_{A}A_{\mathbbm{1}}) \to A$  is a unitary $*$-isomorphism:
\begin{calign}
~~\includegraphics[scale=1]{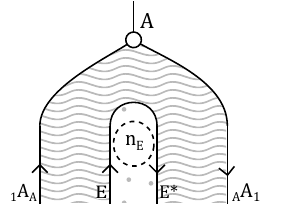}
\end{calign}
Here the 2-morphism $\tilde{m}: {}_1 A_A \otimes {}_A A_1 \to A$ represented by a white circle is concretely the following morphism in $\mathcal{T}$, where this time the white circle represents the multiplication $m: A \otimes A \to A$ of the Frobenius algebra $A$, as usual:
\begin{calign}
\includegraphics[scale=1]{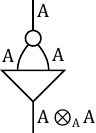}
\end{calign}
It is easy to check that $\tilde{m}$ is unitary by separability of the Frobenius algebra, the Frobenius equation and~\eqref{eq:moritaidempotentsplit}.

We first need to show that $f$ is a $*$-homomorphism. We will begin with multiplicativity (the first equation of~\eqref{eq:homo}). We will need the following equation in $\Bimod(\mathcal{T})$, which can be straightforwardly checked by considering the underlying morphisms in $\mathcal{T}$:
\begin{calign}\label{eq:adualspullround}
\includegraphics[scale=1]{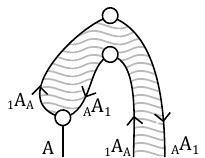}
~~:=~~
\includegraphics[scale=1]{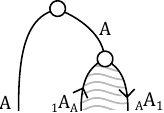}
\end{calign}
We also need the following equation, which can again be straightforwardly checked by considering the underlying morphisms in $\mathcal{T}$:
\begin{calign}\label{eq:amultcap}
\includegraphics[scale=1]{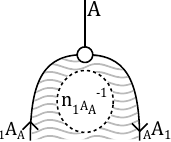}
~~:=~~
\includegraphics[scale=1]{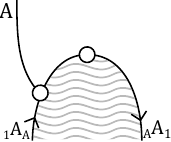}
\end{calign}
(Here the morphism ${}_{\mathbbm{1}} A_A \to A \otimes  A_A $ represented by a white circle is concretely just the comultiplication of the Frobenius algebra $A$.)

Now we prove multiplicativity:
\begin{calign}
\includegraphics[scale=.7]{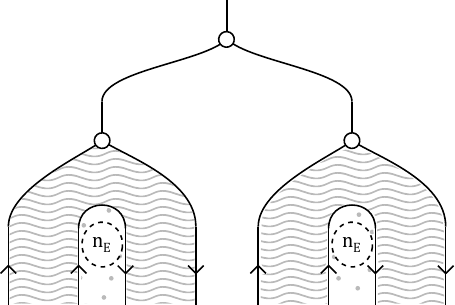}
~~=~~
\includegraphics[scale=.7]{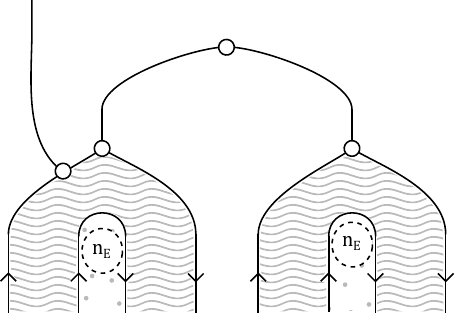}
\\=~~
\includegraphics[scale=.7]{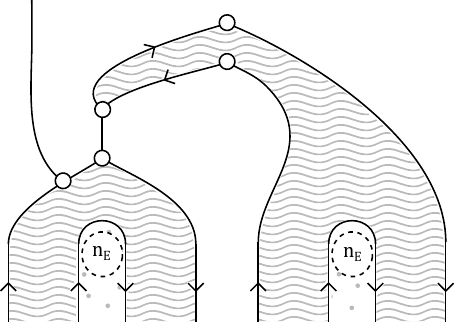}
~~=~~
\includegraphics[scale=.7]{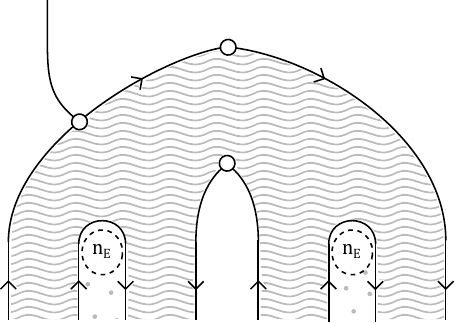}
\\
=~~
\includegraphics[scale=.7]{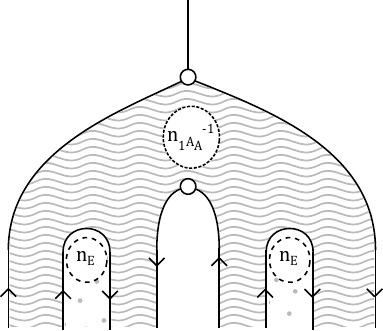}
~~=~~
\includegraphics[scale=.7]{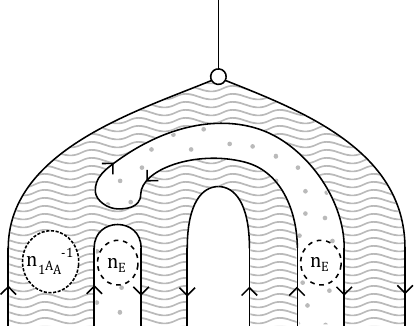}
\\=~~
\includegraphics[scale=.7]{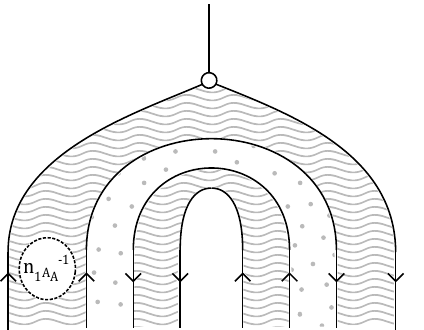}
\end{calign}
Here the first equality is clear; the second equality is by~\eqref{eq:adualspullround}; the third equality is by unitarity of $\tilde{m}$; the fourth equality is by~\eqref{eq:amultcap}; the fifth equality is by the pivotal dagger structure on $\Bimod(\mathcal{T})$; and the sixth equality is by~\eqref{eq:atildepipe}.

Unitality (the second equation of~\eqref{eq:homo}) is shown by the following equalities:
\begin{calign}\label{eq:funitality}
\includegraphics[scale=.7]{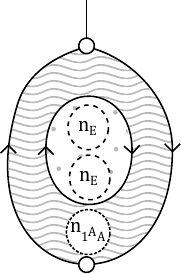}
~~=~~
\includegraphics[scale=.7]{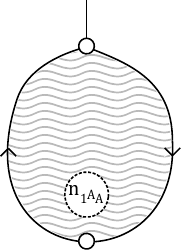}
~~=~~
\includegraphics[scale=.7]{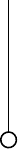}
\end{calign}
Here the first equality can be seen by inserting $v^{-1} \circ v$ on the $E^*$-wire and using~\eqref{eq:vinvdag}. The second equality can straightforwardly be seen by considering the underlying morphisms in $\mathcal{T}$. 

For $*$-preservation (the third equation of~\eqref{eq:homo}), we have the following equalities:
\begin{calign}
\includegraphics[scale=.7]{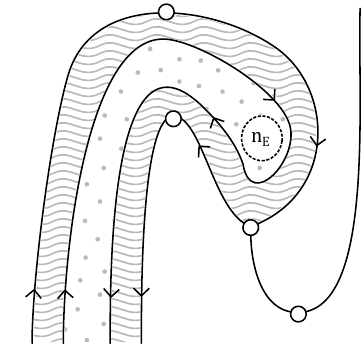}
~~=~~
\includegraphics[scale=.7]{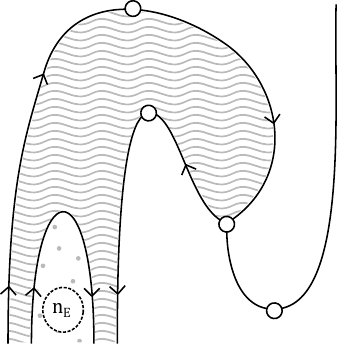}
\\
=~~
\includegraphics[scale=.7]{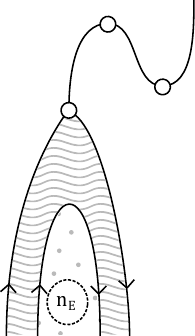}
~~=~~
\includegraphics[scale=.7]{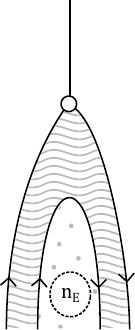}
\end{calign}
Here the first equality is by the pivotal dagger structure on $\Bimod(\mathcal{T})$; the second equality is by~\eqref{eq:adualspullround} (or more precisely,~\eqref{eq:adualspullround} precomposed by $\tilde{m} \otimes \tilde{m}^{\dagger}$); and the final equality is by a snake equation for $A$. 

We have shown that $f$ is a $*$-homomorphism. Now we need only show that $f$ is unitary:
\begin{calign}
\includegraphics[scale=.7]{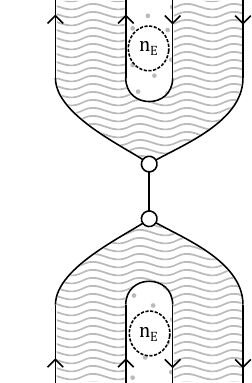}
~~=~~
\includegraphics[scale=.7]{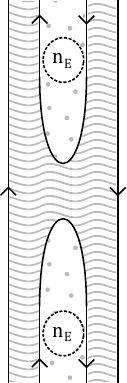}
~~=~~
\includegraphics[scale=.7]{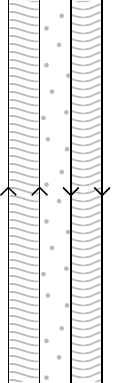}
\\
\includegraphics[scale=.7]{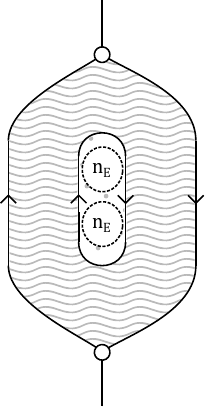}
~~=~~
\includegraphics[scale=.7]{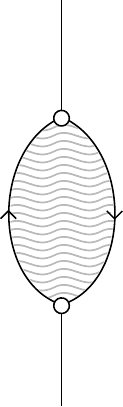}
~~=~~
\includegraphics[scale=.7]{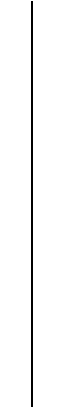}
\end{calign}
Here in the first line the first equality is by unitarity of $\tilde{m}$ and the second equality is by~\eqref{eq:atildepipe}; in the second line the first equality is by the same argument as in~\eqref{eq:funitality} and the second equality is by unitarity of $\tilde{m}$. The proof of the `only if' direction is complete.

\emph{If.} We now show the opposite implication: if there exists a separable 1-morphism ${}_{\mathbbm{1}} M_{\tilde{A}}: \mathbbm{1} \to \tilde{A}$ and a unitary $*$-isomorphism $f: {}_{\mathbbm{1}}M_{\tilde{A}} \otimes ({}_{\mathbbm{1}}M_{\tilde{A}})^* \to A$, then $A$ and $\tilde{A}$ are Morita equivalent. 

We first observe that the right $\tilde{A}$-dagger module $M_{\tilde{A}}$ is in fact an $A$-$\tilde{A}$-dagger bimodule by the following left $A$-action:
\begin{calign}
\includegraphics[scale=.7]{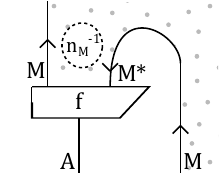}
\end{calign}
We therefore obtain a 1-morphism ${}_A M_{\tilde{A}}: A \to \tilde{A}$ in $\Bimod(\mathcal{T})$. We will show that it is an equivalence with weak inverse $({}_A M_{\tilde{A}})^*$, proving that $A$ and $\tilde{A}$ are Morita equivalent. For this we need to produce unitary 2-morphisms ${}_A A_A \to {}_A M_{\tilde{A}} ~\otimes ({}_A M_{\tilde{A}})^*$ and ${}_{\tilde{A}} \tilde{A}_{\tilde{A}} \to ({}_A M_{\tilde{A}})^* \otimes {}_A M_{\tilde{A}}$.

The following equalities show that $f^{\dagger}: A \to {}_{\mathbbm{1}}M_{\tilde{A}} \otimes ({}_{\mathbbm{1}} M_{\tilde{A}})^*$ is in fact an $A$-$A$ bimodule morphism:
\begin{calign}
\includegraphics[scale=.7]{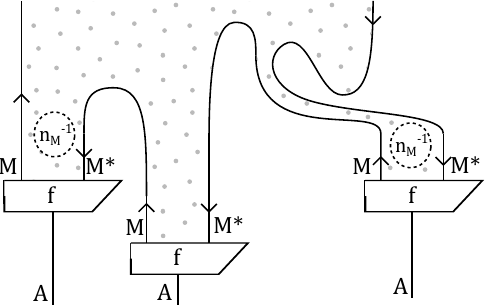}
~~=~~
\includegraphics[scale=.7]{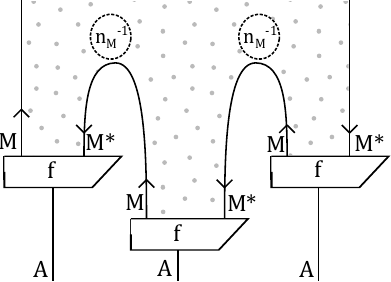}
~~=~~
\includegraphics[scale=.7]{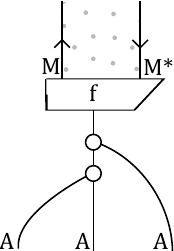}
\end{calign}
Here in the first diagram the right $A$-action on $({}_{\mathbbm{1}} M_{\tilde{A}})^*$ is defined in terms of the left $A$-action on ${}_{\mathbbm{1}} M_{\tilde{A}}$ as in~\eqref{eq:standardpivmod}, using the fact, communicated in Remark~\ref{rem:standarddualbimod}, that ${}_{\tilde{A}}(M^*)_{A}$ is a standard right dual bimodule for ${}_{A}M_{\tilde{A}}$. For the first equality we used isotopy of the diagram, and for the second equality we used the first $*$-homomorphism condition~\eqref{eq:homo}.

We have therefore found the first desired unitary 2-morphism,  $f^{\dagger}: {}_A A_A \to {}_A M_{\tilde{A}} ~\otimes ({}_A M_{\tilde{A}})^*$. We will now obtain the second. Let $x:=\sqrt{\dim_L({}_A M_{\tilde{A}})} \in \End(\id_{{}_{\tilde{A}} \tilde{A}_{\tilde{A}}})$, and let $\epsilon: {}_A M_{\tilde{A}} ~\otimes ({}_A M_{\tilde{A}})^* \to {}_A A_A $ be the cap of the standard duality in $\Bimod(\mathcal{T})$. The following equations show that the 2-morphism $\epsilon \circ (\id_{{}_A M_{\tilde{A}}} \otimes x \otimes \id_{({}_A M_{\tilde{A}})^*})$ is unitary:
\begin{calign}\label{eq:epsilonunitary1}
\includegraphics[scale=.7]{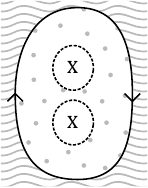}
~~=~~
\includegraphics[scale=.7]{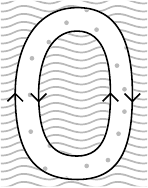}
~~=~~
\includegraphics[scale=.7]{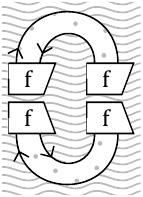}
~~=~~
\includegraphics[scale=.7]{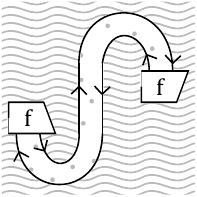}
~~=~~
\includegraphics[scale=.7]{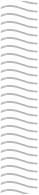}
\\\label{eq:epsilonunitary2}
\includegraphics[scale=.7]{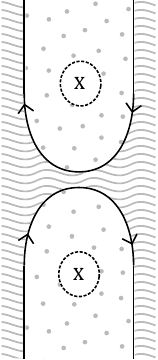}
~~=~~
\includegraphics[scale=.7]{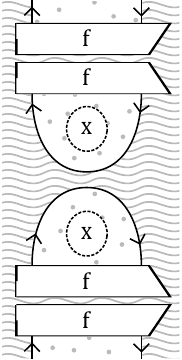}
~~=~~
\includegraphics[scale=.7]{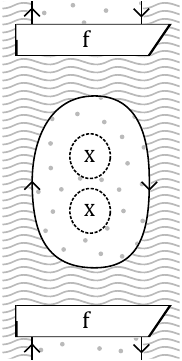}
~~=~~
\includegraphics[scale=.7]{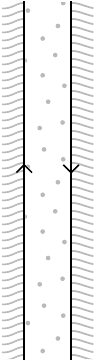}
\end{calign}
Here for the last equality of~\eqref{eq:epsilonunitary2} we used~\eqref{eq:epsilonunitary1}. 

We now show that $\dim_L({}_AM_{\tilde{A}})$, and therefore also $x$, is invertible. Indeed, by assumption, $\dim_L({}_{\mathbbm{1}} M_{\tilde{A}}) = \dim_L({}_\mathbbm{1} A_{A} \otimes {}_AM_{\tilde{A}})$ is invertible. By Remark~\ref{rem:matrixstandardduals}, up to permutation of the factors we have the following expressions for left dimensions in the commutative $C^*$-algebra $\End(\id_{\tilde{A}})$:
\begin{align*}
\dim_L({}_\mathbbm{1} A_{A} \otimes {}_A M_{\tilde{A}}) &= [\sum_i d(A_i)^{1/2} d(\tilde{M}_{i1}), \dots, \sum_i d(A_i)^{1/2} d(\tilde{M}_{in_{\tau}})]\\
\dim_L({}_A M_{\tilde{A}}) &= [\sum_i d(\tilde{M}_{i1}), \dots, \sum_i d(\tilde{M}_{in_{\tau}})]
\end{align*}
Here $\tilde{M}: \vec{\sigma} \to \vec{\tau}$ is the matrix of 1-morphisms corresponding to ${}_AM_{\tilde{A}}$ under the equivalence $\Phi: \Mat(\Bimod(\mathcal{T})) \overset{\sim}{\to} \Bimod(\mathcal{T})$. Since all the $d(A_i)$ are nonzero, an entry in the vector $\dim_L({}_A M_{\tilde{A}})$ can be zero only if the corresponding entry in the vector $\dim_L({}_\mathbbm{1} A_{A} \otimes {}_AM_{\tilde{A}})$ is zero; therefore invertibility of $\dim_L({}_{\mathbbm{1}} M_{\tilde{A}})$ implies invertibility of $\dim_L({}_AM_{\tilde{A}})$.

Let $\eta: {}_{\tilde{A}} \tilde{A}_{\tilde{A}} \to ({}_A M_{\tilde{A}})^* \otimes {}_A M_{\tilde{A}}$ be the cup of the standard duality in $\Bimod(\mathcal{T})$. We will now show that $\eta \circ x^{-1}$ is a unitary 2-morphism ${}_{\tilde{A}} \tilde{A}_{\tilde{A}} \to ({}_A M_{\tilde{A}})^* \otimes {}_A M_{\tilde{A}}$, finishing the proof:
\begin{calign}\label{eq:etaunitary1}
\includegraphics[scale=.7]{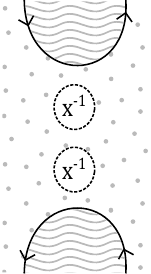}
~~=~~
\includegraphics[scale=.7]{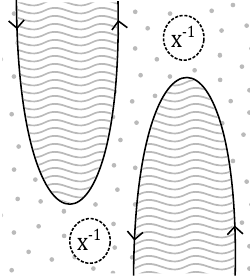}
~~=~~
\includegraphics[scale=.7]{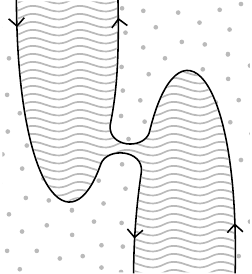}
~~=~~
\includegraphics[scale=.7]{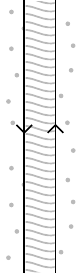}
\\
\includegraphics[scale=.7]{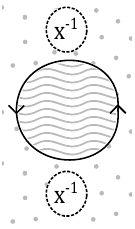}
~~=~~
\includegraphics[scale=.7]{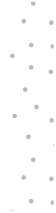}
\end{calign}
Here for the second equality of~\eqref{eq:etaunitary1} we used~\eqref{eq:epsilonunitary2}.
\end{proof}
\noindent
In Lemma~\ref{lem:decomposingfs} we showed that every $\F$ in $\mathcal{T}$ is a direct sum of simple $\F$s. We therefore need only classify the simple $\F$s.  In Lemma~\ref{lem:pairofpants} we showed that two $\F$s in $\mathcal{T}$ are Morita equivalent precisely when one can be expressed as a pair of pants algebra over the other in $\Bimod(\mathcal{T})$. By Theorem~\ref{thm:eilenbergwatts} we can rephrase this in terms of  $\Mod(\mathcal{T})$. We also observe that all nonzero 1-morphisms between simple objects of $\Mod(\mathcal{T})$ are separable. We may therefore construct all simple $\F$s in $\mathcal{T}$ as follows:
\begin{itemize}
\item Obtain representatives $\{\mathcal{M}_i\}_{i \in I}$ of equivalence classes of semisimple cofinite indecomposable left $\mathcal{T}$-module categories. 
\item For each $i \in I$ and for each unitary $\mathcal{T}$-module functor $X: \mathcal{T} \to \mathcal{M}_i$ in $\Bimod(\mathcal{T})$, construct the pair of pants $\F$ $X \otimes X^*$ in $\End_{\mathcal{T}}(\mathcal{T}) \simeq \mathcal{T}$.
\end{itemize}
To turn this into a classification, we need to determine when two 1-morphisms $\mathcal{C} \to \mathcal{M}_i$ in $\Mod(\mathcal{T})$ give rise to the same $\F$.

It is clear that certain nonisomorphic 1-morphisms will give rise to unitarily $*$-isomorphic $\F$s. For instance, set $\mathcal{T} = \Rep(G)$ for some ordinary compact group $G$, and let $\theta$ be a nontrivial one-dimensional representation. For any other  representation $X$, clearly $X$ and $X \otimes \theta$ are nonisomorphic objects in $\End_{\mathcal{T}}(\mathcal{T}) \simeq \Rep(G)$. However, since $\theta \otimes \theta^* \cong \mathbbm{1}$, there is a unitary $*$-isomorphism $X \otimes X^* \cong (X \otimes \theta) \otimes (X \otimes \theta)^*$. In fact, we will now see that this is all that can go wrong; two 1-morphisms will produce the same $\F$ if and only if they are `equivalent up to a phase' in this way. For this we use the following theorem. 
\begin{definition}
We say that two 1-morphisms $X: r \to s$ and $Y: r \to t$ in a dagger 2-category are \emph{equivalent} when there exists an equivalence $E: t \to s$ and a unitary 2-morphism $\tau: X \to Y \otimes E$.
\end{definition}
\begin{theorem}[{\cite[Thm. 5.7]{Verdon2020a}}]\label{thm:pairofpantsisoclassification}
Let $\mathcal{C}$ be a $\mathbb{C}$-linear pivotal dagger 2-category with split dagger idempotents. Let $s, t$ be simple objects, and let $X: r \to s$ and $Y: r \to t$ be 1-morphisms. Then $X$ and $Y$ are equivalent in $\mathcal{C}$ if and only if the separable Frobenius algebras $X \otimes X^*$ and $Y \otimes Y^*$ in $\End(r)$ are unitarily $*$-isomorphic.
\end{theorem}
\noindent
Applying Theorem~\ref{thm:pairofpantsisoclassification} in $\Mod(\mathcal{T})$, we obtain the following classification of simple $\F$s in a rigid $C^*$-tensor category. 
\begin{definition}
Let $\mathcal{T}$ be a rigid $C^*$-tensor category. We say that an object $\theta$ in $\mathcal{T}$ is a \emph{phase} if $\theta \otimes \theta^* \cong \mathbbm{1} \cong \theta^* \otimes \theta$; or, equivalently, if $d(\theta) = 1$. 

Let $\mathcal{M}$ be a right $\mathcal{T}$-module category. We say that two objects $X_1,X_2$ of $\mathcal{M}$ are equivalent \emph{up to a phase in $\mathcal{T}$} if there is a unitary isomorphism $X_1 \cong X_2 \tilde{\otimes} \theta$ for a phase $\theta$ in $\mathcal{T}$.
\end{definition}
\begin{theorem}[Classification of $\F$s in a rigid $C^*$-tensor category]\label{thm:classificationoffs}
Let $\mathcal{T}$ be a rigid $C^*$-tensor category. There is a bijective correspondence between:
\begin{itemize}
\item Morita equivalence classes of simple $\F$s in $\mathcal{T}$. 
\item Equivalence classes of cofinite semisimple indecomposable left $\mathcal{T}$-module categories. 
\end{itemize}
Let $\mathcal{M}$ be a cofinite semisimple indecomposable left $\mathcal{T}$-module category. Since $\mathcal{M}$ is indecomposable, $\End_{\mathcal{T}}(\mathcal{M})$ is a rigid $C^*$-tensor category with a right action on $\mathcal{M}$. There is a bijective correspondence between:
\begin{itemize}
\item Unitary $*$-isomorphism classes of simple $\F$s in the corresponding Morita class.
\item Isomorphism classes of objects in $\mathcal{M}$, up to a phase in $\End_{\mathcal{T}}(\mathcal{M})$.
\end{itemize}
\end{theorem}
\begin{proof}
The first correspondence has already been explained.

For the second correspondence, by Theorem~\ref{thm:pairofpantsisoclassification} there is a bijective correspondence between unitary $*$-isomorphism classes of $\F$s in the corresponding Morita class and isomorphism classes of objects in $\Hom_{\mathcal{T}}(\mathcal{T},\mathcal{M})$ up to a phase in $\End_{\mathcal{T}}(\mathcal{M})$, where $\End_{\mathcal{T}}(\mathcal{M})$ acts on the right by postcomposition. There is a left $\mathcal{T}$-module action on $\Hom_{\mathcal{T}}(\mathcal{T},\mathcal{M})$ induced by the local equivalence $\Psi_{\mathbbm{1},\mathbbm{1}}: \mathcal{T} \to \End_{\mathcal{T}}(\mathcal{T})$. We claim that $\Hom_{\mathcal{T}}(\mathcal{T},\mathcal{M})$ is equivalent to $\mathcal{M}$ as a $\mathcal{T}$-$\End_{\mathcal{T}}(\mathcal{M})$ bimodule category.  Indeed, by essential surjectivity of  $\Psi$, there exists an $\F$ $A$ and an equivalence of $\mathcal{T}$-$\End_{\mathcal{T}}(\mathcal{M})$ bimodule categories $E: \mathcal{M} \overset{\sim}{\to} \Mod$-$A$ (where the right action of $\End_{\mathcal{T}}(\mathcal{M})$ is given by the equivalence $\tilde{E}: \End_{\mathcal{T}}(\mathcal{M}) \to \End_{\mathcal{T}}(\Mod$-$A): F \mapsto E^{-1} \otimes F \otimes E$.) The equivalence $\Psi$ also induces an equivalence of $\mathcal{T}$-$\End_{\mathcal{T}}(\Mod$-$A)$ bimodule categories $\Psi_{\mathbbm{1},A}: \Mod$-$A \overset{\sim}{\to} \Hom_{\mathcal{T}}(\mathcal{T},\Mod$-$A)$, where the right action of $ \End_{\mathcal{T}}(\Mod$-$A)$ on $\Hom_{\mathcal{T}}(\mathcal{T},\Mod$-$A)$ is given by postcomposition; this can be extended to a morphism of $\mathcal{T}$-$\End_{\mathcal{T}}(\mathcal{M})$ bimodule categories using $\tilde{E}$. Finally, there is an equivalence of left $\mathcal{T}$-$\End_{\mathcal{T}}(\mathcal{M})$ bimodule categories $\Hom_{\mathcal{T}}(\mathcal{T},\Mod$-$A) \overset{\sim}{\to} \Hom_{\mathcal{T}}(\mathcal{T},\mathcal{M})$ given by postcomposition with $E^{-1}$.
\end{proof}
\noindent
By what was already said at the beginning of this section, to obtain  a classification of finite-dimensional $G$-$C^*$-algebras for a compact quantum group $G$, simply set $\mathcal{T} = \Rep(G)$ in Theorem~\ref{thm:classificationoffs}.

Before moving on we make a brief remark about how connectedness (a.k.a. ergodicity) of $\F$s (considered in e.g.~\cite{Bichon2005,DeCommer2012,Arano2015}) relates to the above classification. 
\begin{definition}
Let $\mathcal{T}$ be a rigid $C^*$-tensor category. We say that an simple $\F$ $A$ in $\mathcal{T}$ is \emph{connected} if $\Hom(\mathbbm{1},A)$ (i.e. the Hom-space between these objects in $\mathcal{T}$) is one-dimensional.
\end{definition}
\begin{proposition}\label{prop:ergodic}
Let $A$ be an $\F$ in $\mathcal{T}$, let $\mathcal{M}$ be the $\mathcal{T}$-module category representing its Morita class, and let $X$ be an object of $\mathcal{M} \simeq \Hom_{\mathcal{T}}(\mathcal{T},\mathcal{M})$ such that $X \otimes X^* \cong A$. Then $A$ is connected precisely when $X$ is a simple object in $\mathcal{M}$.
\end{proposition}
\begin{proof}
Rigidity of $\Mod(\mathcal{T})$ induces a linear isomorphism between the vector spaces $\Hom(\mathbbm{1},X \otimes X^*)$ in $\mathcal{T} \simeq \End_{\mathcal{T}}(\mathcal{T})$ and $\End(X)$ in $\Hom_{\mathcal{T}}(\mathcal{T},\mathcal{M})$.
\end{proof}

\subsection{A covariant Stinespring theorem}\label{sec:covstinespring}

We now consider covariant channels between $G$-$C^*$-algebras. 

Let us consider the case without symmetry first. Let $A,B$ be two finite-dimensional $C^*$-algebras. As explained in Section~\ref{sec:gc*alg}, using the canonical trace on these $C^*$-algebras we define an inner product giving rise to $\F$s $A,B$ in $\Hilb$. The standard notion of a physical transformation, or \emph{channel}, is a completely positive trace-preserving linear map. It was shown in~\cite{Coecke2016,Heunen2019} that complete positivity of a linear map $A \to B$ as a morphism in $\Hilb$ can be expressed in terms of the Frobenius algebra structures on $A,B$. To this end we make the following definition, which makes sense in any rigid $C^*$-tensor category. 
\begin{definition}\label{def:channel}
Let $\mathcal{T}$ be a rigid $C^*$-tensor category and let $A,B$ be $\F$s in $\mathcal{T}$. Let $f: A \to B$ be a morphism. We say that $f$ satisfies the \emph{CP condition}, or is a \emph{CP morphism}, when there exists an object $S$ of $\mathcal{T}$ and a morphism $g: A \otimes B \to S$ such that the following equation holds:
\begin{calign}\label{eq:cpcond}
\includegraphics[scale=.8]{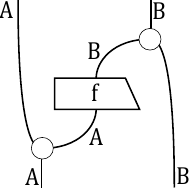}
~~=~~
\includegraphics[scale=.8]{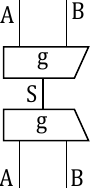}
\end{calign}
In other words, the morphism on the LHS of~\eqref{eq:cpcond} is positive as a morphism in $\Rep(G)$.
\end{definition}
\noindent 
It is shown in~\cite[Thm. 7.18]{Heunen2019} that a linear map $A \to B$ is completely positive precisely when it obeys the CP condition as a morphism in $\Hilb$. To complete the definition of a channel, we observe that, since the canonical trace on a finite-dimensional $C^*$-algebra is precisely the counit of the corresponding $\F$, trace-preservation corresponds to counit-preservation. 
\begin{definition}
We say that a CP morphism $f: A \to B$ is a  \emph{channel} when it satisfies the second equation of~\eqref{eq:cohomo}.
\end{definition}
\noindent
This characterisation extends straightforwardly to $G$-$C^*$-algebras for a compact quantum group $G$. In Section~\ref{sec:gc*alg} we saw how an $\F$ $A$ in $\Rep(G)$ corresponds to a $G$-$C^*$-algebra equipped with its canonical $G$-invariant functional; the concrete $C^*$-algebra is obtained as the image $F(A)$ of the $\F$ under the canonical fibre functor $F: \Rep(G) \to \Hilb$, and the concrete $A_G$-coaction is obtained by T-K-W duality. The canonical fibre functor maps a morphism $A \to B$ in $\Rep(G)$ to a \emph{covariant} linear map $F(A) \to F(B)$ (that is, an intertwiner of $G$-representations). It is known (see e.g.~\cite[Prop. 3.22]{Verdon2020b}) that, for $\F$s $A,B$ in $\Rep(G)$, for any covariant completely positive map $f: F(A) \to F(B)$ there is a unique CP morphism $\tilde{f}:A \to B$ in $\Rep(G)$ such that $F(\tilde{f}) = f$. Preservation of the canonical $G$-invariant functional precisely corresponds to counit preservation. Completely positive maps/channels between $G$-$C^*$-algebras can therefore be identified with CP morphisms/channels between $\F$s in $\Rep(G)$. 

We can further generalise by considering CP morphisms and channels between $\F$s in $\mathcal{T}$, where $\mathcal{T}$ is a general rigid $C^*$-tensor category. Without a fibre functor there is no obvious way to identify $\F$s with concrete $C^*$-algebras or morphisms with linear maps; however, the theory holds in this general setting.

We now state the result. Let $\mathcal{T}$ be a rigid $C^*$-tensor category. We saw in Corollary~\ref{cor:embedtmodt} that $\mathcal{T}$ embeds as the endomorphism category $\End_{\mathcal{T}}(\mathcal{T})$ in the semisimple $C^*$-2-category $\Mod(\mathcal{T})$. By semisimplicity, for every $\F$ $A$ in $\mathcal{T} \simeq \End_{\mathcal{T}}(\mathcal{T})$ there exists an object $\mathcal{M}$ of $\Mod(\mathcal{T})$ and a separable 1-morphism $X: \mathcal{T} \to \mathcal{M}$ such that $A \cong X \otimes X^*$. This is the context for the following theorem.

\begin{theorem}[Covariant Stinespring theorem]\label{thm:covstinespring}
Let $\mathcal{C}$ be a semisimple $C^*$-2-category and let $r$ be any object. Let $X:r \to s$, $Y: r \to t$ be separable 1-morphisms, and let $f: X \otimes X^* \to Y \otimes Y^*$ be a CP morphism between the corresponding $\F$s in $\End(r)$. 

Then there exists a 1-morphism $E: t \to s$ (the `environment') and a 2-morphism $\tau: X \to Y \otimes E$ such that the following equation holds:
\begin{calign}\label{eq:stinespring}
\includegraphics[scale=.8]{pictures/covstinespring/covstinespring11.pdf}
~~=~~
\includegraphics[scale=.8]{pictures/covstinespring/covstinespring12.pdf}
\end{calign}
We say that $\tau$ is a \emph{dilation} of $f$. The morphism 
\begin{calign}\label{eq:ssisometrycond}
\includegraphics[scale=.8]{pictures/covstinespring/covstinespringextra1.pdf}
\end{calign}
is an isometry if and only if $f$ is a channel.

In the other direction, for any 1-morphism $E: t \to s$ and 2-morphism $\tau: X \to Y \otimes E$, the morphism $f: X \otimes X^* \to Y \otimes Y^*$ defined by~\eqref{eq:stinespring} is CP, and a channel if and only if~\eqref{eq:ssisometrycond} is an isometry.

Different dilations for a CP morphism $f: X \otimes X^* \to Y \otimes Y^*$ are related by a partial isometry on the environment. Specifically, let $\tau_1: X \to Y \otimes E_1$, $\tau_2: X \to Y \otimes E_2$ be two dilations of $f$. Then there exists a partial isometry $\alpha: E_1 \to E_2$ such that 
\begin{align*}\label{eq:alphareldilations}
(\id_Y \otimes \alpha) \circ \tau_1 = \tau_2 
&&
(\id_Y \otimes \alpha^{\dagger}) \circ \tau_2 = \tau_1 
\end{align*}
In particular, the dilation minimising the quantum dimension of the environment $d(E)$ is unique up to unitary $\alpha$. (A concrete construction of the minimal dilation from any other dilation is specified in the last paragraph of the proof.)
\end{theorem}
\begin{proof}
The fact that a morphism between SSFAs is CP iff it admits a representation~\eqref{eq:stinespring} was shown in~\cite[Lem. 5.12]{Henriques2020}. It is straightforward to see that~\eqref{eq:ssisometrycond} is an isometry if and only if $f$ is a channel:
\begin{calign}
\includegraphics[scale=.8]{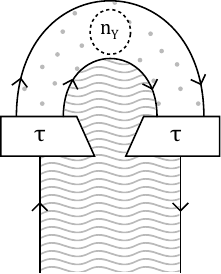}
~~=~~
\includegraphics[scale=.8]{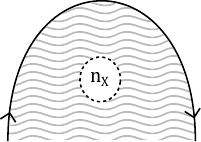}
\qquad \Leftrightarrow \qquad 
\includegraphics[scale=.8]{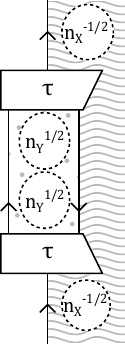}
~~=~~
\includegraphics[scale=.8]{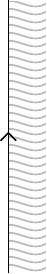}
\end{calign}
We now show that different dilations are related by a partial isometry. Let $\tau_1: X \to Y \otimes E_1$, $\tau_2: X \to Y \otimes E_2$ be two dilations of the same CP morphism. For each $i \in \{1,2\}$ we define the following morphism $\tilde{\tau}_i: Y^* \otimes X \to E_i:$
\begin{calign}
\includegraphics[scale=.8]{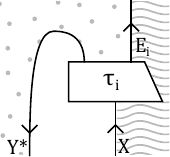}
\end{calign}
The fact that $\tau_1$ and $\tau_2$ are dilations of the same CP morphism comes down to the following equation:
\begin{calign}
\includegraphics[scale=.8]{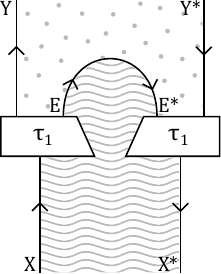}
~~=~~
\includegraphics[scale=.8]{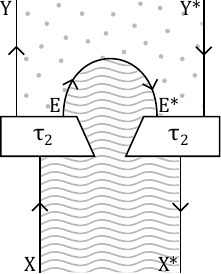}
\qquad \Leftrightarrow \qquad
\includegraphics[scale=.8]{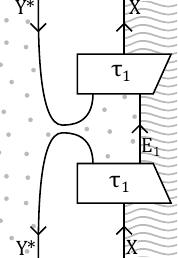}
~~=~~
\includegraphics[scale=.8]{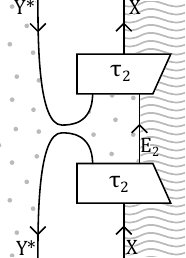}
\end{calign}
In inline notation, this is:
\begin{equation}\label{eq:cpinline}
\tilde{\tau}_1^{\dagger} \circ \tilde{\tau}_1 = \tilde{\tau}_2^{\dagger} \circ \tilde{\tau}_2 
\end{equation}
We perform the polar decomposition (Remark~\ref{rem:polar}) on $\tilde{\tau}_1$ and $\tilde{\tau}_2$. Observe that
$$|\tilde{\tau}| := |\tilde{\tau}_1| = (\tilde{\tau}_1^{\dagger} \circ \tilde{\tau}_1)^{1/2} = (\tilde{\tau}_2^{\dagger} \circ \tilde{\tau}_2)^{1/2} = |\tilde{\tau}_2|,$$
where the second equality is by~\eqref{eq:cpinline}. We therefore have 
\begin{align*}
\tilde{\tau}_1 = u_1 \circ |\tilde{\tau}| && \tilde{\tau}_2 = u_2 \circ |\tilde{\tau}|
\end{align*}
where $u_i: Y^* \otimes X \to E_i$ is a partial isometry such that $u_i^{\dagger} \circ u_i = s(|\tilde{\tau}|)$ and $u_i \circ u_i^{\dagger} = s(|\tilde{\tau}^{\dagger}|)$.

Now we define $\alpha:= u_2 \circ u_1^{\dagger}$. To see that $\alpha: E_1 \to E_2$ is a partial isometry:
\begin{align}\label{eq:alphapartialisomdilation}
\alpha^{\dagger} \circ \alpha = u_1 \circ u_2^{\dagger} \circ u_2 \circ u_1^{\dagger} = u_1 \circ s(|\tilde{\tau}|) \circ u_1^{\dagger} =  u_1 \circ u_1^{\dagger} \circ u_1 \circ u_1^{\dagger} = s(|\tilde{\tau}^{\dagger}|)
\end{align}
The fact that $(\id_Y \otimes \alpha) \circ \tau_1 = \tau_2$ follows immediately from the following equation (simply transpose the $Y$-wire):
\begin{align*}
\alpha \circ (\tilde{\tau}_1) = u_2 \circ u_1^{\dagger} \circ u_1 \circ |\tilde{\tau}| = u_2 \circ s(|\tilde{\tau}|) \circ |\tilde{\tau}| = u_2 \circ |\tilde{\tau}| = \tilde{\tau}_2
\end{align*}
The proof that $(\id_Y \otimes \alpha^{\dagger}) \circ \tau_2 = \tau_1$ is similar. 

To see that the dilation minimising the quantum dimension of the environment is unique up to a unitary, suppose that $\tau_1$ and $\tau_2$ are minimal dilations related by a partial isometry $\alpha: E_1 \to E_2$. Then $\tau_1 = (\id_Y \otimes (\alpha^{\dagger} \circ \alpha)) \circ \tau_1$. Let us split the dagger idempotent $\alpha^{\dagger} \circ \alpha$ to obtain an isometry $i: \tilde{E}_1 \to E_1$. Then $(\id_Y \otimes i^{\dagger}) \circ \tau_1$ is also a dilation, and $d(\tilde{E}_1) \leq d(E_1)$ with equality iff $i$ is unitary; unitarity of $i$ therefore follows by minimality of $\tau_1$, and it follows that $\alpha$ is an isometry. Making the same argument for $\alpha \circ \alpha^{\dagger}$ and $\tau_2$ we obtain that $\alpha$ is a coisometry, and therefore $\alpha$ is unitary. 

Finally, to construct a minimal dilation from a given dilation $\tau: X \to Y \otimes E$, take the projection $s(|\tilde{\tau}^{\dagger}|)$. Split this idempotent to obtain an isometry $i: \tilde{E} \to E$ and define the minimal dilation as $(\id_Y \otimes i^{\dagger}) \circ \tau$. It follows from~\eqref{eq:alphapartialisomdilation} and the definition of $s(|\tilde{\tau}^{\dagger}|)$ that, for any other dilation $\tau': X \to Y \otimes E'$, the partial isometry $\alpha: \tilde{E} \to E'$ relating it to the minimal dilation will be a genuine isometry; therefore $d(\tilde{E}) \leq d(E')$ with equality iff $\alpha$ is unitary. 
\end{proof}
\noindent
We now show how this theorem recovers previous results in the literature.
\begin{example}[{Finite-dimensional Stinespring theorem}]\label{ex:vanillastinespring}
Let us show how Theorem~\ref{thm:covstinespring} implies the standard f.d. covariant Stinespring's theorem (e.g.~\cite[Thm. 2]{Holevo2007}\cite[Thm. 1]{Stinespring1955}\cite[Thm. 15]{Szafraniec2010}\cite[Thm. 1]{Scutaru1979}\cite[Thm. 2.1]{Paulsen1982}).

Let $A,B$ be $G$-$C^*$-algebras. There is an equivalence between the 2-category $\Mod(\Rep(G))$ and the 2-category whose objects are finite-dimensional $G$-$C^*$-algebras, whose 1-morphisms are $G$-equivariant finitely generated Hilbert bimodules, and whose 2-morphisms are equivariant bimodule maps. This equivalence takes the object $\Rep(G)$ onto the trivial $G$-$C^*$-algebra $\mathbb{C}$. Therefore $A \cong X \otimes X^*$ and $B \cong Y \otimes Y^*$, where $X$ and $Y$ are equivariant right Hilbert $A$- and right Hilbert $B$-modules respectively, considered as 1-morphisms $\mathbb{C} \to A$ and $\mathbb{C} \to B$. In Theorem~\ref{thm:covstinespring} we characterised CP maps $f: X \otimes X^* \to Y \otimes Y^*$.  Making a slight conventional change (which is clearly equivalent by bending the $E$-wire and scaling $\tau$), the theorem says that $f$ is completely positive if and only if there exists a Hilbert $A$-$B$-bimodule $E: A \to B$ and an equivariant bimodule map $\tau: X \otimes E \to Y$ such that the following equation holds:
\begin{calign}\label{eq:stinespringconventionchange}
\includegraphics[scale=.8]{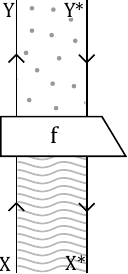}
~~=~~
\includegraphics[scale=.8]{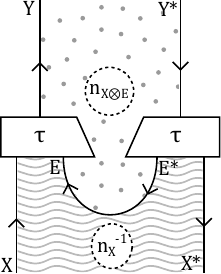}
\end{calign}
Here the regions corresponding to $A$ are shaded with wavy lines and regions corresponding to $B$ with polka dots. We observe that $X \otimes E$ is an equivariant right Hilbert $B$-module; this is the Hilbert $B$-module $\mathcal{E}$ in~\cite[Eq. 22]{Szafraniec2010}. The pair of pants algebra $(X \otimes E) \otimes (E^* \otimes X^*)$ is the $*$-algebra $B^*(\mathcal{E})$. Now it is easy to check that the 2-morphism
\begin{calign}\label{eq:stinespringstarhom}
\includegraphics[scale=.8]{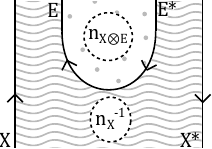}
\end{calign}
is a $*$-homomorphism; we thus obtain the equivariant $*$-homomorphism $\Phi: A \to B^*(\mathcal{E})$. As a map $B^*(\mathcal{E}) \to B$, the 2-morphism 
\begin{calign}
\includegraphics[scale=.8]{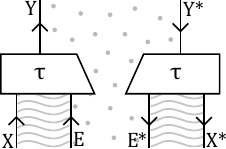}
\end{calign}
can be expressed as $x \mapsto \hat{\tau}  x \hat{\tau}^{\dagger}$, where $\hat{\tau}:= \tau \otimes (n_Y^{-1/2} \circ n_{X \otimes E}^{1/2})$. (Here the normalisation comes from the choice of functional on the algebra.) We therefore set $\tau^{\dagger}: Y \to X \otimes E$ to be the equivariant module map $V$ of~\cite[Eq. 22]{Szafraniec2010}. We thus obtain the characterisation of completely positive maps $A \to B(H)$ given in that theorem. Another common statement (not actually given in~\cite{Stinespring1955}) is that $f$ is unital if and only if $V$ is an isometry. But $f$ is unital if and only if $f^{\dagger}$ is trace-preserving, and so we require that 
\begin{calign}\label{eq:stinspinftp}
\includegraphics[scale=.8]{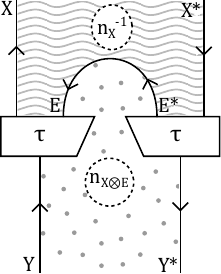}
\end{calign}
is trace-preserving. By Theorem~\ref{thm:covstinespring} we see that~\eqref{eq:stinspinftp} preserves the canonical separable trace if and only if $ \tau^{\dagger} \otimes (n_Y^{-1/2} \circ n_{X \otimes E}^{1/2}) =  V$ is an isometry.
\end{example}

\subsection{A covariant Choi theorem}\label{sec:choi}

We finish by observing the following corollary of the covariant Stinespring theorem. 
\begin{theorem}[Covariant Choi theorem]\label{thm:choi}
Let $\mathcal{C}$ be a semisimple $C^*$-2-category and let $r$ be any object. Let $X: r \to s$, $Y: r \to t$ be separable 1-morphisms, and let $X \otimes X^*$ and $Y \otimes Y^*$ be the corresponding $\F$s in $\End(r)$. Then there is a bijective correspondence (in fact, an isomorphism of convex cones, in the sense that it preserves positive linear combinations) between:
\begin{itemize}
\item CP morphisms $f: X \otimes X^* \to Y \otimes Y^*$.
\item Positive elements $\tilde{f} \in \End(Y^* \otimes X)$.
\end{itemize}
The correspondence is given as follows:
\begin{calign}
\includegraphics[scale=.7]{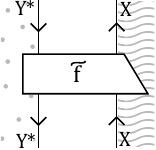}
~~=~~
\includegraphics[scale=.7]{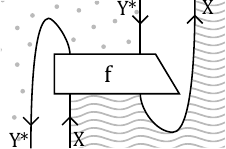}
\end{calign}
\end{theorem}
\begin{proof}
Let $\tau: X \to Y \otimes E$ be a dilation of $f$, then:
\begin{calign}
\includegraphics[scale=.7]{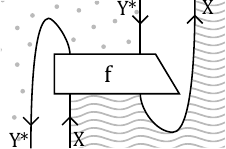}
~~=~~
\includegraphics[scale=.7]{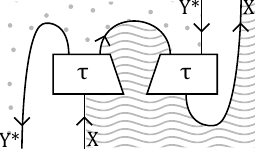}
~~=~~
\includegraphics[scale=.7]{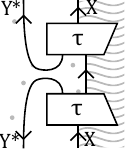}
\end{calign}
The last diagram is clearly the composition of a 2-morphism with its dagger and is therefore positive.

In the other direction, let $\tilde{f}$ be positive. Then we can choose $m$ such that $\tilde{f} = m^{\dagger} \circ m$, and transposing the relevant wires we obtain a dilation for $f$.
\end{proof}
\begin{remark}
To recover the usual Choi's theorem for matrix $C^*$-algebras~\cite{Choi1975}, let $\mathcal{T} = \Hilb$ and $X,Y$ be 1-morphisms $\Hilb \to \Hilb$ in $\Mod(\Hilb)$, i.e. Hilbert spaces. Then Theorem~\ref{thm:choi} says precisely that CP maps $B(X) \to B(Y)$ correspond to positive elements of $B(Y^* \otimes X)$.
\end{remark}

\section{Acknowledgements}

We thank Benjamin Musto, David Reutter, Changpeng Shao, Jamie Vicary and Makoto Yamashita for useful discussions. We thank Luca Giorgetti, David Penneys, Franciszek Szafraniec and an anonymous referee for pointing out some relevant previous work we had missed when reviewing the literature; we apologise to those authors whose work was overlooked in the earlier version. We used the \emph{diagrams} package~\cite{Taylor} for commutative diagrams, and the open-source vector graphics editor \emph{Inkscape} for 2-category diagrams. This project has received funding from the European Research Council (ERC) under the European 8 Union’s Horizon 2020 research and innovation programme (grant agreement No. 817581). This project has also received funding from EPSRC.

\bibliographystyle{halpha}
\bibliography{bibliography}

\section{Appendices}
\subsection{Matrix notation for presemisimple 2-categories}\label{app:mat2cats}
\begin{definition}\label{def:matc}
Let $\mathcal{C}$ be a presemisimple $C^*$-2-category, and let $\{r_{\sigma}\}_{\sigma \in \Sigma}$ be representatives of equivalence classes of simple objects in $\mathcal{C}$, with index set $\Sigma$. We define a presemisimple $C^*$-2-category $\Mat(\mathcal{C})$ as follows:
\begin{itemize}
\item Objects: Finite-length vectors $\vec{\sigma} = [\sigma_1, \dots, \sigma_n]$ of elements $\sigma_i \in \Sigma$ (including the empty vector, which is a zero object).
\item 1-morphisms $[\sigma_1, \dots, \sigma_{n_1}] \to{} [\tau_1, \dots, \tau_{n_2}]$: $n_1 \times n_2$ matrices $M$ whose $i,j$-th entry $M_{ij}$ is a 1-morphism $r_{\sigma_i} \to r_{\tau_{j}}$.
\item 2-morphisms $M \to N$: $n_1 \times n_2$ matrices $f$ whose $i,j$-th entry $f_{ij}$ is a 2-morphism $M_{ij} \to N_{ij}$.
\item Composition of 1-morphisms: $(M \otimes N)_{ik}:= \oplus_{j} M_{ij} \otimes M_{jk}$.
\item Identity 1-morphisms: $(\id_{[\sigma_1,\dots,\sigma_n]})_{ij}:= \delta_{ij} \mathbbm{1}_i$ (where $\delta_{ij}$ indicates the zero 1-morphism $r_{\sigma_i} \to r_{\sigma_{j}}$ if $i \neq j$, and $\mathbbm{1}_i := \id_{r_{\sigma_i}}$).
\item Horizontal composition of 2-morphisms: $(f \otimes g)_{km} := \sum_l i_l \circ (f_{kl} \otimes g_{lm}) \circ i_l^{\dagger}$, where $i_l: M_{kl} \otimes M_{lm} \to \oplus_l M_{kl} \otimes M_{lm}$ is the injection isometry of the direct sum.
\item Vertical composition of 2-morphisms: $(g \circ f)_{ij}:= g_{ij} \circ f_{ij}$.
\item Dagger on 2-morphisms: $(f^{\dagger})_{ij} := (f_{ij})^{\dagger}$.
\item $\mathbb{C}$-linear structure on 2-morphisms: $(\lambda f)_{ij} = \lambda f_{ij}$.
\item Associators: Matrix entries $(\alpha_{M,N,O})_{il}$ are the unitary natural isomorphisms $\oplus_k (\oplus_j M_{ij} \otimes N_{jk}) \otimes O_{kl} \cong \oplus_j M_{ij} \otimes (\oplus_k N_{jk} \otimes O_{kl})$ in $\mathcal{C}$.\footnote{We remark that, as far as we are aware, the existence of natural isomorphisms distributing direct sum over tensor product depends on rigidity of $\mathcal{C}$; see e.g.~\cite[Sec. 3.3.2]{Heunen2019}.}
\item Unitors: The matrix entries (modulo ${\bf 0} \otimes X \cong X \cong X \otimes {\bf 0}$) are given by the unitors $\id_{r_{\sigma_i}} \otimes M_{ij} \cong M_{ij} \cong M_{ij} \otimes \id_{r_{\sigma_j}}$ in $\mathcal{C}$.
\item Additive structure on 1-morphisms: $(M \oplus N)_{ij} := M_{ij} \oplus N_{ij}$.
\item Additive structure on objects: $[\sigma_1,\dots,\sigma_n] = [\sigma_1] \boxplus \dots \boxplus [\sigma_n]$ with the following injection and projection 1-morphisms.
\begin{align*}
\iota_i:=
\begin{pmatrix}
{\bf 0} & \cdots & \mathbbm{1}_i & \cdots & {\bf 0}
\end{pmatrix}
&&
\rho_i:=
\begin{pmatrix}
{\bf 0} \cdots & \mathbbm{1}_i \cdots & {\bf 0}
\end{pmatrix}^T
\end{align*}
It is moreover straightforward to check that $\rho_i$ is a right dual for $\iota_i$ with the following cup and cap (here $\diag([\vec{v}])$ indicates a diagonal matrix of 2-morphisms, i.e. the 2-morphisms on the diagonal are given by the vector $\vec{v}$ and all the 2-morphisms not on the diagonal are the zero 2-morphism):
\begin{align}\label{eq:injectiondual}
\eta_i:= \diag([0, \dots, \id_{\mathbbm{1}_i}, \dots, 0])
&&
\epsilon_i := \id_{\mathbbm{1}_i}
\end{align}
\end{itemize}
\end{definition}
\noindent
We now show that $\Mat(\mathcal{C})$ has duals. In fact, it is straightforward to show that all right duals in $\Mat(\mathcal{C})$ are of the following form. Let $M: [\sigma_1,\dots,\sigma_{n_1}] \to{} [\tau_1,\dots,\tau_{n_2}]$ be a 1-morphism in $\Mat(\mathcal{C})$. Pick a dual $((M_{kl})^*,\eta_{kl},\epsilon_{kl})$ for each $M_{kl}$. We now define a dual 1-morphism $M^*$ as the `conjugate transpose' matrix, i.e.:
\begin{align}\label{eq:dualmatrix}
(M^*)_{kl}:= (M_{lk})^*
\end{align} 
We observe that: 
\begin{align*}
(M^* \otimes M)_{km} = \bigoplus_{l} (M_{lk})^* \otimes M_{lm}
&&
(M \otimes M^*)_{km} = \bigoplus_{l} M_{kl} \otimes (M_{ml})^*
\end{align*}
We then define a right cup and cap $\eta: \id_{[\tau_1,\dots,\tau_{n_2}]} \to M^* \otimes M$ and $\epsilon: M \otimes M^* \to \id_{[\sigma_1,\dots,\sigma_{n_1}]}$ as the following diagonal matrices of 2-morphisms:
\begin{align}\label{eq:matrixdualcupcap}
\eta := \diag([\sum_{l} i_{l1} \circ \eta_{l1},\dots, \sum_{l} i_{l n_2}  \circ \eta_{ln_2}])
&&
\epsilon := \diag([\sum_{l} \epsilon_{1l} \circ \bar{i}_{1l}^{\dagger},\dots, \sum_{l} \epsilon_{n_1 l}  \circ \bar{i}_{n_1 l}^{\dagger}])
\end{align}
Here $i_{lk}: (M_{lk})^* \otimes M_{lk} \to \oplus_l (M_{lk})^* \otimes M_{lk}$ and $\bar{i}_{kl}: M_{kl} \otimes (M_{kl})^* \to \oplus_l M_{kl} \otimes (M_{kl})^*$ are the isometric injections of the direct sums. We will show one of the snake equations for $[M^*,\eta,\epsilon]$; the other is shown similarly. Let $s = (\id_{M^*} \otimes \epsilon)\circ (\eta \otimes \id_{M^*})$. Let $\{\rho_{\sigma,i}, \iota_{\sigma,i}\}$ and $\{\rho_{\tau, i}, \iota_{\tau,i}\}$ be the dual projection and injection 1-morphisms of the direct sum decompositions $\vec{\sigma} = \boxplus_i [\sigma_{i}]$, $\vec{\tau} = \boxplus_i [\tau_i]$. We need to show that $s = \id_{M^*}$; or equivalently that $\id_{\iota_{\tau,j}} \otimes s \otimes \id_{\rho_{\sigma_,i}} = \id_{\iota_{\tau,j}} \otimes \id_{M^*}\otimes \id_{\rho_{\sigma_,i}} = \id_{(M_{ij})^*}$ for all $i,j$. Now:
\begin{calign}\nonumber
\includegraphics[scale=1]{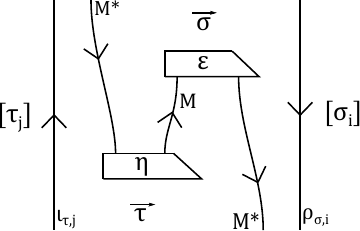}
~~=~~
\includegraphics[scale=1]{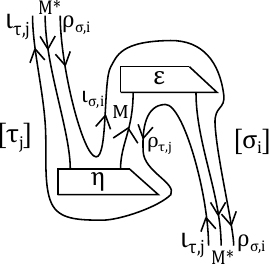}
\\
~~=~~
\includegraphics[scale=1]{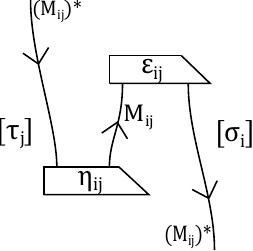}
~~=~~
\includegraphics[scale=1]{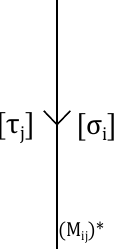}
\end{calign}
Here the first equality uses duality of $\iota$ and $\rho$, the second equality follows straightforwardly from the definitions~(\ref{eq:injectiondual},\ref{eq:matrixdualcupcap}), and the final equality follows since $[(M_{ij})^*,\eta_{ij},\epsilon_{ij}]$ is a right dual for $M_{ij}$ in $\mathcal{C}$.

We now define an equivalence between $\Mat(\mathcal{C})$ and $\mathcal{C}$. 
\begin{definition}
The unitary $\mathbb{C}$-linear 2-functor $\Phi: \Mat(\mathcal{C}) \to \mathcal{C}$ is defined as follows:
\begin{itemize}
\item On objects: $\Phi([\sigma_1,\dots,\sigma_n]) := r_{\sigma_1} \boxplus \dots \boxplus r_{\sigma_n}$.
\item On 1-morphisms: Let $M: [\sigma_1,\dots,\sigma_{n_1}] \to{} [\tau_1,\dots,\tau_{n_2}]$. Let $\{\iota_{\sigma,i},\rho_{\sigma,i}\}$ and $\{\iota_{\tau,i},\rho_{\tau,i}\}$ be the injection and projection 1-morphisms for the direct sums $r_{\sigma_1} \boxplus \dots \boxplus r_{\sigma_{n_1}}$ and $r_{\tau_1} \boxplus \dots \boxplus r_{\tau_{n_2}}$. Then define $\Phi(M) := \oplus_{ij} ((\rho_{\sigma,i} \otimes M_{ij}) \otimes \iota_{\tau,j})$.
\item On 2-morphisms: Let $f: M \to N$ be a 2-morphism. Then we define $\Phi(f):= \sum_{kl }i_{kl} \circ ((\id_{\rho_{\sigma,k}} \otimes f_{kl}) \otimes \id_{\iota_{\tau,l}}) \circ i_{kl}^{\dagger}$, where $i_{kl}: ((\rho_{\sigma,k} \otimes M_{kl}) \otimes \iota_{\tau,l}) \to \oplus_{kl} ((\rho_{\sigma,k} \otimes M_{kl}) \otimes \iota_{\tau,l})$ is the isometric injection of the direct sum. 
\item Multiplicators: Given by the natural unitary isomorphisms $\Phi(M) \otimes \Phi(N) 
= 
(\oplus_{ij} ((\rho_{\sigma,i} \otimes M_{ij}) \otimes \iota_{\tau,j})) \otimes (\oplus_{kl} ((\rho_{\tau,k} \otimes N_{kl}) \otimes \iota_{\nu,l})) 
\cong 
\oplus_{ijkl}~ ((\rho_{\sigma,i} \otimes M_{ij}) \otimes (\iota_{\tau,j} \otimes \rho_{\tau,k})) \otimes (N_{kl} \otimes \iota_{\nu,l})
\cong
\oplus_{il}~ (\rho_{\sigma,i} \otimes (\oplus_j (M_{ij} \otimes N_{jl}))) \otimes \iota_{\nu,l}
=
\Phi(M \otimes N).
$ 
Here for the first isomorphism we used the associators and distributivity of the direct sum over tensor product in $\mathcal{C}$, and for the second isomorphism we used $\iota_{\tau,j} \otimes \rho_{\tau,k} \cong \delta_{jk} \id_{X_{\sigma_j}}$.
\item Unitors: Given by the unitary isomorphisms 
$\id_{\boxplus_i  r_{\sigma_{i}}} \cong \oplus_{i} (\iota_i \otimes \rho_i) \cong \Phi(\id_{[\sigma_1,\dots,\sigma_n]})
$.
\end{itemize}
\end{definition}
\begin{proposition}\label{prop:matcequiv}
The 2-functor $\Phi: \Mat(\mathcal{C}) \to \mathcal{C}$ is an equivalence. 
\end{proposition}  
\begin{proof}
\begin{itemize}
\item \emph{Essentially surjective on objects.} By presemisimplicity every object is a direct sum of simples, and direct sums are unique up to equivalence. 
\item \emph{Essentially surjective on 1-morphisms.} Let $X: \boxplus_i r_{\sigma_i} \to  \boxplus_j r_{\tau_j}$ be a 1-morphism. Now we have 
\begin{align*}
X &\cong ((\oplus_i~ \rho_{\sigma,i} \otimes \iota_{\sigma,i}) \otimes X) \otimes (\oplus_j~ \rho_{\tau,j} \otimes \iota_{\tau,j})
\\ 
&\cong \oplus_{i,j}~ (\rho_{\sigma,i} \otimes ((\iota_{\sigma,i} \otimes X) \otimes \rho_{\tau,j})) \otimes \iota_{\tau,j}
\\
&\cong \Phi(M)
\end{align*}
where $M: [\sigma_1,\dots,\sigma_{n_1}] \to{} [\tau_1,\dots,\tau_{n_2}]$ is defined by $M_{ij}:= (\iota_{\sigma,i} \otimes X) \otimes \rho_{\tau,j}$.
\item \emph{Full on 2-morphisms.} Let $f: \Phi(M) = \oplus_{i,j} ((\rho_{\sigma,i} \otimes M_{ij}) \otimes \iota_{\tau,j}) \to \oplus_{i,j} ((\rho_{\sigma,i} \otimes N_{ij}) \otimes \iota_{\tau,j}) = \Phi(N)$ be a 2-morphism in $\mathcal{C}$. Let $\kappa_{\sigma}: \id_{\boxplus_{i} r_{\sigma_i}} \overset{\sim}{\to} \oplus_{i} \rho_{\sigma,i} \otimes \iota_{\sigma,i}$ and $\lambda_{\sigma,i}: \iota_i \otimes \rho_i \overset{\sim}{\to} \id_{r_{\sigma_i}}$ be the unitary isomorphisms in the definition of the direct sum $\boxplus_{i} r_{\sigma_i}$, and define $\kappa_{\tau}$ and $\lambda_{\tau,i}$ similarly. Let $\nabla_{M,i,j}: ((\rho_{\sigma,i} \otimes M_{ij}) \otimes \iota_{\tau,j}) \to \oplus_{i,j} ((\rho_{\sigma,i} \otimes M_{ij}) \otimes \iota_{\tau,j})$, $\nabla_{N,i,j}: ((\rho_{\sigma,i} \otimes N_{ij}) \otimes \iota_{\tau,j}) \to \oplus_{i,j} ((\rho_{\sigma,i} \otimes N_{ij}) \otimes \iota_{\tau,j})$, $\nabla_{\sigma,i}: \rho_{\sigma,i} \otimes \iota_{\sigma,i} \to \oplus_{i} \rho_{\sigma,i} \otimes \iota_{\sigma,i}$, $\nabla_{\tau,i}: \rho_{\tau,i} \otimes \iota_{\tau,i} \to \oplus_{i} \rho_{\tau,i} \otimes \iota_{\tau,i}$ be the isometric injections of the various direct sums of 1-morphisms; we depict these as labelled downwards-pointing triangles in the diagram, and their daggers as labelled upwards-pointing triangles. Then we have the following sequence of equalities in $\mathcal{C}$:
\begin{calign}
\includegraphics[scale=1]{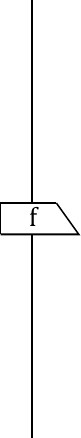}
~~=~~
\sum_{i,j,k,l,m,n}~
\includegraphics[scale=1]{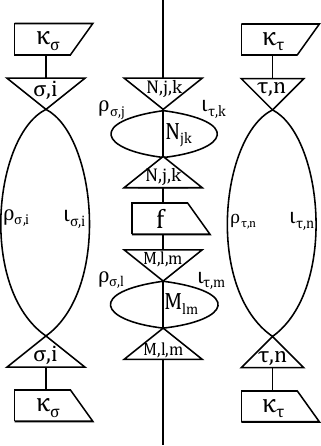}
\\
=~~
\sum_{i,n}~
\includegraphics[scale=1]{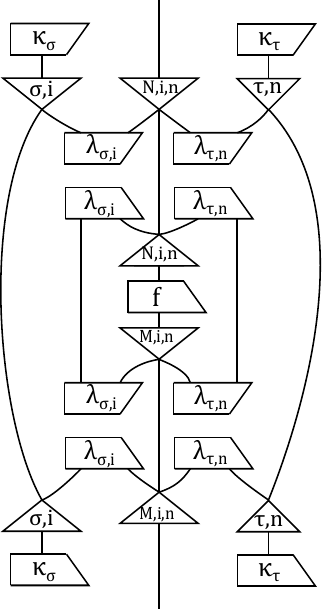}
\\=~~
\sum_{i,n}~
\includegraphics[scale=1]{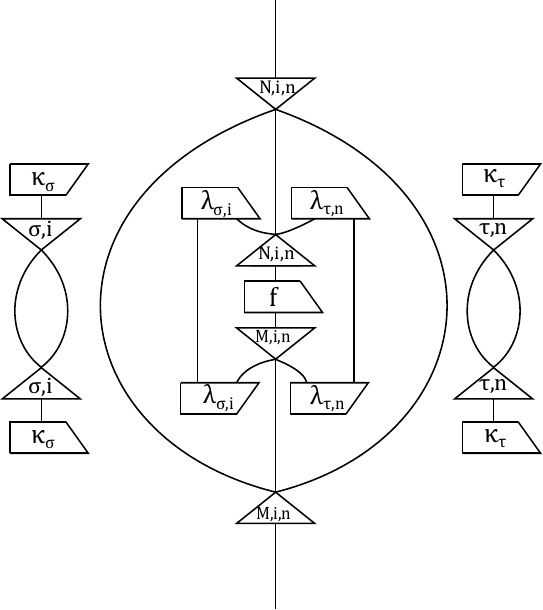}
\\
=~~\sum_{i,n}~
\includegraphics[scale=1]{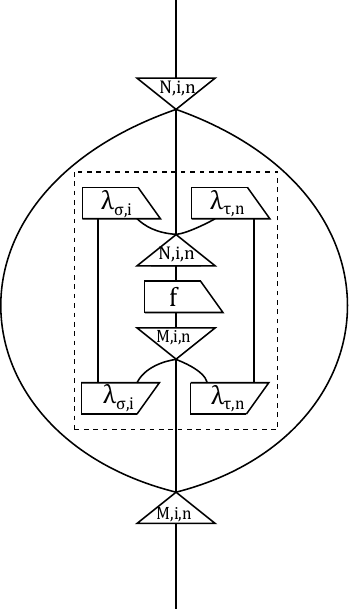}
\end{calign}
Here the first equality is by $\kappa_{\sigma}^{\dagger} \circ (\sum_i \nabla_{\sigma,i} \circ \nabla_{\sigma,i}^{\dagger}) \circ \kappa_{\sigma} = \id_{\id_{\boxplus_{i}r_{\sigma_i}}}$, $\kappa_{\tau}^{\dagger} \circ (\sum_n \nabla_{\tau,n} \circ \nabla_{\tau,n}^{\dagger}) \circ \kappa_{\tau} = \id_{\id_{\boxplus_{n}r_{\tau_n}}}$, $\sum_{l,m} \nabla_{M,l,m} \circ \nabla_{M,l,m}^{\dagger} = \id_{\Phi(M)}$ and $\sum_{j,k} \nabla_{N,j,k} \circ \nabla_{N,j,k}^{\dagger} = \id_{\Phi(N)}$. The second equality is by $\id_{\iota_{\sigma,i} \otimes \rho_{\sigma,l}} = \delta_{il} (\lambda_{\sigma,i}^{\dagger} \circ \lambda_{\sigma,i})$, $\id_{\iota_{\tau,m} \otimes \rho_{\tau,n}} = \delta_{mn} (\lambda_{\tau,n}^{\dagger} \circ \lambda_{\tau,n})$, $\id_{\iota_{\sigma,i} \otimes \rho_{\sigma,j}} = \delta_{ij} (\lambda_{\sigma,i}^{\dagger} \circ \lambda_{\sigma,i})$ and $\id_{\iota_{\tau,k} \otimes \rho_{\tau,n}} = \delta_{kn} (\lambda_{\tau,n}^{\dagger} \circ \lambda_{\tau,n})$. The third equality is by unitarity of $\lambda_{\sigma,i}$ and $\lambda_{\tau,n}$. The fourth equality is by $\id_{\iota_{\sigma,j} \otimes \rho_{\sigma,i}} = \delta_{ij} \id_{\iota_{\sigma,i} \otimes \rho_{\sigma,i}}$ and 
$\id_{\iota_{\sigma,n} \otimes \rho_{\sigma,o}} = \delta_{no} \id_{\iota_{\sigma,n} \otimes \rho_{\sigma,n}}$, which allows us to use $\kappa_{\sigma}^{\dagger} \circ (\sum_i \nabla_{\sigma,i} \circ \nabla_{\sigma,i}^{\dagger}) \circ \kappa_{\sigma} = \id_{\id_{\boxplus_{i}r_{\sigma_i}}}$ and $\kappa_{\tau}^{\dagger} \circ (\sum_n \nabla_{\tau,n} \circ \nabla_{\tau,n}^{\dagger}) \circ \kappa_{\tau} = \id_{\id_{\boxplus_{n}r_{\tau_n}}}$ again.

In the last diagram we see $\Phi(\tilde{f})$ where $\tilde{f}_{in}: M_{in} \to N_{in}$ is the morphism in the dashed box.
\item \emph{Faithful on 2-morphisms}. Let $f,g: M \to N$ be 2-morphisms in $\Mat(\mathcal{C})$. Then:
\begin{align*}
\Phi(f) = \Phi(g) &\Leftrightarrow  \sum_{ij} \nabla_{N,i,j} \circ (\id_{\rho_{\sigma,i}} \otimes f_{ij} \otimes \id_{\iota_{\tau,j}}) \circ \nabla_{M,i,j}^{\dagger} = \sum_{ij} \nabla_{N,i,j} \circ (\id_{\rho_{\sigma,i}} \otimes g_{ij} \otimes \id_{\iota_{\tau,j}}) \circ \nabla_{M,i,j}^{\dagger}
\\
& \Rightarrow \id_{\rho_{\sigma,i}} \otimes f_{ij} \otimes \id_{\iota_{\tau,j}} = \id_{\rho_{\sigma,i}} \otimes g_{ij} \otimes \id_{\iota_{\tau,j}} ~~~ \forall i,j
\\
& \Rightarrow 
(\kappa_{\sigma}^{\dagger} \otimes \id \otimes \kappa_{\tau}^{\dagger}) 
\circ 
(\id \otimes f_{ij} \otimes  \id)
\circ
(\kappa_{\sigma} \otimes \id \otimes \kappa_{\tau}) \\
& \qquad \qquad \qquad \qquad =
(\kappa_{\sigma}^{\dagger} \otimes \id \otimes \kappa_{\tau}^{\dagger}) 
\circ 
(\id \otimes g_{ij} \otimes  \id)
\circ
(\kappa_{\sigma} \otimes \id \otimes \kappa_{\tau}) ~~~\forall i,j
\\
& \Leftrightarrow f_{ij} = g_{ij} ~~~ \forall i,j
\end{align*} 
\end{itemize}
\end{proof}

\begin{remark}[Matrix notation for standard duals and traces]\label{rem:matrixstandardduals}
Let $X: r \to s$ be any 1-morphism in $\mathcal{C}$. By Proposition~\ref{prop:matcequiv}, for any objects $r,s$ of $\mathcal{C}$ there exist adjoint equivalences $[E_r,E_r^*,\alpha_r,\beta_r]: r \to \Phi([\sigma_1,\dots,\sigma_{n_{1}}])$ and $[E_s,E_s^*,\alpha_s,\beta_s]: s \to \Phi([\tau_1,\dots,\tau_{n_2}])$, and unitary isomorphisms $u: E_{r}^* \otimes X \otimes E_s \to \Phi(M)$ and $\bar{u}: E_{s}^* \otimes X^* \otimes E_r \to \Phi(M^*)$ for some morphisms $M: [\sigma_1, \dots, \sigma_{n_1}] \to{} [\tau_1,\dots, \tau_{n_2}]$ and $M^*: [\tau_1,\dots,\tau_{n_2}] \to{} [\sigma_1,\dots, \sigma_{n_1}]$ in $\Mat(\mathcal{C})$. 

In the following diagrams we draw the wires for $E_r,E_r^*,E_s,E_s^*$ in blue, and draw $\alpha_r,\alpha_s$ as cups and $\beta_r,\beta_s$ as caps. We shade the functorial boxes corresponding to the equivalence $\Phi: \Mat(\mathcal{C}) \to \mathcal{C}$ in a lighter blue. Using fullness and faithfulness of $\Phi$ we define 2-morphisms $\tilde{\eta}: \id_{\vec{\tau}} \to M^* \otimes M$ and $\tilde{\epsilon}: M \otimes M^* \to \id_{\vec{\sigma}}$ in $\Mat(\mathcal{C})$ as follows:
\begin{align}\nonumber
\includegraphics[scale=1]{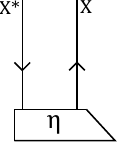} 
~~:=~~ 
\includegraphics[scale=1]{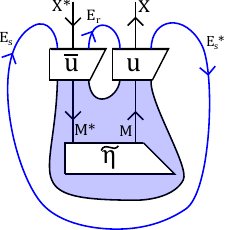}
\\\label{eq:standardcupequiv}
\includegraphics[scale=1]{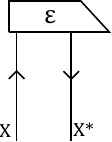}
~~:=~~
\includegraphics[scale=1]{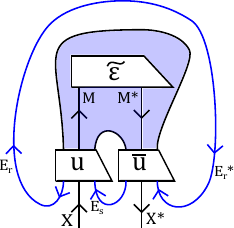}
\end{align} 
It is easy to check that $[M^*,\tilde{\eta},\tilde{\epsilon}]$ is a right dual for $M$ in $\Mat(\mathcal{C})$. By the discussion following Definition~\ref{def:matc}, there exist standard right duals $[(M_{ij})^*,\tilde{\eta}_{ij},\tilde{\epsilon}_{ij}]$ for each of the 1-morphisms $M_{ij}$, so that $M^*$ is defined as in~\eqref{eq:dualmatrix} and $\tilde{\eta},\tilde{\epsilon}$ as in~\eqref{eq:matrixdualcupcap}.

We will now consider the trace $\psi_X = \phi_X: \End(X) \to \mathbb{C}$. Let $T \in \End(X)$. In the notation just defined, we have the following expression for $\eta^{\dagger} \circ (\id_{X^*} \otimes T) \circ \eta$:
\begin{calign}\label{eq:standardtrace1}
\includegraphics[scale=1]{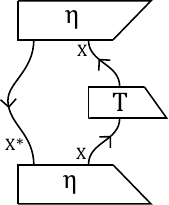}
~~=~~
\includegraphics[scale=1]{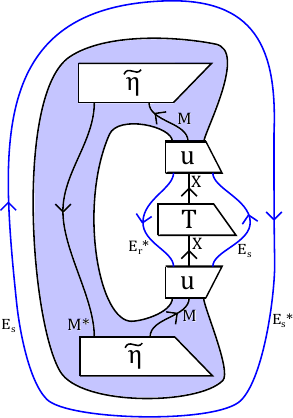}
~~=~~
\includegraphics[scale=1]{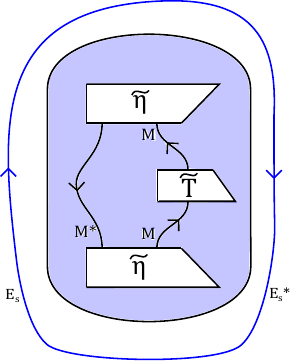}
\end{calign}
For the first equality we used~\eqref{eq:standardcupequiv} and unitarity of $\bar{u}$, $\beta_r$ and $\beta_s$. For the second equality we used fullness and faithfulness of the equivalence $\Phi$ and unitarity of the multiplicator $\mu_{M^*,M}$; here $\tilde{T}: M \to M$ is a uniquely defined 2-morphism in $\Mat(\mathcal{C})$.

It is straightforward to check that $x \mapsto \beta_s \circ (\id_{E_s} \otimes x \otimes \id_{E_s^*}) \circ \beta_s^{\dagger}$ defines a $*$-isomorphism $\nu_1: \End(\Phi(\vec{\tau})) \to \End(s)$, and likewise $x \mapsto \upsilon_{\vec{\tau}}^{\dagger} \circ x \circ \upsilon_{\vec{\tau}}$ defines a $*$-isomorphism $\nu_2: \End(\vec{\tau}) \to \End(\Phi(\vec{\tau}))$ (recall $\upsilon_{\vec{\tau}}$ is our notation for the unitor of the 2-functor $\Phi$). Since these are $*$-isomorphisms of commutative f.d.~$C^*$-algebras they must preserve the canonical trace in the sense that $\Tr(\nu_i(x)) = \Tr(x)$. In the rightmost diagram of~\eqref{eq:standardtrace1}  we see $(\nu_1 \circ \nu_2)(\eta^{\dagger} \circ (\id_{M^*} \otimes \tilde{T}) \circ \eta)$ and so it follows that:
$$
\phi_X(T) = \Tr_{\vec{\tau}}[\tilde{\eta}^{\dagger} \circ (\id_{M^*} \otimes \tilde{T}) \circ \tilde{\eta}]
$$
But it is straightforward to calculate in $\Mat(\mathcal{C})$ that:
$$\tilde{\eta}^{\dagger} \circ (\id_{M^*} \otimes \tilde{T}) \circ \tilde{\eta}
=
\diag([\sum_i (\tilde{\eta}_{i1})^{\dagger} \circ (\id_{(M_{i1})^*} \otimes \tilde{T}_{i1}) \circ \tilde{\eta}_{i1},~ \dots ~, \sum_i (\tilde{\eta}_{in_{\tau}})^{\dagger} \circ (\id_{(M_{in_{\tau}})^*} \otimes \tilde{T}_{in_{\tau}}) \circ \tilde{\eta}_{in_{\tau}}])$$
Using this together with a similar argument for $\psi_X$, we see that:
\begin{equation}\phi_X(T) = \sum_{i,j} \tilde{\eta}_{ij}^{\dagger} \circ (\id_{(M_{ij})^*} \otimes \tilde{T}_{ij}) \circ \tilde{\eta}_{ij} = \sum_{i,j} \tilde{\epsilon}_{ij} \circ (\tilde{T}_{ij} \otimes \id_{(M_{ij})^*}) \circ \tilde{\epsilon}_{ij}^{\dagger} = 
\psi_X(T) 
\end{equation}
\end{remark}

\subsection{Proof of Theorem~\ref{thm:eilenbergwatts}}
\label{app:eilenbergwatts}
\begin{theorem*}
The 2-functor $\Psi$ is an equivalence.
\end{theorem*}
\begin{proof}
We prove the result now.
\begin{itemize}
\item \emph{Essentially surjective on objects.} We use~\cite[Thm. A.1]{Neshveyev2018}, which shows that for any nonzero cofinite semisimple indecomposable left\footnote{Strictly speaking the cited theorem proves this for \emph{right} $\mathcal{T}$-module categories, but this is just a matter of convention. Indeed, a right $\mathcal{T}$-module category is just a left module category over the category $\mathcal{T}^{\otimes\op}$ with opposite tensor product. But $\dagger \circ *: \mathcal{T} \to \mathcal{T}^{\otimes\op}$ is an equivalence, where $\dagger$ is the dagger functor and $*$ is the right duals functor; this induces an equivalence $E:\Mod(\mathcal{T}) \overset{\sim}{\to} \Mod(\mathcal{T}^{\otimes\op})$. Then it suffices to observe that there is an equivalence of right $\mathcal{T}$-module categories $E(\Mod$-$A) \simeq A$-$\Mod$ (which takes a right $A$-module to its dual left $A$-module and a right $A$-module morphism to its conjugate).} $\mathcal{T}$-module category $\mathcal{M}$ there exists an SSFA $A$ in $\mathcal{T}$ such that $\mathcal{M}$ is equivalent to $\Mod$-$A$. Since direct sums of objects are preserved by linear 2-functors, essential surjectivity follows.
\item \emph{Essentially surjective on 1-morphisms.} We need to show that for any $\F$s $A, B$, the local functor $\Psi_{A,B}: A$-$\Mod$-$B \to \Hom_{\mathcal{T}}(\Mod$-$A, \Mod$-$B)$ is essentially surjective. In other words, for any module functor $F: \Mod$-$A \to \Mod$-$B$ there exists some dagger bimodule ${}_A M_B$ such that $F$ is unitarily isomorphic to $\Psi({}_A M_B)$. 

We first consider the special case where $A=B=\mathbbm{1}$. Here the left $\mathcal{T}$-module action on $\Mod$-$\mathbbm{1} = \mathcal{T}$ is given by tensor product on the left, and the functor $\Psi_{\mathbbm{1},\mathbbm{1}}: \mathcal{T} \to \End_{\mathcal{T}}(\mathcal{T})$ is given by tensor product on the right. We will show that $\Psi_{\mathbbm{1},\mathbbm{1}}: \mathcal{T} \to \End_{\mathcal{T}}(\mathcal{T})$ is an equivalence, implying essential surjectivity in this case. For this, observe that $\mathcal{T}$ is an invertible $\mathcal{T}$-$\mathcal{T}$ bimodule category in the sense of~\cite[Def. 3.1]{Neshveyev2018} (consider the two-object rigid $C^*$-2-category $\mathcal{C}$ with the set $\{0,1\}$ of objects, where $\mathcal{C}_{00} = \mathcal{C}_{01} = \mathcal{C}_{10} = \mathcal{C}_{11} = \mathcal{T}$, composition is by tensor product keeping track of the objects, and the dual functor takes a 1-morphism in $\mathcal{C}_{01}$ to its dual in $\mathcal{C}_{10}$ and vice versa). By~\cite[Thm. 3.2 (c)]{Neshveyev2018}, it follows that the functor $\Psi_{\mathbbm{1},\mathbbm{1}}$ is an equivalence.\footnote{Strictly speaking, this cited theorem also uses the opposite convention and shows that the functor from $\mathcal{T}$ acting on the \emph{left} to the endomorphism category of $\mathcal{T}$ as a \emph{right} $\mathcal{T}$-module category is an equivalence. To get round this, run the argument for $\mathcal{T}^{\otimes \op}$ rather than $\mathcal{T}$.}

We will now show that this is enough to imply essential surjectivity for the other $\Hom$-categories. Indeed, let $F: \Mod-A \to \Mod-B$ be a module functor. Now $\Psi({}_{\mathbbm{1}}A_A) \otimes F \otimes \Psi({}_B B_{\mathbbm{1}}) \in \End_{\mathcal{T}}(\mathcal{T})$; therefore, by the special case just proven, there exists some object $O_F$ of $\mathcal{T}$ and a unitary isomorphism $U: \Psi({}_{\mathbbm{1}}A_A) \otimes F \otimes \Psi({}_BB_{\mathbbm{1}}) \overset{\sim}{\to} \Psi(O_F)$. 

In what follows we use light blue shading for the functorial boxes of the 2-functor $\Psi$, and we label regions corresponding to objects of $\Mod(\mathcal{T})$ with the name of the corresponding object. We first observe that $O_F$ naturally has the structure of an $A$-$B$ dagger bimodule. Indeed, we define the following morphism $\Psi(A \otimes O_F \otimes B) \to \Psi(O_F)$ in $\End_{\mathcal{T}}(\mathcal{T})$:
\begin{calign}\label{eq:alphapushforward}
\includegraphics[scale=1]{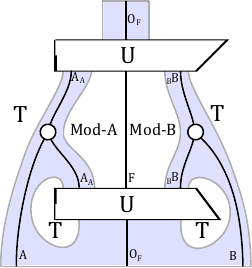}
\end{calign}
Using fullness and faithfulness of the equivalence $\Psi_{\mathbbm{1},\mathbbm{1}}: \mathcal{T} \to \End_{\mathcal{T}}(\mathcal{T})$, we can pull this back uniquely to obtain a preimage $\alpha: A \otimes O_F \otimes B \to O_F$ in $\mathcal{T}$. Using faithfulness of $\Psi_{\mathbbm{1},\mathbbm{1}}$ we now show that $\alpha$ is a dagger bimodule action. For the first condition of~\eqref{eq:module}:
\begin{calign}
\includegraphics[scale=1]{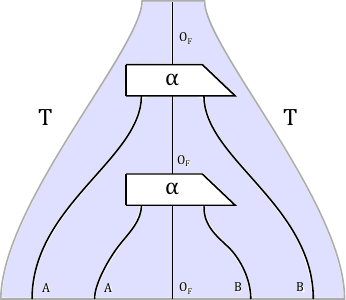}
~~=~~
\includegraphics[scale=1]{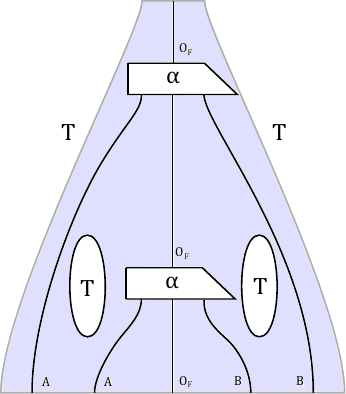}
\\
=~~
\includegraphics[scale=1]{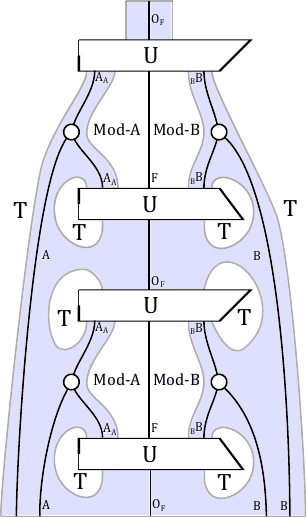}
~~=~~
\includegraphics[scale=1]{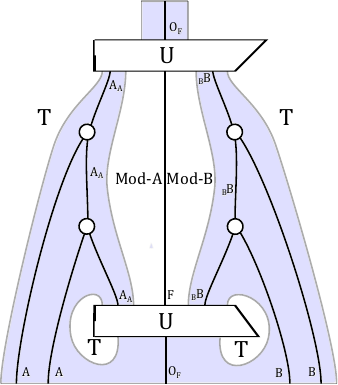}
\\
\includegraphics[scale=1]{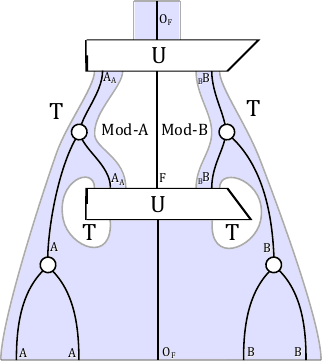}
~~=~~
\includegraphics[scale=1]{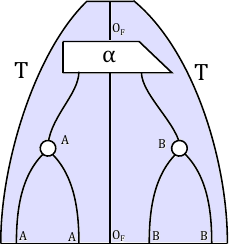}
\end{calign}
Here for the first equality we used unitarity of the unitor $\upsilon_{\mathbbm{1}}$ of the 2-functor $\Psi$; for the second equality we used the definition of $\alpha$ as the pullback of the morphism~\eqref{eq:alphapushforward}; for the third equality we used unitarity of the unitor $\upsilon_{\mathbbm{1}}$ and unitarity of $U$; for the fourth equality we used associativity of the $\F$s $A$ and $B$; and for the final equality we used the definition of $\alpha$ again. By faithfulness of $\Psi_{\mathbbm{1},\mathbbm{1}}$ we can remove the functorial boxes and so we obtain the desired equality.

For the second condition of~\eqref{eq:module}:
\begin{calign}
\includegraphics[scale=1]{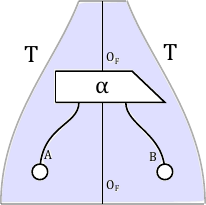}
~~=~~
\includegraphics[scale=1]{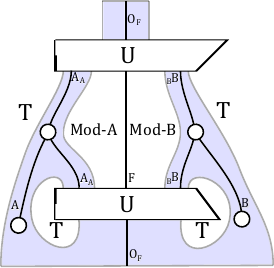}
\\=~~
\includegraphics[scale=1]{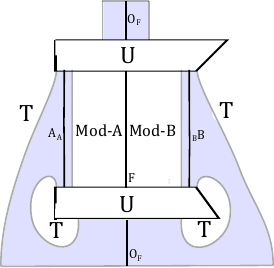}
~~=~~
\includegraphics[scale=1]{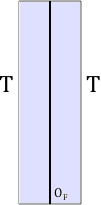}
\end{calign}
Here the first equality is by definition of $\alpha$; the second equality is by unitality of the Frobenius algebras $A$ and $B$; and the final equality is by unitarity of $\upsilon_{\mathbbm{1}}$, manipulation of functorial boxes, and unitarity of $U$. Again, by faithfulness of $\Psi_{\mathbbm{1},\mathbbm{1}}$ we can remove the functorial boxes and so obtain the desired equality.

For the third condition of~\eqref{eq:module}:
\begin{calign}
\includegraphics[scale=1]{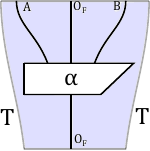}
~~=~~
\includegraphics[scale=1]{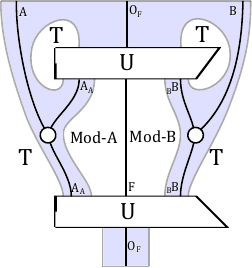}
\\=~~
\includegraphics[scale=1]{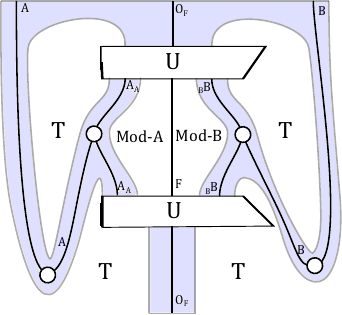}
~~=~~
\includegraphics[scale=1]{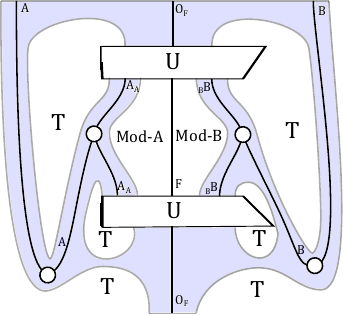}
\\=~~
\includegraphics[scale=1]{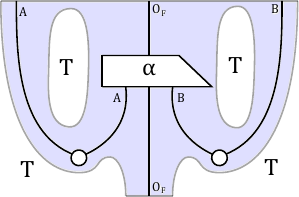}
~~=~~
\includegraphics[scale=1]{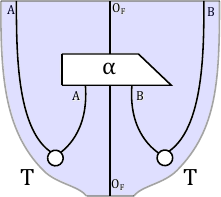}
\end{calign}
Here the first equality is by definition of $\alpha$ and unitarity of $\Psi$; the second equality is by the Frobenius equation for $A$ and $B$ and manipulation of functorial boxes; the third equality is by unitarity of $\upsilon_{\mathbbm{1}}$; the fourth equality is by definition of $\alpha$; and the final equality is by unitarity of $\upsilon_{\mathbbm{1}}$.
Again, by faithfulness of $\Psi_{\mathbbm{1},\mathbbm{1}}$ we can remove the functorial boxes and so obtain the desired equality.

Having defined an $A$-$B$ dagger bimodule structure on $O_F$, we now claim that $F$ is unitarily isomorphic to $\Psi({}_A(O_F)_B)$. To prove this, we will define a dagger idempotent in $\Hom_{\mathcal{T}}(\Mod$-$A,\Mod$-$B)$ for which the objects $F$ and $\Psi({}_A(O_F)_B)$ are both valid splittings. Since the splitting of a dagger idempotent is unique up to a unitary isomorphism, the result follows. 

The dagger idempotent will be an endomorphism of $\Psi({}_A A \otimes O_F \otimes B_B)$. We make some preliminary remarks. First, it is obvious that the comultiplication morphisms of the $\F$s $A$ and $B$ are $A$-$A$ and $B$-$B$ bimodule homomorphisms. In $\Bimod(\mathcal{T})$ we depict these as follows (recall that, as the unit objects in $A$-$\Mod$-$A$ and $B$-$\Mod$-$B$, the bimodules ${}_A A_A$ and ${}_B B_B$ are not depicted in the diagrammatic calculus of $\Bimod(\mathcal{T})$):
\begin{calign}
&\includegraphics[scale=1]{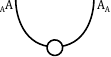}
&&
&\includegraphics[scale=1]{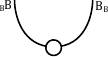}
\\
&(m_A)^{\dagger}: {}_A A_A \to {}_A A \otimes A_A
&&
&(m_B^{\dagger}): {}_B B_B \to {}_B B \otimes B_B
\end{calign}
To avoid any confusion we remark that these are not the same as the cups of~\eqref{eq:acupcap}. By separability of $A$ and $B$, these 2-morphisms in $\Bimod(\mathcal{T})$ are isometries.

We now define the following 2-morphism $\iota_F: F \to \Psi({}_A A \otimes O_F \otimes B_B)$:
\begin{calign}
\includegraphics[scale=1]{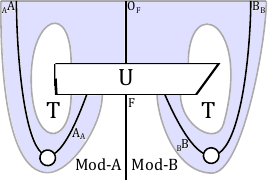}
\end{calign}
Clearly $\iota_F$ is an isometry:
\begin{calign}
\includegraphics[scale=1]{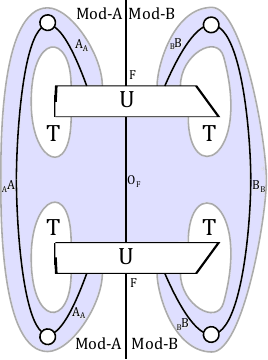}
~~=~~
\includegraphics[scale=1]{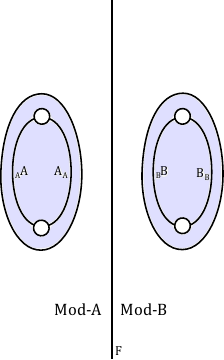}
~~=~~
\includegraphics[scale=1]{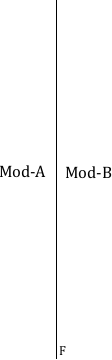}
\end{calign}
Here the first equality is by unitarity of $\upsilon_{\mathbbm{1}}$ and of $U$, and the second equality is by separability of $A$ and $B$ and unitarity of $\upsilon_{A}$ and $\upsilon_{B}$. 

It follows that $\pi:= \iota_F \circ \iota_F^{\dagger} \in \End(\Psi({}_A A \otimes O_F \otimes B_B))$ is a dagger idempotent which is split by $F$. 

To show that $\Psi({}_A(O_F)_B)$ also splits $\pi$ we need to define an isometry $\iota: \Psi({}_A(O_F)_B) \to \Psi({}_A A \otimes O_F \otimes B_B)$ such that $\iota \circ \iota^{\dagger} = \pi$. Consider the action $\alpha: A \otimes O_F \otimes B \to O_F$ defining the $A$-$B$ bimodule structure on ${}_A (O_F)_B$. By the first condition of~\eqref{eq:module}, $\alpha$ is a bimodule homomorphism ${}_A A \otimes O_F \otimes B_B \to {}_A (O_F)_B$. We therefore define $\iota$ as follows:
\begin{calign}
\includegraphics[scale=1]{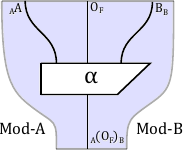}
\end{calign}
To see that $\iota$ is an isometry it suffices to show that $\alpha \circ \alpha^{\dagger} = \id_{O_F}$. But this follows straightforwardly from the dagger bimodule equations and separability of $A$ and $B$:
\begin{calign}
\includegraphics[scale=1]{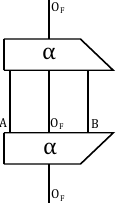}
~~=~~
\includegraphics[scale=1]{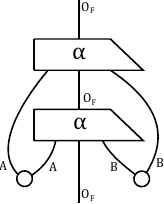}
~~=~~
\includegraphics[scale=1]{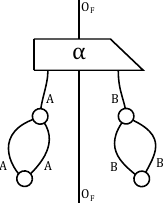}
~~=~~
\includegraphics[scale=1]{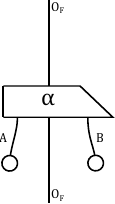}
~~=~~
\includegraphics[scale=1]{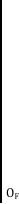}
\end{calign}
Here the first equality is by the third equation of~\eqref{eq:module}; the second equality is by the first equation of~\eqref{eq:module}; the third equality is separability of $A$ and $B$; and the final equality is by the second equation of~\eqref{eq:module}.

To finish we need to show that $\iota \circ \iota^{\dagger} = \pi$. We first observe the following equation for $\alpha^{\dagger} \circ \alpha$:
\begin{calign}\label{eq:alphadaggeralpha}
\includegraphics[scale=1]{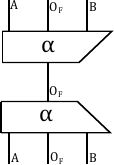}
~~=~~
\includegraphics[scale=1]{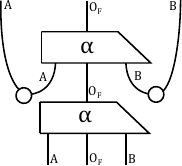}
~~=~~
\includegraphics[scale=1]{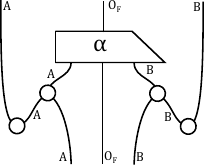}
~~=~~
\includegraphics[scale=1]{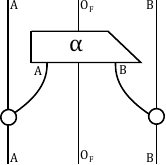}
\end{calign}
Here the first equality is by the third equation of~\eqref{eq:module}; the second equality is by the first equation of~\eqref{eq:module}; and the final equality is by the Frobenius equation~\eqref{eq:Frobenius}.

We also observe that by the Frobenius equation we have the following equation in $\Bimod(\mathcal{T})$ (the analogous equation for $B$ also holds):
\begin{calign}\label{eq:frobinbimod}
\includegraphics[scale=1]{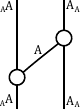}
~~=~~
\includegraphics[scale=1]{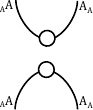}
~~=~~
\includegraphics[scale=1]{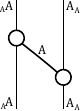}
\end{calign}
We then have the following equation for $\iota \circ \iota^{\dagger}$:
\begin{calign}\label{eq:iotaiotadagger}
\includegraphics[scale=1]{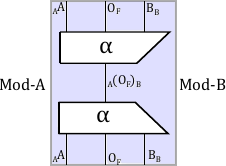}
~~=~~
\includegraphics[scale=1]{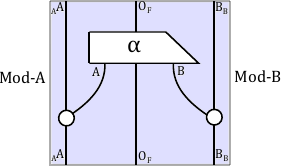}
~~=~~
\includegraphics[scale=1]{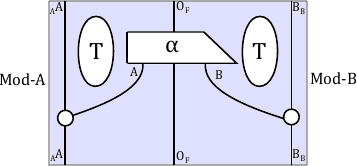}
\\=~~
\includegraphics[scale=1]{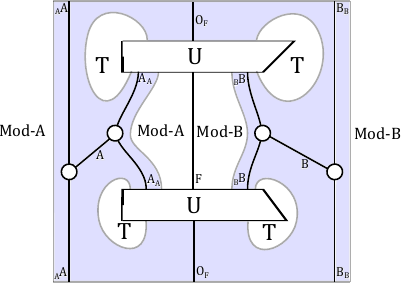}
~~=~~
\includegraphics[scale=1]{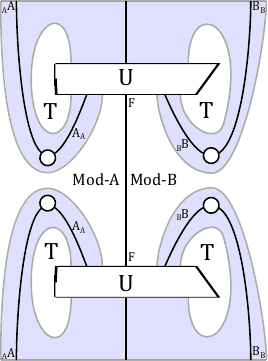}
\end{calign}
Here the first equality is by~\eqref{eq:alphadaggeralpha}; the second equality is by unitarity of $\upsilon_{\mathbbm{1}}$; the third equality is by definition of $\alpha$; and the final equality is by~\eqref{eq:frobinbimod}. In the last diagram of~\eqref{eq:iotaiotadagger} we indeed see $\pi = \iota_f \circ \iota_f^{\dagger}$ and so essential surjectivity on 1-morphisms is proven.
\item \emph{Faithful on 2-morphisms}. Let $f,g: {}_A M_B \to {}_A N_B$ be bimodule homomorphisms. Faithfulness is the statement that $\Psi(f) = \Psi(g) \Rightarrow f = g$. Now suppose $\Psi(f)=\Psi(g)$ and consider the component of this natural transformation on the object $A_{A}$ of $\Mod$-$A$ as a morphism in $\mathcal{T}$:
\begin{calign}
\includegraphics[scale=1]{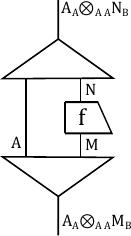}
~~=~~
\includegraphics[scale=1]{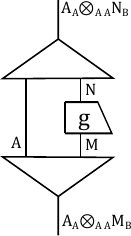}
&&\Rightarrow&&
\includegraphics[scale=1]{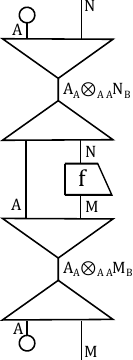}
~~=~~
\includegraphics[scale=1]{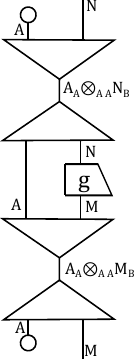}
\\
\Leftrightarrow~~
\includegraphics[scale=1]{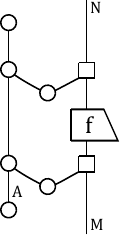}
~~=~~
\includegraphics[scale=1]{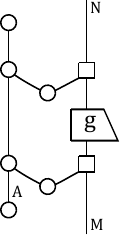}
&&\Leftrightarrow&&
\includegraphics[scale=1]{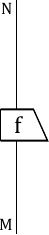}
~~=~~
\includegraphics[scale=1]{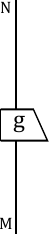}
\end{calign}
Here the first implication is by pre- and postcomposition on both sides of the equality; the second implication is by definition of the idempotent~\eqref{eq:idempotentforrelprod}; and the final implication is by unitality, counitality and separability of $A$, the Frobenius equation for $A$, the fact that $f$ and $g$ are bimodule homomorphisms, and~\eqref{eq:module}.
\item \emph{Full on 2-morphisms.} We must show that for any morphism of module functors $f: \Psi({}_A M_B) \to \Psi({}_A N_B)$ there exists a bimodule morphism $\phi: {}_A M_B \to {}_A N_B$ such that $\Psi(\phi) = f$. This is seen as follows:
\begin{calign}
\includegraphics[scale=1]{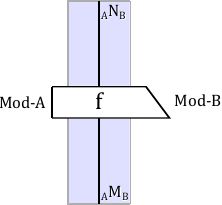}
~~=~~
\includegraphics[scale=1]{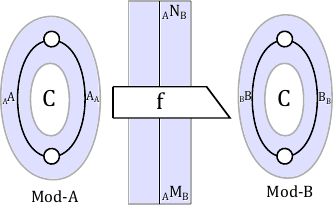}
\\=~~
\includegraphics[scale=1]{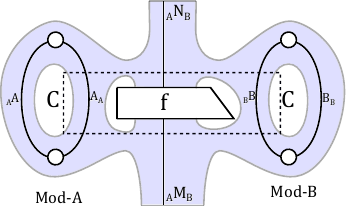}
~~=~~
\includegraphics[scale=1]{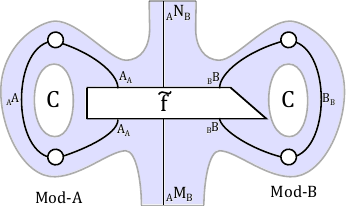}
\\=~~
\includegraphics[scale=1]{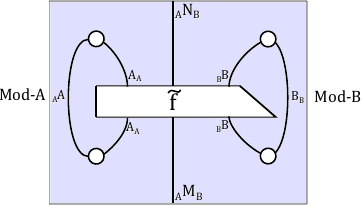}
\end{calign}
Here the first equality is by unitarity of $\upsilon_{A}$, $\upsilon_{B}$ and $\upsilon_{\mathbbm{1}}$, and separability of $A$ and $B$; the second equality is by unitarity of $\upsilon_{A}$ and $\upsilon_{B}$; the third equality is by fullness and faithfulness of the equivalence $\Psi_{\mathbbm{1},\mathbbm{1}}$ (here $\tilde{f}: A \otimes_A M \otimes_B  B \to A \otimes_A  N \otimes_B  B$ is the unique preimage of the morphism $\Psi(A \otimes_A M \otimes_B B) \to \Psi(A \otimes_A  N \otimes_B B)$ contained in the dashed box); and the final equality is by unitarity of $\upsilon_{\mathbbm{1}}$ and manipulation of functorial boxes.
\end{itemize}
\end{proof}

\end{document}